\documentclass[twocolumn]{svjour3}

\usepackage{url}

\usepackage{booktabs} % For formal tables
\usepackage{graphicx}
\usepackage{here}
\usepackage[ruled, linesnumbered, vlined]{algorithm2e}
\usepackage{subfigure}
\usepackage{calc}
\usepackage{tabularx}
\usepackage{hyperref}
\PassOptionsToPackage{hyphens}{url}
\usepackage{balance}
\usepackage{mathtools}
\usepackage{lipsum,multicol}
\usepackage{soul}
\usepackage{amssymb}

\usepackage{amsmath}

\mathchardef\mhyphen="2D % Define a "math hyphen"

\usepackage{xcolor}

\newenvironment*{revise-env}{\color{red}}{}

\usepackage{array,multirow}
\usepackage{float}

\usepackage{listings}
\makeatletter
\def\hlinewd#1{%
\noalign{\ifnum0=`}\fi\hrule \@height #1 %
\futurelet\reserved@a\@xhline}
\makeatother

\makeatletter
\newcommand{\removelatexerror}{\let\@latex@error\@gobble}
\makeatother
\let\oldnl\nl% Store \nl in \oldnl
\newcommand{\nonl}{\renewcommand{\nl}{\let\nl\oldnl}}%

\usepackage{cite}

\begin{document}
% Copyright
% \setcopyright{none}
% \setcopyright{acmcopyright}
% \setcopyright{acmlicensed}
% \setcopyright{rightsretained}
% \setcopyright{usgov}
% \setcopyright{usgovmixed}
% \setcopyright{cagov}
% \setcopyright{cagovmixed}

%% The following content must be adapted for the final version
% paper-specific
% \newcommand\vldbdoi{XX.XX/XXX.XX}
% \newcommand\vldbpages{XXX-XXX}
% % issue-specific
% \newcommand\vldbvolume{14}
% \newcommand\vldbissue{1}
% \newcommand\vldbyear{2020}
% % should be fine as it is
% \newcommand\vldbauthors{\authors}
% \newcommand\vldbtitle{\shorttitle} 

% % DOI
% \acmDOI{XXXXXX}

% % ISBN
% \acmISBN{123-4567-24-567/08/06}

% % Conference
% \acmConference[SIGMOD '21]{ACM SIGMOD conference}{June 20--25, 2021}{Xi'an, Shaanxi, China}
% \acmBooktitle{2021 International Conference on Management of Data (SIGMOD '21), June 20--25, 2021, Xi'an, Shaanxi, China}
% \acmYear{2021}
% \copyrightyear{2021}
% \acmPrice{15.00}
% \acmSubmissionID{123-A12-B3}

% \vldbTitle{Dynamic Scaling of Graph Partitions: A preprocessing Approach}
% \vldbTitle{Elastic Graph Partitioning based on Graph Edge Ordering for Dynamic Graph on Cloud} 
% \vldbAuthors{Masatoshi Hanai, Toyotaro Suzumura, Wentong Cai, and Georgios Theodoropoulos}
% \vldbDOI{https://doi.org/TBD}
% \vldbVolume{xx}
% \vldbNumber{xxx}
% \vldbYear{2020}

% \title{Graph Dynamic Scaling Based on Preprocessing}
\title{Time-Efficient and High-Quality Graph Partitioning for Graph Dynamic Scaling}

\author{
  Masatoshi Hanai \and
  Nikos Tziritas \and
  Toyotaro Suzumura \and
  Wentong Cai \and
  Georgios Theodoropoulos
%   \IEEEauthorrefmark{5}\thanks{\IEEEauthorrefmark{5}Prof. Georgios Theodoropoulos is a corresponding author.},
% \IEEEauthorblockA{
%   \IEEEauthorrefmark{1}\textit{Southern University of Science and Technology, Shenzhen, China} \\
%   \IEEEauthorrefmark{2}\textit{University of Thessaly, Lamia, Greece} \\
%   \IEEEauthorrefmark{3}\textit{IBM T.J. Watson Research Center, New York, USA} \\
%   \IEEEauthorrefmark{4}\textit{Nanyang Technological University, Singapore}
%   %   Air Force Office of Scientific Research, Arlington, Virginia, USA}
%   }
% \IEEEauthorblockA{mhanai@acm.org, nitzirit@uth.gr, suzumura@acm.org, aswtcai@ntu.edu.sg, georgios@sustech.edu.cn,}
}

\institute{M. Hanai and G. Theodoropoulos\at
              Southern University of Science and Technology, China \\
            %   Tel.: +123-45-678910\\
            %   Fax: +123-45-678910\\
              \email{mhanai@acm.org, georgios@sustech.edu.cn}           %  \\
          \and
          N. Tziritas \at
              University of Thessaly, Laria, Greece \\
              \email{nitzirit@uth.gr}
          \and
          T. Suzumura \at
             IBM T.J. Watson Research Center, New York, USA \\
             \email{suzumura@acm.org}
          \and
          W. Cai \at 
            Nanyang Technological University, Singapore \\
            \email{aswtcai@ntu.edu.sg}
}

\date{Received: date / Accepted: date}

\maketitle

\begin{abstract}
% for Journal
% We address a fundamental issue of the \emph{k}-way graph partitioning, where the number of partitions, $k$, has to be fixed \emph{before} computing the time-consuming graph partitioning algorithm.
% The inherent feature of the \emph{k}-way graph partitioning compels us to make full use of the recent elastic computational resources such as the cloud.
% For example, every time we change the number of virtual machines on demand, we need to repeat the time-consuming computations.

% In this paper, we propose a new preprocessing-based graph partitioning, called \emph{graph edge ordering}, which enables us to fix $k$ \emph{after} the time-consuming computation.
% Our idea is that the graph data are pre-processed into an ordered edge list so that edges with high locality are closer to each other.
% Once the time-consuming preprocessing is done, the ordered edge list can be immediately divided into arbitrary $k$ partitions.
% Our comprehensive evaluation with three important case studies, including dynamic scaling, architecture-aware graph partitioning, and heterogeneous graph partitioning, demonstrates that our approach flexibly generates high-quality partitions and improves the performance of distributed graph analysis.

The dynamic scaling of distributed computations plays an important role in the utilization of elastic computational resources, such as the cloud. 
It enables the provisioning and de-provisioning of resources to match dynamic resource availability and demands.
% It enables us to flexibly manage the number of virtual machines for various requirements, e.g., the time limits and financial constraints.
In the case of distributed graph processing, changing the number of the graph partitions while maintaining high partitioning quality imposes serious computational overheads as typically a time-consuming graph partitioning algorithm needs to execute each time repartitioning is required.
% In the distributed graph processing, to change the number of graph partitions while obtaining the high partitioning quality causes serious computational overheads in the traditional way because it needs to execute a high-quality but time-consuming graph partitioning algorithm every time the number of partitions is changed.

In this paper, we propose a dynamic scaling method that can efficiently change the number of graph partitions while keeping its quality high.
Our idea is based on two techniques: preprocessing and very fast edge partitioning, called \emph{graph edge ordering} and \emph{chunk-based edge partitioning}, respectively.
% \footnote{Available in  \url{https://github.com/masatoshihanai/GraphEdgeOrdering}}.
The former converts the graph data into an ordered edge list in such a way that edges with high locality are closer to each other.
The latter immediately divides the ordered edge list into an arbitrary number of high-quality partitions.
The evaluation with the real-world billion-scale graphs demonstrates that our proposed approach significantly reduces the repartitioning time, while the partitioning quality it achieves is on par with that of the best existing static method.
% Our evaluation compared to the existing methods demonstrates that our method significantly reduces the partitioning time, which is three-to-eight orders of magnitude faster than the others, while achieving the high-partitioning quality, which is almost the same quality as the state-of-the-art methods.
% , resulting in the speed-up of the distributed graph applications, such as, shortest path and PageRank.
\end{abstract}

\section{Introduction}\label{sec:introduction}
Graph analysis is a powerful method to gain valuable insights into the characteristics of real networks, such as web graphs and social networks.
To analyze large-scale graphs efficiently, one of the major approaches is 
to distribute the entire graph across multiple machines and process each partition in parallel. 
% to make the entire graph be distributed across multiple machines and process in parallel.
Over the last decade, several distributed graph-processing systems have been developed~\cite{malewicz2010pregel,gonzalez2014graphx,joseph2012powergraph,hong2015pgx,Chen:2015:PDG:2741948.2741970}
% ~\cite{malewicz2010pregel,giraph,gonzalez2014graphx,low2012distributed,joseph2012powergraph,roy2013x,khayyat2013mizan,salihoglu2014optimizing,suzumura2015scalegraph,yan2015effective,hong2015pgx,wu2015g,Chen:2015:PDG:2741948.2741970,zhu2016gemini,Ahmad:2018:LSL:3204028.3228395}.
% Over the last decade, several distributed graph-processing systems have been developed such as Pregel, Giraph, GraphX, GPS, PowerGraph, Pregel+, PGX.D, X-Stream, ScaleGraph, GraM, PowerLyra, Gemini, and LA3~\cite{malewicz2010pregel,giraph,gonzalez2014graphx,low2012distributed,joseph2012powergraph,roy2013x,khayyat2013mizan,salihoglu2014optimizing,suzumura2015scalegraph,yan2015effective,hong2015pgx,wu2015g,Chen:2015:PDG:2741948.2741970,zhu2016gemini,Ahmad:2018:LSL:3204028.3228395}.

% \begin{figure}[t]
%   \centering
%     \subfigure[Traditional Way]{\includegraphics[width=.85\columnwidth]{Flow.eps}\label{fig:flow}}
%     \subfigure[Proposed Approach based on Preprocessing]{\includegraphics[width=\columnwidth]{OrderingFlow.eps}\label{fig:orderflow}}%
%   \vspace{-10pt}
%   \caption{Iterative Analysis for Graph.}
%   \vspace{-10pt}
% \end{figure}

% \todo{graph partitioning}
For efficient parallel computation on a distributed graph-processing system, the common problem is to divide the input graph into $k$ parts in such a way that the number of edge/vertex cuts (i.e., communication cost among the distributed processes) becomes minimal while keeping each part balanced; this is known as the \emph{balanced $k$-way graph partitioning}.
Since the computation of the optimal-quality partitions, namely, partitions with the minimum cuts, is an NP-hard problem~\cite{garey1974some,Andreev:2004:BGP:1007912.1007931,Bourse:2014:BGE:2623330.2623660,Zhang:2017:GEP:3097983.3098033}, the high-quality graph partitioning algorithms, such as METIS~\cite{Karypis:1998:FHQ:305219.305248} and NE~\cite{Zhang:2017:GEP:3097983.3098033}, are generally time-consuming compared to the low-quality ones, such as FENNEL~\cite{Tsourakakis:2014:FSG:2556195.2556213}, DBH~\cite{NIPS2014_5396}, HDRF~\cite{Petroni:2015:HSP:2806416.2806424}.
There is a clear trade-off between partitioning efficiency and quality.

\begin{figure}
  \centering
   \includegraphics[width=\columnwidth]{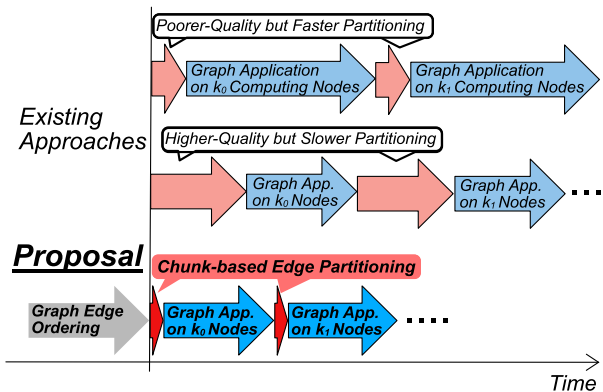}%
  \caption{Workload Example. \#~partitions is $k_0 \rightarrow k_1 \rightarrow ...$.}%
  \label{fig:workflow}
\end{figure}

In a related development, with the utilization of elastic infrastructures such as cloud platforms, \emph{dynamic
scaling} of computational resources has become increasingly important for parallel and distributed computation.
Especially, one of our motivated scenarios for the cloud is the effective utilization of \emph{unreliable VM instances} that do not have any lifetime guarantee, such as Spot Instances in AWS~\cite{spotinstance} and Preemptible VMs in GCE~\cite{preemptible}.
In such an unreliable VM environment, the price and availability of the VM is changed depending on spare resources of a data center.
The spare VM resources may be suddenly available while the executing VMs may be forcely terminated by the infrastructure side.
% Such a situation also happens in the large-scale cluster environments such as a supercomputer.
% The effective usage of spare resources is crucial to improve the overall utilization rate.
Dynamic scaling plays an important role to handle such a dynamic change of computational resources.

% In a related development, with the utilization of elastic infrastructures such as cloud platforms, \emph{dynamic
% scaling}~\cite{dynamicscaling,chieu2009dynamic,shen2011cloudscale,mao2011auto,Das:2013:EES:2445583.2445588,castro2013integrating,taft2014store,heinze2015online,adya2016slicer,floratou2017dhalion,madsen2017integrative,taft2018p,qiao2018litz,marcus2018nashdb,Borkowski:2019:MCR:3317315.3329476} of computational resources has become increasingly important for distributed computation.
% % On the other hand, \emph{dynamic
% % scaling}~\cite{dynamicscaling,chieu2009dynamic,shen2011cloudscale,mao2011auto,Das:2013:EES:2445583.2445588,castro2013integrating,taft2014store,heinze2015online,adya2016slicer,floratou2017dhalion,madsen2017integrative,taft2018p,qiao2018litz,marcus2018nashdb,Borkowski:2019:MCR:3317315.3329476} of the computational resources becomes more important in the parallel and distributed applications as the recent computational resources can be utilized in an elastic way, such as the cloud.
% Cloud computing allows the flexible control of the number of 
% % The cloud computing enables us to flexibly control the number of 
% virtual machines (VMs), depending on various requirements, e.g., 
% the size of data, the workload characteristics, the time limits, and the financial constraints.
% Moreover, in the utilization of \emph{non-lifetime-guaranteed VM instances}, such as Spot Instances in AWS~\cite{spotinstance} and Preemptible VMs in GCE~\cite{preemptible}, dynamic scaling is necessary for coping with the forced termination of VMs.

\begin{figure*}[t]
  \centering
   \includegraphics[width=2.0\columnwidth]{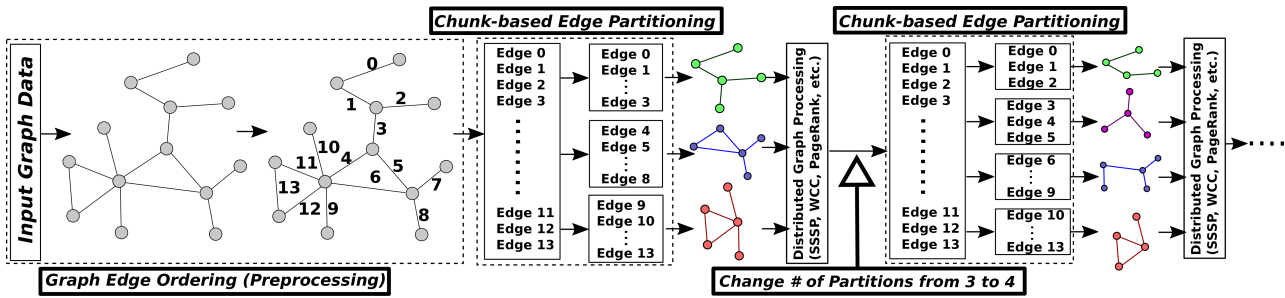}
  \caption{Dynamic Scaling based on Graph Edge Ordering and Chunk-based Edge Partitioning.}\label{fig:orderflow}
\end{figure*}

In the case of distributed graph analysis, however, scaling the number of graph partitions efficiently while achieving high quality is a challenging endeavor due to the trade-off between efficiency and quality.
On the one hand, several dynamic scaling methods based on efficient graph partitioning have been proposed~\cite{pujol2011little,vaquero2014adaptive,8798698,8514898,dynamicscaling}, which, however, exhibit limited quality. 
It results in high communication costs, affecting the performance of distributed graph processing, as shown in the top of Figure~\ref{fig:workflow}.
% On one hand, as traditionally investigated in~\cite{pujol2011little,vaquero2014adaptive,8798698,8514898,dynamicscaling}, the efficient graph partitioning methods may be used, but the quality is limited. 
% It results in poor performance of the graph applications imposed by the high communication cost (the first line of Figure~\ref{fig:workflow}).
On the other hand, conventional approaches based on high-quality graph partitioning, e.g., ~\cite{Karypis:1998:FHQ:305219.305248,Zhang:2017:GEP:3097983.3098033}, may cause redundant computation as typically time-consuming calculations are required each time the number of partitions is changed, as shown in the middle of Figure~\ref{fig:workflow}.
% On the other hand, the high-quality graph partitioning methods cause redundant computation as typically the time-consuming calculation is required each time the number of partitions is changed (the second line of Figure~\ref{fig:workflow}).

The problem of graph dynamic scaling is similar to the \emph{dynamic load balancing} for the distributed graph analysis~\cite{shang2013catch,khayyat2013mizan,xu2014loggp,huang2016leopard,zheng2016paragon,zheng2016planar} as both need to repartition a graph.
However, these methods focus on 
repartitioning the graph to reflect changes in the behavior of the application workload only and do not consider dynamic changes of the computational infrastructure.
Furthermore, in these cases, the number of graph partitions does not change. 
% the change of graph application workloads and does not address the change of graph partitions caused by the underlying infrastructure (e.g., unreliable VMs).
In this paper, we focus on a different problem, where the dynamic scaling of the graph is triggered by the dynamic scaling of the computational infrastructure rather than the application workload.
In such cases, the graph has to be repartitioned to make use of the newly available computational resources.

% However, the traditional high-quality $k$-way graph partitioning algorithms cause redundant computation when we apply it to the dynamic scaling, which dynamically changes the number of graph partitions, $k$ (the first line of Figure~\ref{fig:workflow}).
% For improving the time efficiency for the partitioning, a low-quality method may be used, e.g. as recently proposed in~\cite{dynamicscaling}; however 
% it results in the poor performance of the graph applications (the second line of Figure~\ref{fig:workflow}).
% The problem is slightly similar to the \emph{dynamic load balancing} for the distributed graph analysis~\cite{shang2013catch,xu2014loggp,huang2016leopard,zheng2016paragon,zheng2016planar}, but these methods do not address the dynamic change of $k$, i.e., $k$ is static.
% The graph cannot be partitioned \emph{before} fixing $k$ since the algorithm basically takes $k$ as an input parameter.
% For example in Figure~\ref{fig:flow}, we have to execute the partitioning algorithm every time $k$ is changed.
% As Bulu{\c{c}} et al. also have pointed out in the recent survey, ``\emph{the fixed value for k becomes questionable when we want to achieve malleable computations}~\cite{bulucc2016recent}''

% \smallskip
In this paper, we propose a novel approach to the dynamic scaling of graph partitions, which enables us to efficiently recompute the partitioning when the number of partitions is changed while keeping partitioning quality high.
As with the latest work~\cite{dynamicscaling}, we focus on (vertex-cut) edge partitioning rather than traditional (edge-cut) vertex partitioning, as discussed in the other existing work~\cite{pujol2011little,vaquero2014adaptive,8798698,8514898}. 
The edge partitioning is known
% because of two reasons.
% First, the edge partitioning is known 
to provide a better workload balance because the computational cost in the graph processing essentially depends on the number of edges rather than that of vertices~\cite{joseph2012powergraph,gonzalez2014graphx}.

% \revise{XXXXXX SAY STH HERE ?????}

% of edge processing than the vertex one~\cite{joseph2012powergraph,gonzalez2014graphx}
% , especially for the large-scale real-world graphs, which mostly have skewed-degree distribution, such as web graphs and social networks~\cite{joseph2012powergraph,gonzalez2014graphx}.
% Second, by effectively utilizing the advantages of the edge partitioning, our proposed method can be the fastest way theoretically and practically.
% as we will discuss in Section~\ref{sec:dynamic} and evaluate in Section~\ref{sec:evaluation}. 
% we address the issue of the redundant computation for the dynamic scaling in the distributed graph processing.
% In this paper, we address the fundamental issue of the redundant computation due to the fixed $k$ in the graph partitioning.
% Our approach is based on the preprocessing, called \emph{graph edge ordering}, which enable us to execute the partitioning algorithm as fast as possible while keeping the partitioning quality high.

The dynamic scaling method which we propose in this paper is based on two techniques: \emph{graph edge ordering} and \emph{chunk-based edge partitioning}.
Figure~\ref{fig:orderflow} shows an overview.
The graph edge ordering is a preprocessing method, which orders the edges of the input graph in such a way that edges with closer ids have a higher access locality (e.g., input edges are ordered to Edge 0,1,2,3... in Figure~\ref{fig:orderflow}).
Then, the chunk-based edge partitioning, which is a simple yet very fast partitioning method, splits the ordered edge lists.
% The essential but time-consuming computation is extracted into the graph edge ordering, where the edges of the input graph is reordered in such a way that edges with closer ids have the higher access locality (e.g., Edge A,B,C.. are reordered to Edge 0,1,2,3... in Figure~\ref{fig:orderflow}).
% Then, in the partitioning phase, the ordered edge lists are immediately divided by very fast $\mathcal{O}(1)$ partitioning, called \emph{chunk-based edge partitioning}, inspired by the chunk-based vertex partitioning~\cite{zhu2016gemini}.
% For example in Figure~\ref{fig:orderflow}, the edge list from Edge A to N is permuted into the list from Edge 0 to 13. 
% After that, the edge list is evenly divided into 3 chunks by the chunk-based partitioning.
% Such graph-edge-ordering-based approach significantly reduces the redundant computation.
Once the ordering is computed, the result can be reused, and time-consuming processing is unnecessary to repeat when the number of partitions is changed.
% $k$ can be fixed \emph{after} computing the ordering.
% Our approach consisting of the reordering and chunk-based partitioning has three main advantages.
% First, the result of ordering can be utilized for arbitrary $k$-way partitioning.
% Thus, once the ordering is computed, the result can be reused and no time-consuming processing is necessary. 
% $k$ is fixed \emph{after} computing the ordering.
% Second, the chunk-based partitioning is one of the fastest and most scalable ways to partition the graph in practice.
% The chunk-based partition only computes to search the separation point in the edge list, which can be computed by a few seek operations of the file system since edges are usually stored in a continuous manner.
% If each edge is stored in the fixed byte size (e.g., the length of each edge line is the same), the partition can be computed without knowing the contents of the graph file such as source id and destination id.
% In this case, the chunk-based partition is faster than simple hash-based partitioning.

% Such graph-ordering techniques have been investigated for the maximization of data locality among graph elements.
% Different algorithms have focused on different problems, such as graph compression~\cite{boldi2011layered,lim2014slashburn,dhulipala2016compressing}, CPU-cache utilization~\cite{wei2016speedup,arai2016rabbit}, graph databases~\cite{Goonetilleke:2017:ELS:3085504.3085516}, and matrix bandwidth reduction~\cite{Cuthill:1969:RBS:800195.805928}.
% Our work is the first attempt to apply the ordering technique to the graph partitioning problem.

The contributions of this paper are as follows:

\smallskip
\noindent\textit{\textbf{A Novel Approach to Efficient and Effective Dynamic Scaling of Graph Edge Partitions.}} \\
We formalize the dynamic scaling problem for graph edge partitions as the maximization of both efficiency and quality.
Then, we propose an efficient and effective dynamic scaling approach based on chunk-based edge partitioning and graph edge ordering (Sec.~\ref{sec:dynamic}).
The efficiency is theoretically maximized as we show that the chunk-based edge partitioning is $\mathcal{O}(1)$ while the quality is theoretically guaranteed by the upper bound obtained by the graph edge ordering. 
% The former is an $\mathcal{O}(1)$ algorithm to (re)partition a graph (\S~\ref{sec:chunk}).
% The latter improves the quality of partitions~(\S~\ref{sec:def}).

% Once the graph edge ordering is done, our approach can for any

% For obtaining high partitioning quality, we propose the pre
% For high partitioning quality, we propose the graph edge ordering that is formalized as the optimization problem~(\S~\ref{sec:def}), where the objective is to find the optimal ordering of edges which maximize the quality of partitions when the chunk-based partitioning is applied to the ordered graph.

\smallskip
\noindent\textit{\textbf{A Fast Graph Edge Ordering Algorithm.}}\\
We formalize the graph edge ordering problem as an optimization problem and show its NP-hardness.
To address the NP-hard problem, we propose an efficient $\mathcal{O}(n \log n)$ greedy algorithm based on a greedy expansion.
To enhance the greedy expansion, we propose a novel priority-queue which is significantly effective for the graph edge ordering problem.
We show that partitions generated by the graph edge ordering and the chunk-based edge partitioning have a theoretical upper bounds of the partitioning quality. 
The theoretical result is similar to the best existing static method~(Sec.~\ref{sec:algorithm}). 

% This can compute a billion-edge graph within the currently acceptable time for graph ordering and partitioning.
% This is similar to the existing methods.
% Based on the offline algorithm, we further propose an incremental algorithm for dynamic graphs (\S~\ref{sec:incremental}). 
% The demand to change $k$ is stronger when a graph dynamically increases.

\smallskip
\noindent\textit{\textbf{A Comprehensive Quantitative Evaluation.}}\\
By using large-scale real-world graphs, we evaluate the efficiency and quality of our method and compare it with state-of-the-art dynamic scaling, graph partitioning, and graph ordering methods.
The evaluation shows that the chunk-based edge partitioning is practically between three to eight orders of magnitude faster than the existing methods while achieving comparable quality to that of the best existing static method.
As a result, the high-quality partitions obtained by our method significantly improve the performance of typical benchmarking applications~(Sec.~\ref{sec:evaluation}).
% due to the large reduction of communication volumes~(\S~\ref{sec:evaluation}).
% if we once reorder the graph in advance via the graph edge ordering. 
% We then apply our approach to a distributed graph-processing system and show the performance benefit for the common graph applications, such as SSSP, connected component, and PageRank~(\S~\ref{sec:evaluation}).

\section{Preliminaries and Related Work}
% In this section, we first summarize notations.
% Then, we review the graph edge partitioning and the dynamic scaling of graph edge partitions.
\subsection{Notation}
Let $G = (V, E)$ be an undirected and unweighted graph that consists of a set of vertices $V$ and a set of edges $E$, respectively.
For $E$, the set of its $k$ disjoint subsets are represented as $\mathcal{E}_k := \{ \mathcal{E}_k[p]: 0 \leq p < k, \mathcal{E}_k[p] \subset E, \mathcal{E}_k[i] \cap \mathcal{E}_k[j] = \varnothing \  \text{for}\ i \neq j \}$.
An edge $e$ $(\in E)$ connecting vertex $v$ and $u$ is represented by $e_{v,u}$.
$N(v)$ represents the set of $v$'s neighboring vertices.
The vertex set involved in $E$ is defined as $V(E)$,
that is, $V(E) := \{v\ |\ v \in V, \exists e_{v,u} \in E \}$.
The number of elements in a set is represented by $|\cdot|$, e.g., $|V|$ and $|E|$. 

In this paper, we are interested in the order of elements in $E$.
Let $\phi:~E~\mapsto~\{0,1,2,...,|E|\!-\!1\}$ be a bijective function taking an edge $e$ $(\in E)$ and returning an index $i$ $(0 \leq i < |E|)$.
We refer to $\phi$ as an \emph{ordering function}.
A list (i.e., an ordered set) of $E$ ordered by $\phi$ is represented as $E^{\phi}$.
The $i$-th element in $E^{\phi}$ is represented as $E^{\phi}[i]$.
We also define an \emph{append} operation for the ordered edges, represented~by~$+$.

For example, suppose $A := \{A[0], A[1], A[2] \}$ and $B := \{B[0], B[1]\}$, then
$(A+B)[0] := A[0]$; $(A+B)[1] := A[1]$; $(A+B)[2] := A[2]$; $(A+B)[3] := B[0]$; and $(A+B)[4] := B[1]$.

Notation which we frequently use through the paper is summarized in Table~\ref{tab:notation}.

\begin{table}[h]
\begin{center}
\caption{Summary of Notation}\label{tab:notation}
\scalebox{.9}{
\begin{tabular}{l||l} \hline 
  Symbol             & Description \\ \hline
  $V$, $E$, $G(V,E)$ & Vertices, edges, and a graph with $V$ and $E$ \\
  $N(v)$             & $v$'s neighbor vertices \\ 
  $V(E)$             & Vertices involved in $E$ \\ 
  $k$                & \# of edge partitions \\
  $p$                & Partition id ($0 \leq p < k$)\\
  $\mathcal{E}_k$, $\mathcal{E}_k[p]$ & Set of edge partitions and its $p$-th edge partition \\
  $\phi$             & Ordering function \\
  $E^{\phi}$, $E^{\phi}[i]$         & Edge list ordered by $\phi$ and its $i$-th element \\
  $\mathit{E^{\phi}_{\mathit{ch}}}(i,w)$ & Chunk with $w$ edges from $i$-th edge (\S~\ref{sec:chunk}) \\
  \texttt{ID2P}$_{k}(\cdot)$ & Conversion from Order $i$  to Partition $p$ (\S~\ref{sec:def}) \\ \hline
\end{tabular}
}
\end{center}
\end{table}

\subsection{Graph Edge Partitioning}\label{sec:graphedgepartitioning}
% \smallskip
% \noindent\textit{\textbf{Balanced $k$-way Edge Partitioning}}:
The edge partitioning algorithm divides a set of edges $E$ into $k$ disjoint subsets $\mathcal{E}_{k}$.
The edge partitioning is to find partitions where the communication cost among the partitions becomes as small as possible while keeping the size of each subset balanced.
In the edge partitioning, the communication occurs at boundary vertices, which are replicated into multiple partitions.
Specifically, the number of boundary vertices causing the communication is represented as $\sum_{p = 0}^{k-1} |V\bigl(\mathcal{E}_{k}[p]\bigr)| - |V|$. 
% follows:
% \begin{equation*}
%     \sum_{p = 0}^{k-1} |V\bigl(\mathcal{E}_{k}[p]\bigr)| - |V|
%     \label{eq:replicationfactor}
% \end{equation*}
For evaluating the communication cost, a normalized factor, called \emph{replication factor (RF)}~\cite{joseph2012powergraph}, is typically used:
% specifically, $\mathit{RF}\bigl(\mathcal{E}_{k}\bigr) := \frac{1}{|V|} \sum_{p = 0}^{k-1} |V\bigl(\mathcal{E}_{k}[p]\bigr)|$.
\begin{definition}[Replication Factor]\label{def:replicationfactor}
\begin{equation*}
    \mathit{RF}\bigl(\mathcal{E}_{k}\bigr) := \frac{1}{|V|} \sum_{p = 0}^{k-1} |V\bigl(\mathcal{E}_{k}[p]\bigr)|
\end{equation*}
\end{definition}
Based on $RF$, the edge partitioning problem~\cite{joseph2012powergraph} is defined as follows:
\begin{definition}[Balanced $k$-way Edge Partitioning] \label{def:edgepartitioning}
The objective of \textbf{the balanced $k$-way edge partitioning of} $\boldsymbol{G}$ is formalized as follows:
% \begin{eqnarray*}
%   &             & \min_{\mathit{part} \in \mathcal{P}} RF\bigl(\mathit{part}(E,k)\bigr) \nonumber \\
%   & \text{s.t.} & \max_{0 \leq p < k} |\mathcal{E}_{k}[p]| < (1+\epsilon) \frac{|E|}{k},
% \end{eqnarray*}
\begin{equation*}
  \min_{\mathit{part} \in \mathcal{P}} RF\bigl(\mathit{part}(E,k)\bigr) \ \ \  \text{s.t.}\ \ \max_{0 \leq p < k} |\mathcal{E}_{k}[p]| < (1+\epsilon) \frac{|E|}{k},
\end{equation*}
where $\mathit{part}: (E,k) \mapsto \{\mathcal{E}_{k}[p]: p = 0, 1,..., k\!-\!1\}$ is a partitioning method, and $\mathcal{P}$ is the set of all partitioning methods.
The balance factor $\epsilon \geq 0$  is a constant parameter.
\end{definition}

% \subsection{Dynamic Scaling of Graph Edge Partitions}\label{sec:def-dynamicscaling}
% The dynamic scaling problem for graph edge partitions is formalized recently by W. Fan et al.~\cite{dynamicscaling} as follows:

% \noindent
% \begin{definition}[Dynamic Scaling of Graph] \label{def:dynamicscaling}\ \\
% Suppose $k$ partitions, $\mathcal{E}_{k}$, are assigned to $k$ computational units and we add/remove $x$ units.

% \textbf{Scaling in/out, $\boldsymbol{sc(\mathcal{E}_{k},\pm x)}$}, is to recompute new $k \pm x$ edge partitions, $\mathcal{E}_{k \pm x}$ ($\mathcal{E}_{k \pm x} := sc(\mathcal{E}_{k},\pm x)$).
% The number of the moved edges due to scaling is defined as \textbf{a migration cost}, $\boldsymbol{m(\mathit{sc}, \mathcal{E}_{k}, \pm x)}$.

% The objective of \textbf{the dynamic scaling problem for} $\boldsymbol{\mathcal{E}_{k}}$ is to minimize the migration cost as follows:
% \begin{equation*}
% \min_{\mathit{sc} \in \mathcal{SC}} m(\mathit{sc}, \mathcal{E}_{k},\pm x)\ \ \text{s.t.}\ \ \epsilon < \overline{\epsilon},\ \mathit{RF} < \overline{\mathit{RF}}, 
% \end{equation*}
% where $\mathcal{SC}$ is the set of all scaling methods; $\overline{\epsilon}$ is the bound of the balance factor; and,  $\overline{\mathit{RF}}$ is the bound of the replication factor.
% \end{definition}

% Note that the dynamic scaling here does \emph{not} deal with a dynamic graph, where its structure may change over time. 
% In this problem, the number of partitions is dynamic, but the graph is static.

\subsection{Related Work}
\noindent\textit{\textbf{Dynamic Scaling of Graph Partitions.}}
The dynamic scaling has been extensively investigated for various distributed applications, such as web applications~\cite{chieu2009dynamic,shen2011cloudscale}, database systems~\cite{das2011albatross,Das:2013:EES:2445583.2445588,taft2014store,serafini2014accordion,adya2016slicer,taft2018p,marcus2018nashdb}, streaming systems~\cite{ishii2011elastic,shen2011cloudscale,castro2013integrating,heinze2015online,madsen2017integrative,floratou2017dhalion,Borkowski:2019:MCR:3317315.3329476,Wang:2019:ERE:3299869.3319868}, data analysis~\cite{shen2011cloudscale}, scientific applications~\cite{mao2011auto}, and machine learning~\cite{qiao2018litz}.
The major difference from these efforts is that distributed graph applications are typically communication-intensive workloads.
Thus, our work focuses on the quality of the partitioning as well as the efficiency of the dynamic scaling.

The dynamic scaling for the traditional vertex graph partitioning has been studied in some work~\cite{pujol2011little,vaquero2014adaptive,8798698,8514898}.
% For example, in \cite{8514898}, the proposed framework uses nested hash partitioning, where a partition is computed by the migration function in addition to the basic random hashing.
% \cite{8798698} presents characteristic-based repartitioning, which is aware of system characteristics such as CPU power and network capacity.
% \cite{pujol2011little} proposes Social Partitioning And Replication, which uses a greedy heuristic to solve the MIN\_REPLICA problem instead of the traditional MIN\_CUT when scaling the graph.
% In \cite{vaquero2014adaptive}, a greedy heuristic based on label propagation is used for vertex migration. 
The main difference from these efforts is that our proposal is based on edge partitioning.
Our chunk-based edge partitioning makes full use of the edge partitioning so that its time complexity becomes $\mathcal{O}(1)$.
Achieving $\mathcal{O}(1)$ for vertex partitioning is a very challenging endeavor (if at all possible).
% because each edge needs to be processed to determine its partition one-by-one (thus, it needs $O(|E|)$) even if the $\mathcal{O}(1)$ chunk-based vertex partitioning is utilized.

The work, which appears to be closer to ours, is~\cite{dynamicscaling}. 
To the best of our knowledge, this is the only one to discuss the dynamic scaling of edge partitions.
In this paper, the authors confirm that the minimization of the migration cost in dynamic scaling is NP-complete. 
They propose an approximate algorithm and a generic scheme based on consistent hashing.
The hashing does not take into account the data locality.
As a result, the quality of partitioning is not considered.
% The algorithm is extended to generic scheme to utilize the existing partitioning methods such as HDRF~\cite{Petroni:2015:HSP:2806416.2806424}. 
% Our pre-processing-based approach makes two key differences from their work.
% First, 
In contrast, our approach aims to achieve also high partitioning quality due to the preprocessing (i.e., graph edge ordering) as compared theoretically in Sec.~\ref{sec:upperbound} and empirically in Sec.~\ref{sec:evaluation}.

\noindent\textit{\textbf{Graph Ordering.}}
Due to the structural complexity of the real-world networks, it is difficult to grasp data locality among each graph element.
The graph ordering is one of the major approaches to increase the data locality~\cite{zhao2020graph}.
The most traditional method is Reverse Cuthill McKee (RCM) for matrix bandwidth reduction~\cite{Cuthill:1969:RBS:800195.805928}.
Different algorithms have a different focus, such as graph compression~\cite{boldi2011layered,lim2014slashburn,dhulipala2016compressing}, CPU-cache utilization~\cite{wei2016speedup,arai2016rabbit}, and graph databases~\cite{Goonetilleke:2017:ELS:3085504.3085516}.
% \begin{revise-env}
% In a similar way to our methods, some of the existing methods are based on BFS to efficiency handle large scale graphs, such as \cite{Cuthill:1969:RBS:800195.805928,wei2016speedup}.
% The algorithmic difference from them is our design of the priority queue (Eq.~\eqref{eq:priority}, which is theoretically derived from the graph ordering problem (Eq.~\eqref{eq:ordering1} and Eq.~\eqref{eq:ordering2}).
% \end{revise-env}
Our work is the first attempt to utilize the graph ordering technique for the graph partitioning problem and provides the best partitioning quality as compared in Sec.~\ref{sec:evaluation}.

\section{Proposed Dynamic Scaling Method} \label{sec:dynamic}
In this section, we first provide a formal definition of the problem.
Second, we outline our our approach which is based on preprocessing the graph.
Third, we present the chunk-based edge partitioning algorithm.
% and show that its complexity is $\mathcal{O}(1)$.
Finally, we introduce the graph edge ordering algorithm.
% , which is derived from the chunk-based edge partitioning algorithm.

\subsection{Problem Definition} \label{sec:problemdef}
We formalize the dynamic scaling problem as a multi-objective problem: (i) to maximize the efficiency of the scaling and (ii) to minimize the replication factor of edge partitions generated by the scaling. 
% This is different from the single-objective problem that only focuses on the migration cost (Def.~\ref{def:dynamicscaling}).

Let the number of initial partitions be $k$; the partitioned edge sets be $\mathcal{E}_{k}$; the number of added/removed computing~unit~be~$x$.
\begin{definition}[Dynamic Scaling]\label{def:multi-object}\
\textbf{Scaling in/out, $\boldsymbol{sc(\mathcal{E}_{k},\pm x)}$}, is to recompute new $k \pm x$ edge partitions, $\mathcal{E}_{k \pm x}$ ($\mathcal{E}_{k \pm x} := sc(\mathcal{E}_{k},\pm x)$).

The objective of \textbf{the dynamic scaling problem for $\boldsymbol{\mathcal{E}_{k}}$} is to maximize the efficiency of the scaling ($\mathit{sc}$) and to minimize the replication factor ($RF$) as follows:
\begin{eqnarray*}
  \max_{\mathit{sc} \in \mathcal{SC}} \mathit{EF}(\mathit{sc}(\mathcal{E}_{k}, \pm x)),\  \min_{\mathit{sc} \in \mathcal{SC}} RF\bigl(sc(\mathcal{E}_{k}, \pm x)\bigr) \nonumber \\
  \text{s.t.} \max_{0 \leq p < k\pm x} |\mathcal{E}_{k\pm x}[p]| < (1+\epsilon) \frac{|E|}{k\pm x},
\end{eqnarray*}
where $RF(sc)$ is the replication factor of the new partitions after $sc$, and the efficiency ($\mathit{EF}$) is evaluated by the time complexity to calculate partition IDs of edges.
\end{definition}

% \begin{highlight}

Note that, in a similar way to the state-of-the-art work~\cite{dynamicscaling}, we focus on the dynamic scaling of static graphs, where the structure of the graph \emph{does not} change over time.
In this case,  the graph is static, while the number of partitions changes dynamically to reflect changes in the underlying computational infrastructure.

\subsection{Overview of Proposed Approach}
We address the two objectives above one by one.
Specifically, at first, the efficiency is maximized as we design the very fast $\mathcal{O}(1)$ graph partitioning method.
Then, the quality is maximized by preprocessing of an input graph.

The overall computation consists of five steps as shown in Figure~\ref{fig:orderflow}. 
(i) and (ii) are executed once, whereas (iii)~--~(v) are repeated:
\begin{itemize}
\item[(i)] \textbf{\textit{Graph Edge Ordering:}} The graph-edge-ordering algorithm converts the original graph data into the ordered edge list.
\item[(ii)] \textbf{\textit{Initial Partitioning to $k$ Parts:}} The chunk-based edge partitioning initially computes $k$ edge partitions of the ordered edge list. The graph elements (i.e., vertices and edges) are distributed to $k$ machines accordingly.
\item[(iii)] \textbf{\textit{Resource Provisioning / De-provisioning:}} $x$ computational units are added/removed (e.g., add/remove machine(s), CPU core(s), or CPU Socket(s)).
\item[(iv)] \textbf{\textit{Scaling to $k \pm x$ Parts:}} The chunk-based edge partitioning computes the \mbox{$k \pm x$-way} edge partitions for the ordered edge list. 
The additional graph elements are moved from the other processes or reloaded from the storage.
% , depending on the case.
% For example, if each graph element has a temporal variable, then they need to be moved from the other process; otherwise the data can be from the file system.
% Whereas, if processes are suddenly removed, e.g., by the crush, then the data need to be from the file system (i.e., fault-tolerant systems).
\item[(v)] \textbf{\textit{Graph Application:}} Distributed graph applications are executed on the \mbox{$k \pm x$} machines.
\end{itemize}

\subsection{Chunk-based Edge Partitioning}\label{sec:chunk}
The chunk-based edge partitioning evenly splits the ordered edge list into continuous chunks of edges.
Specifically, \textbf{the chunk-based edge partitioning algorithm} for $p$-th part $(0 \leq p < k)$ takes 3 arguments: (i) the ordered edge list, $E^{\phi}$, (ii) the partition ID, $p$, (iii) the total number of partitions, $k$; and returns a disjoint edge set, $\mathcal{E}_{k}[p]$, in such a way that:
\begin{equation*}
    \mathcal{E}_{k}[p] = \mathit{E^{\phi}_{\mathit{ch}}}\left(\sum_{x=0}^{p-1}\left\lfloor\tfrac{|E|+x}{k}\right\rfloor, \left\lfloor\tfrac{|E|+p}{k}\right\rfloor\right),
    % \mathcal{E}_{k}[p] = \mathit{E^{\phi}_{\mathit{ch}}}\left(\left\lfloor\tfrac{|E|}{k}\right\rfloor p + \theta_{k}(p) , \left\lfloor\tfrac{|E|+p}{k}\right\rfloor\right),
\end{equation*}
where $\mathit{E^{\phi}_{\mathit{ch}}}(\cdot,\cdot)$ is the edge chunk. We define it using its beginning point $i$ and chunk size $w$ as follows:
\begin{eqnarray*}
\mathit{E^{\phi}_{\mathit{ch}}}(i,w) &:=& \{E^{\phi}[i], E^{\phi}[i+1], ..., E^{\phi}[i+ w-1]\}.
% \theta_{k}(p) &:=& \max \big(0, p-k+\left(|E| \text{ mod } k\right)\big) 
\end{eqnarray*}
% Note that if $|E|\text{ mod } k = 0$, then $\mathcal{E}_{k}[p]$ is simplified as follows:
% \begin{equation*}
% \mathcal{E}_{k}[p] = \mathit{E^{\phi}_{\mathit{ch}}}\left(\tfrac{|E|}{k}p,\tfrac{|E|}{k}\right).
% \end{equation*}

There are two noted things.
First, the chunk-based edge partitioning always provides the perfect edge balance, i.e., $\epsilon \approx 0$ in Def.~\ref{def:edgepartitioning}. 
Second, if $|E|\text{ mod } k = 0$, then $\mathcal{E}_{k}[p]$ is simplified as
\begin{equation*}
\mathcal{E}_{k}[p] = \mathit{E^{\phi}_{\mathit{ch}}}\left(\tfrac{|E|}{k}p,\tfrac{|E|}{k}\right).
\end{equation*}
% \end{remark}

\begin{figure*}[t]
  \centering
   \includegraphics[width=2.0\columnwidth]{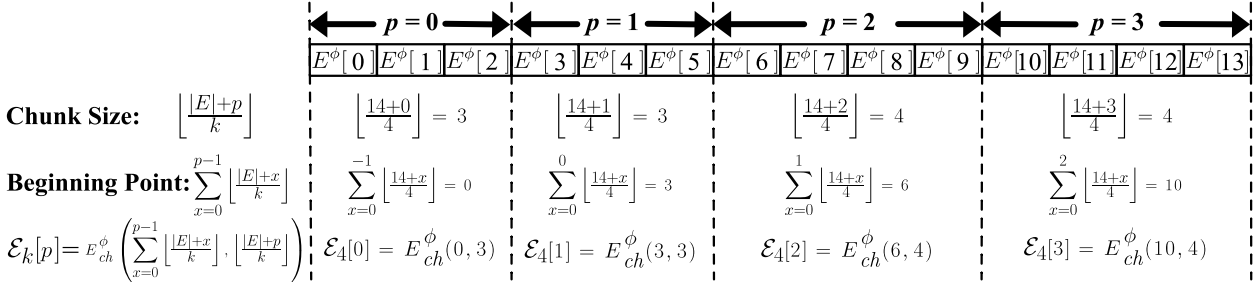}
   
  \caption{Chunk-based Edge Partitioning of $E$ into 4 parts ($k = 4$, $p = 0,1,2,3$, and $E = \{E^{\phi}[0], E^{\phi}[1], ... E^{\phi}[13]\}$).}\label{fig:chunk-based}
  
\end{figure*}

Figure~\ref{fig:chunk-based} shows the example of the chunk-based edge partitioning for $E = \{E^{\phi}[0], E^{\phi}[1], ... E^{\phi}[13]\}$ and $k=4$.
The 14 edges are divided into 3 + 3 + 4 + 4 edges because $\left\lfloor\tfrac{|E|+p}{k}\right\rfloor$ for each $p \ (0\leq p < k)$ is equal to $\left\lfloor\tfrac{14+0}{4}\right\rfloor = 3$, $\left\lfloor\tfrac{14+1}{4}\right\rfloor = 3$, $\left\lfloor\tfrac{14+2}{4}\right\rfloor = 4$, and $\left\lfloor\tfrac{14+3}{4}\right\rfloor = 4$, respectively.
Therefore, $\mathcal{E}_{k}[p]$ becomes $\mathcal{E}_{4}[0] = E^{\phi}_{ch}(0,3)$, $\mathcal{E}_{4}[1] = E^{\phi}_{ch}(3,3)$, $\mathcal{E}_{4}[2] = E^{\phi}_{ch}(6,4)$, and $\mathcal{E}_{4}[3] = E^{\phi}_{ch}(10,4)$, respectively.

\medskip
Since the chunk-based partitioning just splits the edge list, the computational time complexity excluding the graph data movement is basically $\mathcal{O}(1)$.
% Note that even for the very efficient hash-based algorithms such as 1D hashing, their complexity is $\mathcal{O}(n)$ since they need to process the graph elements one by one.

\begin{theorem}[Efficiency of Partitioning]\label{thr:cep}
Suppose the edges of $E^{\phi}$ are stored continuously (e.g., to an array or a file system), and an operation to find the pointer of $E^{\phi}[i]$ by using $i$ is $\mathcal{O}(1)$ (e.g., RAM or standard file systems).
Then, there exists an $\mathcal{O}(1)$ algorithm to compute the chunk-based edge partitioning excluding the graph data movement, and it does not depend on the graph size, such as $|V|$ and $|E|$.
% Moreover, the data movement in the chunk-based edge partitioning requires $\mathcal{O}(|E| / k)$ computational costs.
\end{theorem}

\begin{proof}
In order to compute the chunk-based edge partitioning in $\mathcal{O}(1)$, the algorithm needs to calculate \\ $\sum_{x=0}^{p-1}\left\lfloor\tfrac{|E|+x}{k}\right\rfloor$ in $\mathcal{O}(1)$.
$\sum_{x=0}^{p-1}\left\lfloor\tfrac{|E|+x}{k}\right\rfloor$ requires $\mathcal{O}(p)$ computational time in a naive way.

The summation can be modified as follows:
\begin{eqnarray*}
\sum_{x=0}^{p-1}\left\lfloor\tfrac{|E|+x}{k}\right\rfloor = \sum_{x=0}^{p-1}\left\{ \left\lfloor\tfrac{|E|}{k}\right\rfloor + \left\lfloor\tfrac{(|E| \text{ mod } k) + x}{k}\right\rfloor \right\}
\end{eqnarray*}
Here, $\left\lfloor\tfrac{(|E| \text{ mod } k) + x}{k}\right\rfloor$ is $0$ or $1$ in $x = 0,1,...,p-1$, as follows:
\begin{eqnarray*}
\left\lfloor\tfrac{(|E| \text{ mod } k) + x}{k}\right\rfloor = \begin{cases} 
                                                                  0 & \text{ if } (|E| \text{ mod } k) + x < k\\
                                                                  1 & \text{ otherwise}        
                                                                \end{cases}
\end{eqnarray*}
Therefore, 
\begin{eqnarray*}
\sum_{x=0}^{p-1}\left\lfloor\tfrac{(|E| \text{ mod } k) + x}{k}\right\rfloor &=& 
\begin{cases}
  0 \ \ \ \text{ if } k + (|E| \text{ mod } k) \geq p\\
  p - k + (|E| \text{ mod } k) \ \ \ \text{ otherwise}
\end{cases} \\
&=& \max \big(0, p-k+\left(|E| \text{ mod } k\right)\big)
\end{eqnarray*}
We define $\theta_{k}(p):= \max \big(0, p-k+\left(|E| \text{ mod } k\right)\big)$.
Then, the following formula is established:
\begin{equation*}
\sum_{x=0}^{p-1}\left\lfloor\tfrac{|E|+x}{k}\right\rfloor = p\left\lfloor\tfrac{|E|}{k}\right\rfloor + \theta_{k}(p).
\end{equation*}
This can be computed in $\mathcal{O}(1)$. \qed
\end{proof}

According to~\cite{dynamicscaling}, the migration cost is defined as the number of migrated edges.
The migration cost for the chunk-based edge partitioning is provided as follows:
% \end{highlight}
% \begin{revise-env}
\begin{theorem}[Migration Cost] \label{thr:migration}
Suppose a set of ordered edges is initially split into $k$ partitions via the chunk-based edge partitioning, and the edges are repartitioned into $k+x$ parts by adding $x$ new processes (i.e., scale out).
We assume that $|E|$ is much larger than $k$ and $x$ such that $(|E| \mod k + x) / |E| < (k+x) / |E| \approx 0$ and that the ids of new partitions are $k, k+1, ... , k+x-1$.

Then, the approximate number of migrated edges when applying repartitioning is 
\begin{equation*}
\frac{x |E|}{2 k (k+x)} \left\lceil \frac{k}{x}\right\rceil \left(\left\lceil \frac{k}{x}\right\rceil + 1\right) + \frac{|E|}{k}\left(k - \left\lceil\frac{k}{x}\right\rceil \right).
\end{equation*}
The cost for scaling in is the same (i.e., from $k+x$ to $k$ partitions) since it is a reverse operation of scaling out.
\end{theorem}
\begin{proof}
We consider a simple case where $|E| \mod k = 0, |E| \mod (k+1) = 0, |E| \mod (k+2) = 0, ... ,\ |E| \mod (k+x) = 0$.
Then, there are two cases in the edge migration for partition $i (i \in [0,k))$: (i) some of the edges in partition $i$ are migrated to other partitions, or (ii) all of the edges in partition $i$ are migrated to other partitions.

% Then, for partition $i$ ($i \in [0,k)$), the edges from the beginning (i.e., $i\frac{|E|}{k}$-th edge) to $(i+1)\frac{|E|}{k+1}$-th are kept in partition $i$, while from $(i+1)\frac{|E|}{k+1}$-th to $(i+1)\frac{|E|}{k}$-th edges are migrated to partition $i+1$.
% Therefore, the total number of the migrated edges are $\sum_{i \in [0,k)} \left((i+1)\frac{|E|}{k} - (i+1)\frac{|E|}{k+1}\right) = \frac{|E|}{2}$.

\medskip
\noindent
\underline{\textit{Case (i):}} 
In this case, for partition $i$, the edges from $i\frac{|E|}{k}$-th edge to $(i+1)\frac{|E|}{k+x}$-th are kept in partition $i$, while from $(i+1)\frac{|E|}{k+x}$-th to $(i+1)\frac{|E|}{k}$-th edges are migrated to other partitions.

Thus the number of migrated edges for partition $i$ is represented as follows:
\begin{eqnarray*}
(i+1)\frac{|E|}{k} - (i+1)\frac{|E|}{k+x} = (i+1)\frac{|E|n}{(k+x)k}
\end{eqnarray*}
Case (i) happens when $(i+1)\frac{|E|n}{(k+x)k} > \frac{|E|}{k}$.
\begin{eqnarray*}
(i+1)\frac{|E|n}{k(k+x)} > \frac{|E|}{k} \Leftrightarrow (i+1) > \frac{k+x}{x} \Leftrightarrow i > \frac{k}{x}
\end{eqnarray*}
Therefore, Case (i) happens when $i > \frac{k}{x}$.

\noindent
\underline{\textit{Case (ii):}} 
In the other case (i.e., $i \leq \frac{k}{x}$), all of the edges in partition $i$ are migrated to other partitions.
Thus, the number of migrated edges for partition $i$ is $\frac{|E|}{k}$.

\medskip
Therefore, to summarize Cases (i) and (ii), the total number of migrated edges from $i = 0$ to $i = k-1$ is formalized as follows:
\begin{eqnarray*}
&&\sum_{0 \leq i < \frac{k}{x}} (i+1) \frac{|E|x}{(k+x)k} + \sum_{\frac{k}{x} \leq i < k} \frac{|E|}{k} \\
&=& \frac{|E|x}{(k+x)k} \sum_{0 \leq i < \frac{k}{x}} (i+1) + \frac{|E|}{k} \sum_{\frac{k}{x} \leq i < k} 1 \\
&=& \frac{x |E|}{2 k (k+x)} \left\lceil \frac{k}{x}\right\rceil \left(\left\lceil \frac{k}{x}\right\rceil + 1\right) + \frac{|E|}{k}\left(k - \left\lceil\frac{k}{x}\right\rceil \right)
\end{eqnarray*}

The aforementioned simplified proof can be straightforwardly generalized for the case of $|E| \mod k \not= 0, |E| \mod k+1 \not= 0, ...,\ |E| \mod k+x \not= 0$, based on the assumption $(|E| \mod k + x) / |E| \approx 0$. \qed
\end{proof}

% The proof is discussed in Appendix B of~\cite{full}.
In practice, a process is typically added or removed incrementally, i.e., $x=1$.
We can simply obtain the following corollary from the theorem.
% The number of the migrated edges in chunk-based edge partitioning is significantly smaller than the naive random way. 
\begin{corollary}[Migration Cost in $x=1$]
The number of migrated edges for $x=1$ is approximately $\frac{|E|}{2}$.
\end{corollary}
% \begin{proof}
% For $n=1$, the number of migrated edges in Theorem~\ref{thr:migration} is calculated as follows:
% \begin{eqnarray*}
% &&\frac{n |E|}{2 k (k+n)} \left\lceil \frac{k}{n}\right\rceil \left(\left\lceil \frac{k}{n}\right\rceil + 1\right) + \frac{|E|}{k}\left(k - \left\lceil\frac{k}{n}\right\rceil \right) \\
% &=& \frac{1 |E|}{2 k (k+1)} \left\lceil \frac{k}{1}\right\rceil \left(\left\lceil \frac{k}{1}\right\rceil + 1\right) + \frac{|E|}{k}\left(k - \left\lceil\frac{k}{1}\right\rceil \right) \\
% &=& \frac{|E|}{2}
% \end{eqnarray*}
% We consider a simple case where $|E| \mod k = 0$ and $|E| \mod k+1 = 0$. 
% Then, for partition $i$ ($i \in [0,k)$), the edges from the beginning (i.e., $i\frac{|E|}{k}$-th edge) to $(i+1)\frac{|E|}{k+1}$-th are kept in partition $i$, while from $(i+1)\frac{|E|}{k+1}$-th to $(i+1)\frac{|E|}{k}$-th edges are migrated to partition $i+1$.
% Therefore, the total number of the migrated edges are $\sum_{i \in [0,k)} \left((i+1)\frac{|E|}{k} - (i+1)\frac{|E|}{k+1}\right) = \frac{|E|}{2}$.
% \begin{eqnarray*}
% \sum_{i \in [0,k)} \left((i+1)\frac{|E|}{k} - (i+1)\frac{|E|}{k+1}\right) &=& \frac{|E|}{k(k+1)}\sum_{i \in [0,k)}(i+1) \\ 
% &=& \frac{|E|}{k(k+1)} \times \frac{1}{2} k(k+1) \\
% &=& \frac{|E|}{2}.
% \end{eqnarray*}
% The simplified proof above is straightforwardly generalized for a case of $|E| \mod k \not= 0, |E| \mod k+1 \not= 0$, based on the assumption $(|E| \mod k) / |E| \approx 0$.
% \end{proof}
% Even though our approach does not mainly focus on the migration cost, 
The result (i.e., $\frac{|E|}{2}$) is significantly smaller than the random way, which may migrate $\frac{k}{k+1}|E|$ edges from $k$ to $k+1$ partitions in average, i.e., approximately $\frac{k}{k+1}|E|$ edges are migrated while $\frac{1}{k+1}|E|$ are kept in the same partition. 
% and this is the same as the state-of-the-art method (Theorem 3 in~\cite{dynamicscaling}).
% \end{revise-env}

% Since to deploy the elements of each $k$-way partitioned graph always requires at least $\mathcal{O}(|E| / k)$ per machine, 
% This theorem readily leads to the following corollary.
% \begin{corollary}
% There exists no other balanced $k$-way graph partitioning method which has better time complexity than the chunk-based edge partitioning.
% \end{corollary}

% \todo{the reason why the vertex partitioning is impossible to be $\mathcal{O}(1)$}.
% In practice, as we will show in Section~\ref{sec:evaluation}, the chunk-based edge partitioning is significantly faster than all the other existing methods because it does not process each graph element one-by-one.
% Since the real-world graph is large, 
% % and the file system in the distributed environment, such as NFS, HDFS, GPFS, and Lustrue, is often located in different machines from the processing system, 
% our $\mathcal{O}(1)$ method highly benefits the performance improvement.
% The existing approaches basically require at least $\mathcal{O}(|E|)$ for processing each edge.
% For example, even the simplest 1D-hash edge partitioning processes each edges one-by-one, and thus its time complexity is $\mathcal{O}(|E|)$.
% The vertex partitioning generally need to process each edge one-by-one for assigning the calculated partitioning id to the source and destination vertices, and it also requires over $\mathcal{O}(|E|)$ time.

\subsection{Graph Edge Ordering}\label{sec:def}
To improve the partitioning quality of the chunk-based edge partitioning, the graph edge ordering orders the input edges in advance in such a way that closer edges in the graph have closer edge ids. 
% We define the graph edge ordering problem as an optimization problem.
% It is theoretically derived from the original $k$-way edge partitioning and the chunk-based edge partitioning method.
% Then, we modify the definition in order to design a greedy algorithm which we will present in the next section. 

\medskip
\noindent\textbf{\textit{Formulation of Graph Edge Ordering.}}
We formulate the graph edge ordering problem as an optimization problem.
It is theoretically derived from the balanced $k$-way edge partitioning problem and the chunk-based edge partitioning.

According to Sec.~\ref{sec:chunk}, the replication factor of $k$ edge partitions generated by the chunk-based edge partitioning is represented as follows:
\begin{equation*}
  \frac{1}{|V|} \sum_{p = 0}^{k-1} \left|V\left(E^{\phi}_{ch}\left(\sum_{x=0}^{p-1}\left\lfloor\tfrac{|E|+x}{k}\right\rfloor, \left\lfloor\tfrac{|E| + p}{k}\right\rfloor \right)\right)\right|
    % \frac{1}{|V|} \sum_{p = 0}^{k-1} |V(\{e_{\frac{|E|}{k} p},..., e_{\frac{|E|}{k} p + \frac{|E|}{k} - 1} \})|
\end{equation*}
The goal of our problem is to minimize the replication factor for arbitrary $k$.
Let $k_{min}$ be the upper bound and $k_{max}$ be the lower bound, i.e., $k_{min} \leq k \leq k_{max}$ (as discussed in the empirical analysis of the distributed graph systems and partitioning~\cite{Han:2014:ECP:2732977.2732980,6877273,Verma:2017:ECP:3055540.3055543,abbas2018streaming,Gill:2018:SPP:3297753.3316427,Pacaci:2019:EAS:3299869.3300076}, $k_{min}$ is typically less than ten while $k_{max}$ is close to one hundred in practice).   
% Then, we evaluate the summation from $k=k_{min}$ to $k_{max}$:
% \begin{equation*}
%     \frac{1}{|V|} \sum_{k=k_{min}}^{k_{max}} \sum_{p = 0}^{k-1} \left|V\left(E^{\phi}_{ch}\left(\sum_{x=0}^{p-1}\left\lfloor\tfrac{|E|+x}{k}\right\rfloor, \left\lfloor\tfrac{|E| + p}{k}\right\rfloor \right)\right)\right|,
%     % \sum_{p = 0}^{k-1} \frac{1}{|V|} |V(\{e_{\frac{|E|}{k} p},..., e_{\frac{|E|}{k} p + \frac{|E|}{k} - 1} \})|,
% \end{equation*}
% where $k_{min} \geq 2$ and $k_{max} \leq |E|$.
Thus, the objective is to find edge ordering which minimizes the summation of the above formula from $k=k_{min}$ to $k_{max}$.
\begin{definition}[Graph Edge Ordering I]\label{def:ordering1}
The objective of \textbf{the graph edge ordering problem} is formalized as follows:
\begin{equation}\label{eq:ordering1}
  \min_{\phi \in \Phi} \frac{1}{|V|}  \sum_{k=k_{min}}^{k_{max}} \sum_{p = 0}^{k-1} \left|V\left(E^{\phi}_{ch}\left(\sum_{x=0}^{p-1}\left\lfloor\tfrac{|E|+x}{k}\right\rfloor, \left\lfloor\tfrac{|E| + p}{k}\right\rfloor \right)\right)\right|,
\end{equation}
where $k_{min} \geq 2$; $k_{max} \leq |E|$; and $\Phi$ is the set of all orders for the edges.
\end{definition}

\smallskip
\noindent\textbf{\textit{NP-hardness of Graph Edge Ordering Problem.}}
The graph ordering problem is NP-hard because the graph partitioning is already NP-hard when the number of partitions is fixed.

\begin{theorem}[NP-hardness]\label{thr:nphard}
The graph edge ordering problem is NP-hard if $|E|$ is much larger than $k_{\mathit{max}}$ so that less than $k_{\mathit{max}}$ edges do not affect the optimized result.
\end{theorem}
\begin{proof}
We first show that the graph edge ordering problem is NP-hard for single $k$, i.e., $k_{min} = k_{max}$.
We then prove the general case of multiple $k$, i.e., $k_{min} < k_{max}$.

\smallskip
\noindent
\underline{\emph{Case of Single $k$}:} 
Suppose $k_{min} = k_{max} = k$.
% From the assumption, we can regard $|E| \text{ mod } k_0$ as zero.
% Thus, from the simplified representation of the edge chunk in Definition~\ref{def:chunkpartitioning}, the graph edge ordering problem in Definition~\ref{def:ordering1} is modified as follows:
The objective of the graph edge ordering problem is represented as follows:
\begin{equation} \label{eq:k_0}
%  \min_{\phi \in \Phi} \frac{1}{|V|} \sum_{p = 0}^{k_0 -1} \left|V\left(E^{\phi}_{ch}\left(
%  \tfrac{|E|}{k_0} p, \tfrac{|E|}{k_0} \right)\right)\right|.
\min_{\phi \in \Phi} \frac{1}{|V|} \sum_{p = 0}^{k -1} \left|V\left(E^{\phi}_{ch}\left(
\sum_{x=0}^{p-1}\left\lfloor\tfrac{|E|+x}{k}\right\rfloor, \left\lfloor\tfrac{|E| + p}{k}\right\rfloor \right)\right)\right|.
\end{equation}

Now, we define a function to convert the edge order into the partition, \texttt{ID2P}$_{k}$: $i \mapsto p$, as Algorithm~\ref{alg:id2partition_app}.
By using \texttt{ID2P}$_{k}$, we can generate new edge partitions from the edge orders in linear time.

\setcounter{algocf}{1}
\begin{algorithm}[h]
\caption{Conversion from Edge ID to Partition}\label{alg:id2partition_app}
% \SetKwInOut{Var}{Var}
% \SetKwInOut{VarEmpty}{}
\SetKwInOut{Input}{Input}
\SetKwInOut{InputEmpty}{}
\SetKwInOut{Output}{Output}
\SetKwProg{Fn}{}{}{}
\SetKwFunction{IDtoP}{ID2P$_{k}$}
\Input{$i$ -- Ordered Edge ID}
\Output{$p$ -- Partition ID}

\DontPrintSemicolon
\BlankLine
\Fn{\IDtoP{$i$}}{
  $p$ $\gets$ $0$; $\mathit{cur}$ $\gets$ $\left\lfloor\tfrac{|E| + p}{k}\right\rfloor$\;
  \While{$i < \mathit{cur}$}{
    $p$ $\gets$ $p+1$; $\mathit{cur}$ $\gets$ $\mathit{cur} + \left\lfloor\tfrac{|E| + p}{k}\right\rfloor$\;
  }
  \Return $p$ \;
}
\end{algorithm}

Suppose the order~$\phi_{\mathit{opt}}$ is the optimal solution for the graph edge ordering problem.
Then, the edge partitions converted from $\phi_{\mathit{opt}}$ via \texttt{ID2P}$_{k}$ is also the optimal solution for the edge partitioning problem in a case when $\epsilon \approx 0$ in Def.~\ref{def:edgepartitioning}.

The reason is as follows.
If the edge partitions converted from $\phi_{\mathit{opt}}$ via \texttt{ID2P}$_{k}$ is \emph{not} the optimal solution (more specifically, more than $k_{\mathit{max}}$ edges are in the different partitions from the optimal partitions), then there exist another optimal edge partitions, $\mathcal{E}^{opt}_{k} := \{\mathcal{E}^{opt}_{k}[p]\ |\  0 \leq p < k\}$, which provides a better solution for the edge partitioning problem than $\phi_{\mathit{opt}}$.
Based on $\mathcal{E}^{opt}_{k}$, we can generate new edge ordering $\phi'$ in such a way that for $p$
% , $E^{\phi'}[0], E^{\phi'}[1],..., E^{\phi'}[|E|-1]$, 
% in such a way that 
\begin{equation*}
\mathcal{E}^{opt}_{k}[p] = \left\{E^{\phi'}\left[ \mathit{b} \right], E^{\phi'}\left[ \mathit{b} + 1 \right], ..., E^{\phi'}\left[\mathit{b} + \lfloor\tfrac{|E| + p}{k}\rfloor - 1\right] \right\},
\end{equation*}
where $\mathit{b} := \sum_{x=0}^{p-1}\left\lfloor\tfrac{|E|+x}{k}\right\rfloor$.
% $E^{opt}_{k}[p] =$ \tiny$\left\{E^{\phi'}\left[\sum_{x=0}^{p-1}\left\lfloor\tfrac{|E|+x}{k}\right\rfloor \right], ..., E^{\phi'}\left[\sum_{x=0}^{p-1}\left\lfloor\tfrac{|E|+x}{k}\right\rfloor + \left\lfloor\tfrac{|E| + p}{k}\right\rfloor - 1\right] \right\}$ \normalsize for each $p$.
Since $\mathcal{E}^{opt}_{k}$ provides the optimal solution,
\begin{eqnarray*}
&&RF(\mathcal{E}^{opt}_{k}) := \frac{1}{|V|} \sum_{p = 0}^{k -1} |V\bigl(\mathcal{E}^{opt}_{k}[p]\bigr)| \\
&=&  \frac{1}{|V|} \sum_{p = 0}^{k -1} \left|V\left(E^{\phi'}_{ch}\left(
\sum_{x=0}^{p-1}\left\lfloor\tfrac{|E|+x}{k}\right\rfloor, \left\lfloor\tfrac{|E| + p}{k}\right\rfloor \right)\right)\right|
\end{eqnarray*}
is the optimal value. 
On the other hand, $\phi_{\mathit{opt}}$ provides the optimal value of Eq.~\eqref{eq:k_0} as follows:
\begin{equation*}
\frac{1}{|V|} \sum_{p = 0}^{k -1} \left|V\left(E^{\phi_{\mathit{opt}}}_{ch}\left(
\sum_{x=0}^{p-1}\left\lfloor\tfrac{|E|+x}{k}\right\rfloor, \left\lfloor\tfrac{|E| + p}{k}\right\rfloor \right)\right)\right|.
\end{equation*}
This is a contradiction to the assumption that $\mathcal{E}^{opt}_{k}$ provides the better solution than $\phi_{\mathit{opt}}$.
Thus, $\phi_{\mathit{opt}}$ can provide the optimal solution for the edge partitioning problem as well.

% if the optimal ordering solution $\phi_{\mathit{opt}}$ of the problem~\eqref{eq:k_0} is given, then it can also provide the optimal solution of the edge partitioning problem as discussed in Definition~\ref{def:edgepartitioning}.
% (Please see Appendix~\ref{sec:appendix} for the details).

Therefore, the problem~\eqref{eq:k_0} is reducible to the balanced $k$-way edge partitioning problem, which is an NP-hard problem as proved in~\cite{Zhang:2017:GEP:3097983.3098033}.

\smallskip
\noindent
\underline{\emph{Case of $k_{min} < k_{max}$}:}
We explain the case when $k_{min} = 2$ and $k_{max} = 3$.
The following discussion can be straightforwardly generalized to any $k_{\mathit{min}}$ and $k_{\mathit{max}}$.

According to Def.~\ref{def:ordering1}, we define a function, $\mathit{Num}(k,p)$, for the normalized number of vertices involved in the chunk of edges as follows:
\begin{equation*}
   \mathit{Num}(k,p) := \frac{1}{|V|}\left|V\left(E^{\phi}_{ch}\left(\sum_{x=0}^{p-1}\left\lfloor\tfrac{|E|+x}{k}\right\rfloor, \left\lfloor\tfrac{|E| + p}{k}\right\rfloor\right)\right)\right|.
\end{equation*}
Suppose $k_{min} = 2$ and $k_{max} = 3$, we will show the NP-hardness of the optimization problem as follows:
\begin{align}
\min_{\phi \in \Phi}\sum_{k=2}^{3} \sum_{p = 0}^{k-1} N(k,p) = \min_{\phi \in \Phi}\{\mathit{Num}(2,0)+\mathit{Num}(2,1) \nonumber \\
+\mathit{Num}(3,0)+\mathit{Num}(3,1)+\mathit{Num}(3,2) \label{eq:opt} \}.
\end{align}
Here, based on the above discussion of the single $k$, the following optimization problems are already proved to be NP-hard:
\begin{eqnarray}
&&\min_{\phi \in \Phi} \left\{ \mathit{Num}(2,0) + \mathit{Num}(2,1) \right\} \label{eq:opt1} \\
&&\min_{\phi \in \Phi} \left\{ \mathit{Num}(3,0) + \mathit{Num}(3,1) + \mathit{Num}(3,2) \label{eq:opt2} \right\}.
\end{eqnarray}
Suppose $\phi_{\mathit{opt}}$ is the optimal order for \eqref{eq:opt}, then the order can be also the optimal for \eqref{eq:opt1} and \eqref{eq:opt2}.
Thus, if \eqref{eq:opt} is not NP-hard, it is a contradiction to the NP-hardness of \eqref{eq:opt1} and \eqref{eq:opt2}.
Therefore, \eqref{eq:opt} is also NP-hard.
To summarize, the graph edge ordering problem is NP-hard. \qed
\end{proof}

\section{Greedy Algorithm for Graph Edge Ordering} \label{sec:algorithm}
Due to the NP-hardness of the graph edge ordering problem, we require an approximation algorithm to solve the problem within an acceptable time.
In this section, we propose a greedy algorithm for the graph ordering problem.

Our key idea is \emph{greedy expansion}, as illustrated in Figure~\ref{fig:overviewgreedy}.
The algorithm initially selects a single vertex at random and assigns orders to its neighbors from $0$.
After that, it greedily selects a vertex from the frontier vertices of the already ordered part so that the score of the objective function becomes the local minimum.
Then, new orders are assigned to the neighbors of the selected vertex.
The expansion is executed iteratively until all edges are ordered.

\begin{figure}[h]
  \vspace{-10pt}
  \centering
   \includegraphics[width=\columnwidth]{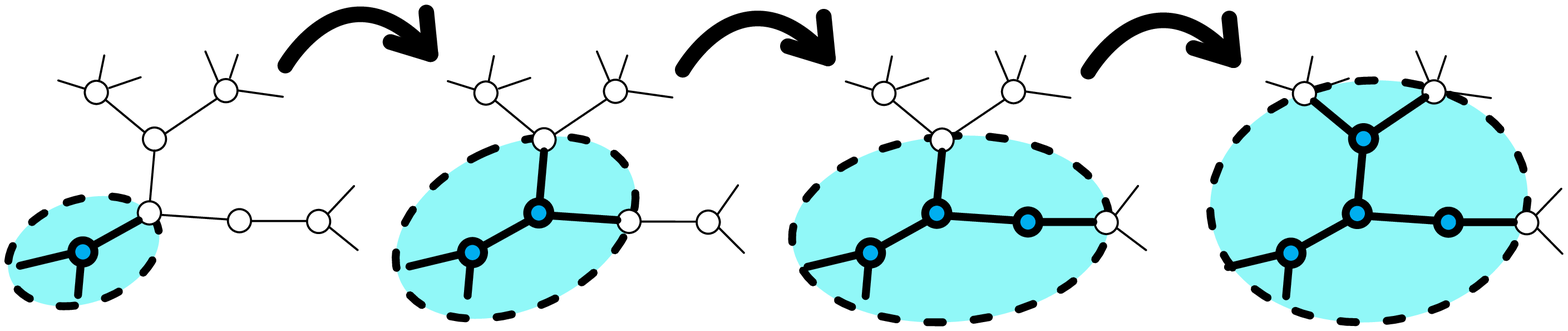}
 
  \caption{Greedy Expansion.}\label{fig:overviewgreedy}

\end{figure}

To find the local optimum in each iteration, the greedy expansion needs to calculate the objective function (Eq.~\eqref{eq:ordering1} in Def.~\ref{def:ordering1}) for \emph{partial ordered edges}, $X^{\phi}$($\subseteq E^{\phi}$). 
However, Eq.~\eqref{eq:ordering1} is defined only for the entire edges (i.e., $E^{\phi}$) and cannot be computed for $X^{\phi}$.
Thus, we modify the summation over $p$ (i.e.,~$\sum_{p}$) in Eq.~\eqref{eq:ordering1} into one over $E$ (i.e.,~$\sum_{E}$) so that the algorithm can evaluate $X^{\phi}$ in each iteration.
% as shown in the following definition.

To do so, we additionally define a function $S$ that detects candidates for the splitting points when the edges will be partitioned via the chunk-based edge partitioning.
Based on $S$, Eq.~\eqref{eq:ordering1} is modified into an summation over $E$ (i.e.,~$\sum_{E}$).

% During the greedy expansion, the algorithm needs to evaluate the objective function for the partly ordered edges to find the local optimal in each iteration.
% To do so, we modify the objective function formed of summation over $p$ (i.e.,~$\sum_{p}$) in Definition~\ref{def:ordering1} into one over $E$ (i.e.,~$\sum_{E}$) as follows:

\begin{definition}[Graph Edge Ordering II]\label{def:ordering2} \ \\ 
Suppose
\begin{eqnarray*}
% F(i, e) := \displaystyle \sum_{w = w_{\mathit{min}}}^{w_{\mathit{max}}} \beta(i, w) \cdot \left|V\left( E^{\phi}_{ch}\left(i-w+1, w\right)\right)\right|, \\
&f_{k}(i, w) := \displaystyle S_{k}(i) \cdot \left|V\left( E^{\phi}_{ch}\left(i-w +1, w \right)\right)\right|, \\
&S_{k}(i) := \begin{cases} 
                  1 \ \text{ if } \texttt{ID2P}_{k}(i) \not= \texttt{ID2P}_{k}(i+1) \text{ or } i = |E| - 1\\
                  0 \ \text{otherwise}
               \end{cases}
\end{eqnarray*}
where we extend the definition of the edge chunk such that $E^{\phi}_{ch}(i', w) := E^{\phi}_{ch}(0, w)$ for $i' < 0$.

Then, the objective of \textbf{the graph edge ordering problem} is redefined as follows:
\begin{equation}\label{eq:ordering2}
    \min_{\phi \in \Phi} \frac{1}{|V|} \sum_{k = k_{\mathit{min}}}^{k_{\mathit{max}}} \sum_{i=0}^{|E|-1} f_{k}\left(i, \left\lfloor\tfrac{|E| + \texttt{ID2P}_{k}(i)}{k}\right\rfloor\right).
\end{equation}
% where $w_{\mathit{min}} := \tfrac{|E|}{k_{\mathit{max}}}$ and $w_{\mathit{max}} := \tfrac{|E|}{ k_{\mathit{min}}}$ in Definition~\ref{def:ordering1}.
\end{definition}

The following gives the correctness of the modification.
\begin{lemma}\label{lemm:eq_def}
Definition~\ref{def:ordering1} and \ref{def:ordering2} are equivalent.
\end{lemma}
\begin{proof}
In Eq.~\eqref{eq:ordering2}, according to the definition of $S_{k}(i)$, \\
$f_{k}\left(i, \left\lfloor\tfrac{|E| + \texttt{ID2P}_{k}(i)}{k}\right\rfloor\right)$ is non-zero only if
\begin{eqnarray*}
i = \left\lfloor\tfrac{|E|}{k}\right\rfloor - 1, \left\lfloor\tfrac{|E|}{k}\right\rfloor + \left\lfloor\tfrac{|E| + 1}{k}\right\rfloor - 1, ..., \sum_{x=0}^{k-1} \left\lfloor\tfrac{|E| + x}{k}\right\rfloor -1
%  = \left\lfloor\tfrac{|E|}{k}\right\rfloor - 1, 2 \left\lfloor\tfrac{|E|}{k}\right\rfloor + \theta_{k}(2) - 1, ..., k\left\lfloor\tfrac{|E|}{k}\right\rfloor + \theta_{k}(k) - 1
\end{eqnarray*}
Therefore,
\begin{eqnarray*}
&&\displaystyle\sum_{k = k_{\mathit{min}}}^{k_{\mathit{max}}} \sum_{i=0}^{|E|-1}  f_{k}\left(i, \left\lfloor\tfrac{|E| + \texttt{ID2P}_{k}(i)}{k}\right\rfloor\right) \\
&=&\displaystyle\sum_{k = k_{\mathit{min}}}^{k_{\mathit{max}}} \sum_{p = 0}^{k-1} f_{k}\left(\sum_{x=0}^{p}\left\lfloor\tfrac{|E|+x}{k}\right\rfloor - 1, \left\lfloor\tfrac{|E|+p}{k}\right\rfloor\right) \\
% &=&\displaystyle\sum_{k = k_{\mathit{min}}}^{k_{\mathit{max}}} \sum_{p = 0}^{k-1} \scriptstyle \left|V\left( E^{\phi}_{ch}\left( \displaystyle \sum_{x=0}^{p} \scriptstyle \left\lfloor\tfrac{|E|+x}{k}\right\rfloor - 1 - \left\lfloor\tfrac{|E| + p}{k}\right\rfloor + 1, \left\lfloor\tfrac{|E| + p}{k}\right\rfloor \right)\right)\right|,\\
&=&\displaystyle\sum_{k = k_{\mathit{min}}}^{k_{\mathit{max}}} \sum_{p = 0}^{k-1} \left|V\left( E^{\phi}_{ch}\left( \sum_{x=0}^{p-1}\left\lfloor\tfrac{|E|+x}{k}\right\rfloor, \left\lfloor\tfrac{|E| + p}{k}\right\rfloor \right)\right)\right|,
\end{eqnarray*}
which is equal to Eq.~\eqref{eq:ordering1}.
Therefore, Def.~\ref{def:ordering1} and Def.~\ref{def:ordering2} are equivalent. \qed
\end{proof}

Then, we extend the objective function in Def.~\ref{def:ordering2} for the partial ordered edges, $X^{\phi}$($\subseteq E^{\phi}$), as follows:
\begin{eqnarray}
    &\displaystyle \frac{1}{|V|} \sum_{k = k_{\mathit{min}}}^{k_{\mathit{max}}} \sum_{i=0}^{|E|-1} f_{k}\left(X^{\phi}, i, \left\lfloor \tfrac{|E|+\texttt{ID2P}_{k}(i)}{k}\right\rfloor\right), \label{eq:genobj}
\end{eqnarray}
where
\begin{eqnarray}
    &f_{k}\left(X^{\phi}, i, w\right) := S_{k}(i) \cdot \left|V\left(X^{\phi}_{ch}\left(i-w+1, w\right)\right)\right|, \nonumber \\
    &\mathit{X^{\phi}_{\mathit{ch}}}\left(i-w+1,w\right) :=  \left\{X^{\phi}\left[i-w+1\right],..., X^{\phi}[i]\right\}. \nonumber
\end{eqnarray}
Note that for $x \geq |X^{\phi}|$, $X^{\phi}[x]$ does not exist. These cases are defined as follows:
\begin{equation*}
\scriptstyle \mathit{X^{\phi}_{\mathit{ch}}}(i-w +1,w) := \begin{cases} 
                                                        \scriptstyle  \{X^{\phi}[i-w+1], ..., X^{\phi}[|X^{\phi}|-1]\} & \scriptstyle (i-w+1 < |X^{\phi}| \leq i) \\
                                                        \scriptstyle\varnothing &\scriptstyle(|X^{\phi}| \leq i-w+1)
                                                        \end{cases}
\end{equation*}
In the remaining of this section, we first propose a baseline algorithm straightforwardly derived from Def.~\ref{def:ordering2}. 
Then, we propose an efficient algorithm for larger graphs, that provides the equivalent ordering result to the baseline's one but is significantly faster.

% \smallskip
% \noindent\textbf{Baseline Greedy Algorithm.}
\subsection{Baseline Greedy Algorithm}
% First, 

% The key idea of the algorithm is a \emph{greedy expansion}, which initially selects a single vertex at random and assigns orders to its neighbors from $0$.
% After that, the algorithm greedily selects a vertex from the frontiers of the already ordered part so that Equation~\ref{eq:genobj} is locally the minimal.
% Then, new orders are assigned to the neighbors of the selected vertex.
% The greedy selection and the ordering are executed iteratively until all edges are ordered.
% The vertex is greedily selected in such a way that Equation~\ref{eq:genobj} is locally minimized.
% Specifically, at the $i$-th iteration, the next edge, $e^{next}$, is selected in such a way that $F(i, e^{next})$ in Definition~\ref{def:ordering2} becomes local minimal.

Algorithm~\ref{alg:greedy} shows the baseline greedy algorithm.
% , which takes a graph, $G(V,E)$, and returns an ordered edge list, $X^{\phi}$.
The algorithm involves two main parts: greedy search and ordering.
Each vertex is greedily selected in Lines~\ref{alg:greedy_from}--\ref{alg:greedy_to}.
Then, its one-hop and two-hop neighbors are processed and appended to $X^{\phi}$ in Lines~\ref{alg:assign_from}--\ref{alg:assign_to}.

In the greedy search (Lines~\ref{alg:greedy_from}--\ref{alg:greedy_to}), the objective function (Eq.~\eqref{eq:genobj}) is calculated for every frontier vertex in the already ordered part (i.e., $V_{\mathit{rest}} \cap V(X^{\phi})$), and then, a vertex which minimizes Eq.~\eqref{eq:genobj}, $v_{\mathit{min}}$, is selected.

% \todo{check $E^{\phi}_{\mathit{ch}}(|X^{\phi}| - \delta$, this is $X^{\phi}$}
% \begin{highlight}
In the ordering part (Lines~\ref{alg:assign_from}--\ref{alg:assign_to}), the algorithm orders all of the $v_{\mathit{min}}$'s one-hop-neighbor edges.
Moreover, let $\delta$ be the range of two-hop-neighbor edges to be considered. 
A two-hop-neighbor edge of $v_{\mathit{min}}$ is considered for ordering if its destination, $w$, is involved in $V(X^{\phi}_{\mathit{ch}}(|X^{\phi}| - \delta, \delta))$.
Each neighbor edge is accessed in ascending order of the destination vertex id (we use the default vertex id of each dataset).
The reason why such a two-hop neighbor may improve partitioning quality is due to the well-known property that: \textit{if vertex $v$ and $u$ are included in partition $P$, then the vertex replications do not increase by adding $e_{v,u}$ to $P$}.
It is commonly used in the existing methods~\cite{joseph2012powergraph,Bourse:2014:BGE:2623330.2623660,Petroni:2015:HSP:2806416.2806424,Chen:2015:PDG:2741948.2741970,Zhang:2017:GEP:3097983.3098033,hanai2019distributed}.

For $\delta$, we choose the size of the smallest chunk (i.e., $\delta = 10^{0} \times \frac{|E|}{k_{\mathit{max}}}$) to maximize both the quality and performance as preliminary evaluated in Figure~\ref{fig:delta} (Replication Factor is the average value for $k = 4,8,16,32,64,128$. $k_{max} = 128$.). 
% Actually, by choosing a quite small $\delta$ (e.g., $\delta=1$), the set of two-hop vertices in $V(X^{\phi}_{\mathit{ch}}(|X^{\phi}| - \delta, \delta))$ shrinks, resulting in poor efficiency.
% Conversely, by choosing a quite big $\delta$ (e.g., $\delta=|E|$), the partitioning quality worsens. 
% This is due to the fact that in such a case, the two-hop-neighbor set becomes quite large, and it may include vertices from different chunks, i.e., the number of vertex replications increases.
% A reasonable value of $\delta$ could be the size of the smallest chunk (i.e., $\delta = \frac{|E|}{k_{\mathit{max}}}$), since in this way it is highly probable to choose two-hop vertices belonging to the same chunk (thus resulting in fewer vertex replications). 
% The experimental evaluation in Section~\ref{sec:evaluation} indeed confirms that the benefit of such a $\delta$ is two-fold (i.e., both high efficiency and quality).

\begin{figure}[h]
  \centering
  \subfigure{\includegraphics[width=.25\textwidth]{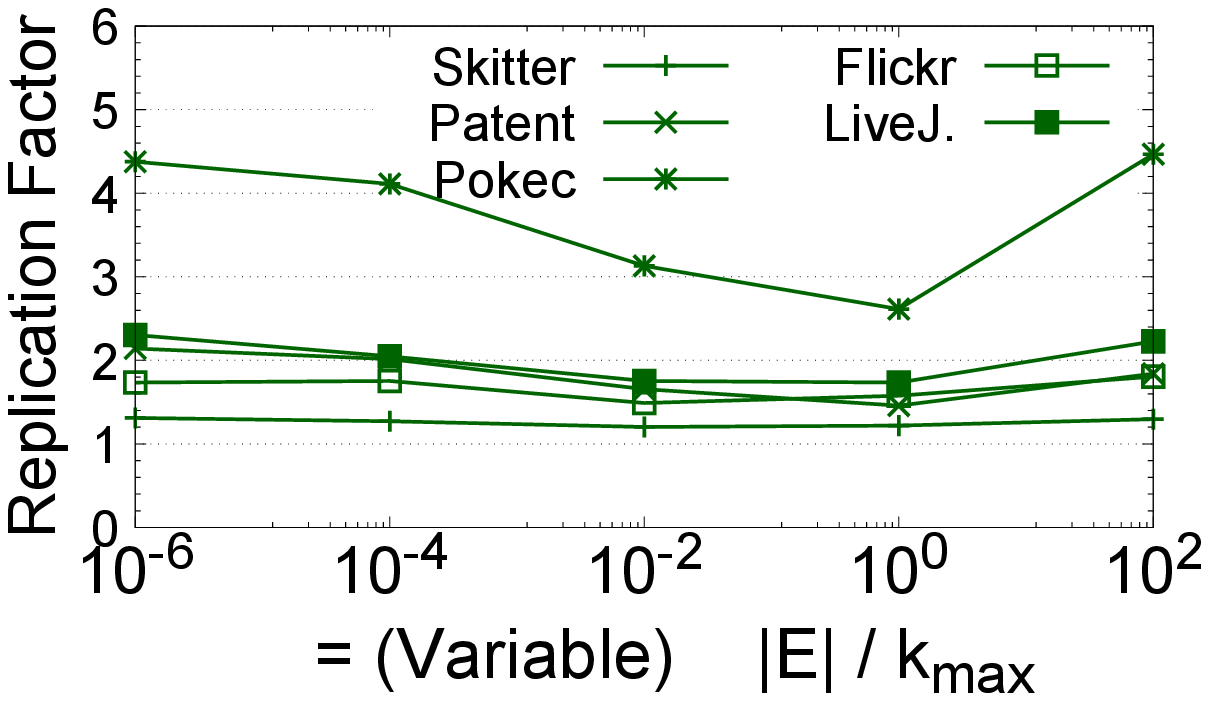}}%
  \subfigure{\includegraphics[width=.25\textwidth]{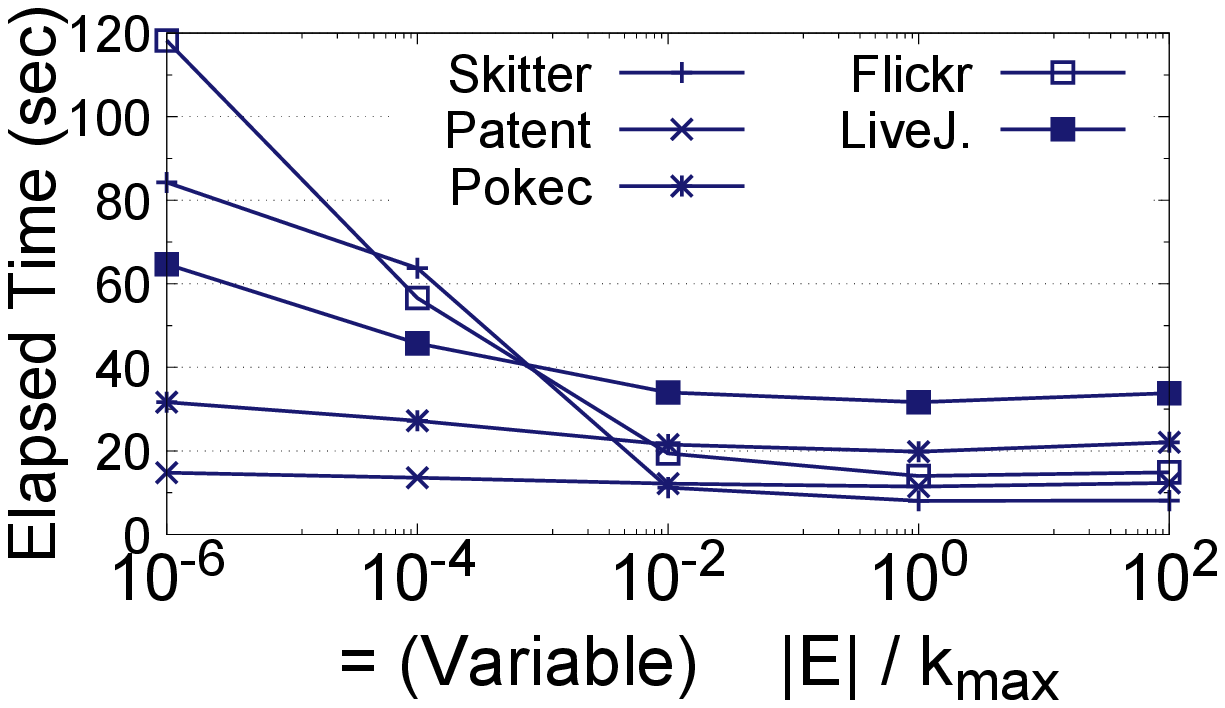}}%
  \caption{Quality and Performance for Different $\delta$.}\label{fig:delta}
\end{figure}

% Thus, by choosing a relatively small $\delta$, a two-hop vertex $w$ included in $V(X^{\phi}_{\mathit{ch}}(|X^{\phi}| - \delta, \delta))$ does hardly increase the number of vertex replications.
% In practice, we use the minimum size unit of the chunk for $\delta$, i.e., $\delta = \frac{|E|}{k_{\mathit{max}}}$, which is an appropriate choice for the real-world graphs used in our experimental analysis.

\begin{algorithm}[t]
\caption{Baseline Greedy Algorithm}\label{alg:greedy}
\SetKwInOut{Var}{Var}
\SetKwInOut{VarEmpty}{}
\SetKwInput{Input}{Input}
\SetKwInput{Output}{Output}
\Input{$G(V,E)$ -- Graph}
\Output{$X^{\phi}$ -- Ordered Edge List}
% \BlankLine
% \Var{$e$ -- Edge}
% \VarEmpty{$i$ -- Counter for ordering}
% \VarEmpty{$X$ -- Frontier Edges}
% \VarEmpty{$|E|$ -- \# of Entire Edges}
% \VarEmpty{$|E_{alloc}|$ -- \# of Entire Allocated Edges}
\SetKw{Break}{break}

\DontPrintSemicolon
\BlankLine

$X^{\phi}$ $\gets$ $\varnothing$; $i$ $\gets$ $0$; $V_{\mathit{rest}}$ $\gets$ $V$\;

% $v$ $\gets$ $V$\texttt{.RandomVertex()}\;
% \For{$e \in G\texttt{.NeighborsOf}(v)$}{
%   $E^{\phi}[i]$ $\gets$ $e$; $i$ $\gets$ $i + 1$ \;
% }

% \BlankLine
\While{$|V_{\mathit{rest}}|$ $\not=$ $0$} { \label{alg:line:loopfrom}
  /* Greedy Search From Frontier Vertices */ \;
  $v_{\mathit{min}}$ $\gets$ $\varnothing$; $F_{\mathit{min}}$ $\gets$ $\infty$ \label{alg:greedy_from}\;
  \uIf{$V_{\mathit{rest}} \cap V(X^{\phi}) = \varnothing$}{
    $v_{\mathit{min}}$ $\gets$ $V_{\mathit{rest}}$\texttt{.RandomVertex()}\;
  }\Else{
    \For{$v \in V_{\mathit{rest}} \cap V(X^{\phi})$} { \label{alg:line:innerloopfrom}
      $X'^{\phi}$ $\gets$ $X^{\phi}$ + $\left(N(v) \setminus X^{\phi}\right)$\;
      $F_v$ $\gets$ $\displaystyle \tfrac{1}{|V|} \sum_{k = k_{\mathit{min}}}^{k_{\mathit{max}}} \sum_{i=0}^{|E|-1} \scriptstyle f_k\left(X'^{\phi}, i, \left\lfloor \tfrac{|E|+\texttt{ID2P}_{k}(i)}{k}\right\rfloor\right)$ \label{alg:line:f} \; %\label{alg:line:innermostloopto}
      \lIf{$F_v$ $<$ $F_{\mathit{min}}$}{
        $F_{\mathit{min}}$ $\gets$ $F_v$; $v_{\mathit{min}}$ $\gets$ $v$
      } \label{alg:line:innerloopto}
    }
  }\label{alg:greedy_to}

  \BlankLine
  /* Assign New Edge Order */\;
  \For{$e_{v_{\mathit{min}},u} \in N(v_{\mathit{min}}) \setminus X^{\phi}$} {\label{alg:assign_from}
    $X^{\phi}[i]$ $\gets$ $e_{v_{\mathit{min}},u}$; $i$ $\gets$ $i + 1$ \;
    
    \For{$e_{u,w}\in N(u) \setminus X^{\phi}$} {
        \If{$w \in V(X^{\phi}_{\mathit{ch}}(|X^{\phi}| - \delta, \delta))$}{
            $X^{\phi}[i]$ $\gets$ $e_{u,w}$; $i$ $\gets$ $i + 1$ \;
        }
    }
  }\label{alg:assign_to}

  \BlankLine
  $V_{\mathit{rest}}$ $\gets$ $V_{\mathit{rest}} \setminus \{v_{\mathit{min}}\}$ \;\label{alg:line:loopto}
}
\Return $X^{\phi}$\;
\end{algorithm}

% \todo{reason why it assigned two-hop neighbors in line 17 -- 19}
% \todo{assign two hop neighbor where (the latest id minus the dst id) is within $|E| / k_{max}$.}
% \todo{check after evaluation}
% \todo{parameterize to use $\alpha$. Then, in lemma, $\alpha$ is zero.}

% \smallskipThe efficiency of the baseline algorithm is as follows:
\begin{theorem}[Efficiency of Baseline Algorithm] \label{theorem:algo1complexity} \ \\
Efficiency of Algorithm~\ref{alg:greedy} is $O\left(\tfrac{k^2_{\mathit{max}}|E|^{2}\cdot|V|^{2}}{k_{\mathit{min}}}\right)$, where $k_{\mathit{max}}$ is much larger than $k_{\mathit{min}}$ s.t. $k_{\mathit{max}} - k_{\mathit{min}} \approx k_{\mathit{max}}$.
\end{theorem}
\begin{proof}
The outermost loop (Lines \ref{alg:line:loopfrom}--\ref{alg:line:loopto}) computes $O(|V|)$ iterations. 
In each iteration, the inner loop from Line \ref{alg:line:innerloopfrom} to \ref{alg:line:innerloopto} also computes $O(|V|)$ iterations.
Moreover, at Line~\ref{alg:line:f}, the summation requires $O(|E|\cdot(k_{\mathit{max}} - k_{\mathit{min}}))$ while $f_k\left(X^{\phi}, i, \tfrac{|E|}{k}\right)$ requires $O\left(\tfrac{k_{\mathit{max}}|E|}{k_{\mathit{min}}}\right)$.
Therefore, in total, it requires \\ $O\left(\tfrac{|V|^{2} \cdot |E| (k_{\mathit{max}} - k_{\mathit{min}}) \cdot k_{\mathit{max}} |E|}{k_{\mathit{min}}}\right) = O\left(\tfrac{k^2_{\mathit{max}}|E|^{2}\cdot|V|^{2}}{k_{\mathit{min}}}\right)$
under the assumption that $k_{\mathit{max}} - k_{\mathit{min}} \approx k_{\mathit{max}}$. \qed
\end{proof}

% \begin{proof}
% The outermost loop (Lines \ref{alg:line:loopfrom}--\ref{alg:line:loopto}) computes $O(|V|)$ iterations. 
% In each iteration, the inner loop from Line \ref{alg:line:innerloopfrom} to \ref{alg:line:innerloopto} also computes $O(|V|)$ iterations.
% Moreover, at Line~\ref{alg:line:f}, the summation requires $O(|E|\cdot(k_{\mathit{max}} - k_{\mathit{min}}))$ computational time while $f_k\left(X^{\phi}, i, \tfrac{|E|}{k}\right)$ requires $O\left(\tfrac{k_{\mathit{max}}|E|}{k_{\mathit{min}}}\right)$.
% Therefore, in total, it requires \\ $O\left(\tfrac{|V|^{2} \cdot |E| (k_{\mathit{max}} - k_{\mathit{min}}) \cdot k_{\mathit{max}} |E|}{k_{\mathit{min}}}\right) = O\left(\tfrac{k^2_{\mathit{max}}|E|^{2}\cdot|V|^{2}}{k_{\mathit{min}}}\right)$
% under the assumption that $k_{\mathit{max}} - k_{\mathit{min}} \approx k_{\mathit{max}}$. 
% \end{proof}

\begin{algorithm}[t]
\caption{PQ-based Fast Algorithm}\label{alg:greedy-opt}
\SetKwInOut{Var}{Var}
\SetKwInOut{VarEmpty}{}
\SetKwInput{Input}{Input}
\SetKwInput{Output}{Output}
\Input{$G(V,E)$ -- Graph}
\Output{$X^{\phi}$ -- Ordered Edge List}
% \BlankLine
% \Var{$e$ -- Edge}
% \VarEmpty{$i$ -- Counter for ordering}
% \VarEmpty{$X$ -- Frontier Edges}
% \VarEmpty{$|E|$ -- \# of Entire Edges}
% \VarEmpty{$|E_{alloc}|$ -- \# of Entire Allocated Edges}
\SetKw{Break}{break}

\DontPrintSemicolon
\BlankLine

$X^{\phi}$ $\gets$ $\varnothing$; $i$ $\gets$ $0$; $V_{\mathit{rest}}$ $\gets$ $V$; $\mathit{PQ}$ $\gets$ $\varnothing$ \;
$D[v]$ $\gets$ $|N(v)|$, $M[v]$ $\gets$ $0$ for each $v \in V$\;

\BlankLine
\While{$|V_{\mathit{rest}}|$ $\not=$ $0$} { \label{alg:greedy-opt:line:loopfrom}
  $v_{\mathit{min}}$ $\gets$ $\varnothing$\;
  \lIf{$\mathit{PQ} = \varnothing$}{
    $v_{\mathit{min}}$ $\gets$ $V_{\mathit{rest}}$\texttt{.RandomVertex()}
  }\lElse{
    $v_{\mathit{min}}$ $\gets$ $\mathit{PQ}$\texttt{.dequeue()}
  }

  \BlankLine
%   /* Assign Edge ID */\;
  \For{$e_{v_{\mathit{min}},u} \in N(v_{\mathit{min}}) \setminus X^{\phi}$} {\label{alg:greedy-opt:line:innerloopfrom}
    $X^{\phi}[i]$ $\gets$ $e_{v_{\mathit{min}},u}$; $i$ $\gets$ $i + 1$ \;
    $D[u]$ $\gets$ $D[u] - 1$; $M[u]$ $\gets$ $i$ \; \label{alg:line:updatem0}

    \For{$e_{u,w}\in N(u) \setminus X^{\phi}$} { \label{alg:greedy-opt:line:innermostloopfrom}
        \If{$w \in V(X^{\phi}_{\mathit{ch}}(|X^{\phi}| - \delta, \delta))$}{
            $X^{\phi}[i]$ $\gets$ $e_{u,w}$; $i$ $\gets$ $i + 1$ \;
            $D[u]$ $\gets$ $D[u] - 1$; $D[w]$ $\gets$ $D[w] - 1$ \;
            $M[w]$ $\gets$ $i$; $M[u]$ $\gets$ $i$ \; \label{alg:line:updatem1}
            $\mathit{PQ}$\texttt{.update($D[w]$,$M[w]$,$w$)}\; \label{alg:line:updatepq}
        }
    }\label{alg:greedy-opt:line:innermostloopto}

    \lIf{$u \not\in \mathit{PQ}$} {
        $\mathit{PQ}$\texttt{.enqueue($D[u]$,$M[u]$,$u$)}
    }\lElse {
        $\mathit{PQ}$\texttt{.update($D[u]$,$M[u]$,$u$)}
    }
  }\label{alg:greedy-opt:line:innerloopto}

  \BlankLine
  $V_{\mathit{rest}}$ $\gets$ $V_{\mathit{rest}} \setminus \{v_{\mathit{min}}\}$ \;\label{alg:greedy-opt:line:loopto}
}
\Return $X^{\phi}$ \;
\end{algorithm}

\subsection{Fast Algorithm Based on Priority Queue}

We propose an efficient greedy algorithm based on the priority queue, which provides the equivalent result to the baseline algorithm (Algorithm~\ref{alg:greedy}).
Although Algorithm~\ref{alg:greedy} provides an approximate solution of the NP-hard problem, its computational cost 
% $O\left(\tfrac{k^2_{\mathit{max}}|E|^{2}\cdot|V|^{2}}{k_{\mathit{min}}}\right)$,
(Theorem~\ref{theorem:algo1complexity}) 
is still high as the real-world graph is typically large, including billions of elements.

Algorithm~\ref{alg:greedy-opt} shows the efficient algorithm.
Overall, its computation is the same as Algorithm~\ref{alg:greedy}, which begins with a random vertex. Then, the algorithm iteratively and greedily expands the ordered parts until all the edges are ordered.

The key difference from Algorithm~\ref{alg:greedy} is that Algorithm~\ref{alg:greedy-opt} utilizes a priority queue, $\mathit{PQ}$, to evaluate the objective function (Eq.~\eqref{eq:genobj}) instead of calculating the equation for every frontier vertex in $V_{\mathit{rest}} \cap V(X^{\phi})$.
Specifically, $\mathit{PQ}$ uses a priority represented as follows:
\begin{equation}
p(v) := \alpha \cdot D[v] - \beta \cdot M[v] \label{eq:priority}, 
\end{equation}
where $\alpha := \sum_{k = k_{\mathit{min}}}^{k_{\mathit{max}}}\left\lfloor \tfrac{|E|}{k}\right\rfloor$ and $\beta := k_{\mathit{max}} - k_{\mathit{min}}$ are computed in advance of the greedy expansion.
The frontier vertices are sorted in the ascending order by $p$.
$D[v]$ $(v \in V)$ is the $v$'s degrees for the rest of the edges (i.e., $D[v] := |N(v) \setminus X^{\phi}|$).
$M[v]$ stores the latest order of an edge which involves $v$.
$M[v]$ is updated each time a new edge order is assigned (Line~\ref{alg:line:updatem0} and Line~\ref{alg:line:updatem1} in Algorithm~\ref{alg:greedy-opt}).

\begin{figure}[t]
  \centering
   \includegraphics[width=\columnwidth]{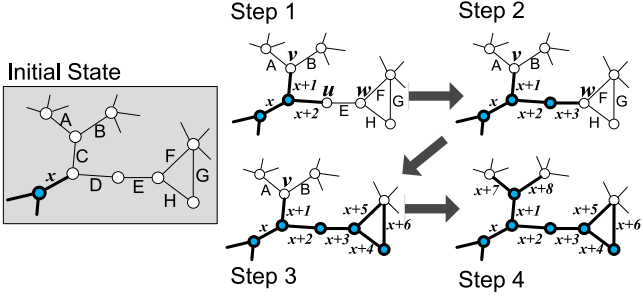}
  \caption{Greedy Expansion in Fast Algorithm. Edge~A,B,..,H is ordered into $x+1, x+2, ... x+8$.}\label{fig:expansion}
\end{figure}

Figure~\ref{fig:expansion} shows an example of the greedy expansion. 
After assigning the order to Edge~C and Edge~D, two frontier vertices, $v$ and $u$, exist in the graph. The algorithm selects $u$ and assigns $x+3$ to Edge~E because $p(u) < p(v)$.
% whose priority score is $\left\lfloor \tfrac{|E|}{k_{\mathit{max}}}\right\rfloor \times 1 - (x+2)$ (whereas $v$'s score is $\left\lfloor \tfrac{|E|}{k_{\mathit{max}}}\right\rfloor \times 2 - (x+1)$) 
Then, $v$ and $w$ become the frontier. Then, $w$ is selected because $p(w) < p(v)$.
The algorithm assigns $x+4$, $x+5$, $x+6$ to Edge~F,~H,~G, respectively.
Finally, $v$ is selected. Then, Edge A and B are ordered to $x+7$ and $x+8$, respectively.

The equivalence of Algorithm~\ref{alg:greedy} and Algorithm~\ref{alg:greedy-opt} is established by the following lemma.
It shows that the calculation of Eq.~\eqref{eq:genobj} can be replaced into $PQ$ based on $p$ (Eq.~\eqref{eq:priority}):
\begin{lemma}\label{lm:priority}
Suppose $|E|$ is much larger than $k_{max}$ such that $w:= \left\lfloor \tfrac{|E|}{k}\right\rfloor = \left\lfloor \tfrac{|E|+\texttt{ID2P}_{k}(\cdot)}{k}\right\rfloor$ and $D[v] < \tfrac{|E|}{k_{max}}$ for $\forall v \in V$.

% \begin{center}
Then, $\forall v, u \in V_{\mathit{rest}} \cap V(X^{\phi})$, \ \ $p(v) > p(u) \Rightarrow F_v >  F_u$,
% \end{center}
where $F_v$ and $F_u$ are the value of Eq.~\eqref{eq:genobj} for $X^{\phi} + N(v)$ and $X^{\phi} + N(u)$ respectively, as shown in Line~\ref{alg:line:f} of Algorithm~\ref{alg:greedy}.
\end{lemma}
\begin{proof}
Suppose $\mathit{Xv}^{\phi} := X^{\phi} + N(v)$, $\mathit{Xu}^{\phi} := X^{\phi} + N(u)$.
% and $w := \left\lfloor\tfrac{|E|}{k}\right\rfloor$ $\left(=\left\lfloor \tfrac{|E|+\texttt{ID2P}_{k}(\cdot)}{k}\right\rfloor\right)$.

\allowdisplaybreaks
\begin{eqnarray}
 && F_v > F_u \ \ \Leftrightarrow \ \ F_v - F_u > 0\nonumber \\
&\Leftrightarrow& \displaystyle \sum_{k = k_{\mathit{min}}}^{k_{\mathit{max}}} \sum_{i=0}^{|E|-1} \left\{ f\left(\mathit{Xv}^{\phi}, i, w \right) - f\left(\mathit{Xu}^{\phi}, i, w \right) \right\} {\displaystyle > 0} \nonumber  \\
% &\Leftarrow& \displaystyle \sum_{i=0}^{|E|-1} \scriptstyle \left\{ f\left(\mathit{Xv}^{\phi}, i, w \right) - f\left(\mathit{Xu}^{\phi}, i, w \right) \right\} {\displaystyle > 0} \nonumber  \\
&\Leftarrow& \displaystyle \sum_{k = k_{\mathit{min}}}^{k_{\mathit{max}}} \sum_{i=0}^{|E|-1} \scriptstyle \left\{ \left|V\left(\mathit{Xv}^{\phi}_{ch}\left(i-w+1, w \right)\right)\right| - \left|V\left(\mathit{Xu}^{\phi}_{ch}\left(i-w+1, w\right)\right)\right|\right\} \nonumber \\
&& {\displaystyle > 0}  \nonumber \\
&\Leftrightarrow& \displaystyle \displaystyle \sum_{k = k_{\mathit{min}}}^{k_{\mathit{max}}} \sum_{i=|X|^{\phi}}^{|E|-1} \left\{\Delta V(v,i) - \Delta V(u,i) \right\} > 0, \label{eq:ff} 
\end{eqnarray} 
where $\Delta V(v,i) := |V(\mathit{Xv}^{\phi}_{ch}(\scriptstyle{i-w+1, w}\displaystyle))|- |V(X^{\phi}_{ch}(\scriptstyle{i-w+1, w}\displaystyle))|$. 

Next, we will calculate $\Delta V(v,i)$ for $i \geq |X^{\phi}|$.
Intuitively, $\Delta V(v,i)$ means the number of additional replicated vertices in a chunk when we select $v$ to expand the ordered edges.
For each chunk determined by $i$, each additional replicated vertex comes from $v$ or $N(v)$.
Thus, $\Delta V(v,i)$ can be represented by the sum of two functions:
% et $j := i - |X^{\phi}|$ ($j \geq 0$). Then, 
% The following equation is established:
% $\displaystyle \sum_{i=|X^{\phi}|}^{|E|-1} \left|V\left(\mathit{Xv}_{ch}\left(i-w+1, w\right)\right)\right|$ for $k \in  [k_{\mathit{min}}, k_{\mathit{max}}]$.
% For $i \geq |X^{\phi}|$, $\left|V\left(\mathit{Xv}_{ch}\left(i-w+1, w\right)\right)\right|$ can be represented as follows:
\begin{eqnarray}
\Delta V(v,i) = \chi(i) + n(i), \nonumber
% \displaystyle &&\left|V\left(\mathit{Xv}_{ch}\left(i-w+1, w\right)\right)\right| - |V(X^{\phi}_{ch}(i-w+1, w))|\nonumber \\
% &&=\delta(i) + n(i) \nonumber
\end{eqnarray}
where $\chi(i)$ is the number of replicated vertices caused by $v$; $n(i)$ is caused by $N(v)$.

First, $\chi(i)$ is the indicator function.
If $X^{\phi}_{ch}(i-w+1, w)$ already involves $v$, then the number of replicated vertices does not increase due to the additional $v$. Therefore, $\chi(i)$ is 0.
Specifically, this case appears if $i > M[v] + w$, because $X^{\phi}_{ch}(i-w+1, w)$ involves an edge $e$ whose order is $M[v]$ (i.e., $\phi(e) = M[v]$).
Otherwise, $v$'s replication is newly added to the chunk $Xv^{\phi}_{ch}(i-w+1, w)$, and thus $\chi(i)$ is 1.
Therefore, $\chi(i)$ is represented as follows:
\begin{eqnarray}
\chi(i)  = \begin{cases} 1 & \text{if } i \in [M[v] + w, |X^{\phi}| + D[v] + w) \\  0 & \text{if } i \not\in [M[v] + w, |X^{\phi}| + D[v] + w) \end{cases} \nonumber
\end{eqnarray}
where we also consider a case that $i$ is larger so that $Xv^{\phi}_{ch}(i-w+1, w)$ is empty. In this case, $\chi(i)$ is obviously 0.

Second, $n(i)$ is the number of the additional vertices derived from $N(v)$.
Its value can be represented as follows:
\begin{eqnarray}
n(i) = \begin{cases} 
i - |X^{\phi}| & (|X^{\phi}| \leq i < |X^{\phi}| + D[v]) \\
D[v] & (|X^{\phi}| + D[v] \leq i < |X^{\phi}| + w) \\
\scriptstyle D[v] - i + |X^{\phi}| + w & (|X^{\phi}| + w \leq i < |X^{\phi}| + D[v] + w) \\
0 & (|X^{\phi}| + D[v] + w < i) \\
\end{cases} \nonumber
\end{eqnarray}

Figure~\ref{fig:lemma} shows an example of these cases.
Suppose $v$ is selected in the greedy algorithm and new edge orders are assigned to $v$'s neighbor edges, $e_{v,u_0}$, $e_{v,u_1}$, $e_{v,u_2}$, $e_{v,u_3}$, and $e_{v,u_4}$.
Then, if $|X^{\phi}| \leq i < |X^{\phi}| + D[v]$, a part of $N(v)$ are added (e.g., $u_0, u_1, u_2$ in Figure~\ref{fig:lemma}).
If $|X^{\phi}| + D[v] \leq i < |X^{\phi}| + w$, all vertices in $N(v)$ are added (e.g., $u_0, u_1, u_2, u_3, u_4$ in Figure~\ref{fig:lemma}).
If $|X^{\phi}| + w \leq i < |X^{\phi}| + D[v] + w$, also a part of $N(v)$ are added (e.g., $u_2, u_3, u_4$ in Figure~\ref{fig:lemma}).
If $|X^{\phi}| + D[v] + w \leq i$, then $Xv^{\phi}_{ch}(i-w+1, w)$ involves no vertices.

\begin{figure}[t]
  \centering
   \includegraphics[width=\columnwidth]{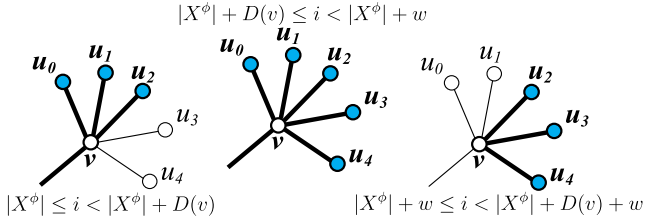}
  \caption{The value of $n(i)$: \# of the additional vertices derived from $N(v) = \{u_0, u_1, ..., u_4\}$. Blue vertices are the additional when $v$ is selected for the expansion.}\label{fig:lemma}
\end{figure}

Therefore,
\begin{eqnarray}
\displaystyle \sum_{\mathclap{i=|X^{\phi}|}}^{\mathclap{|E|-1}} \Delta V(v,i) =&\displaystyle \sum_{\mathclap{i=|X^{\phi}|}}^{\mathclap{|E|-1}} \chi(i)+ \ \ \sum_{\mathclap{i=|X^{\phi}|}}^{\mathclap{|X^{\phi}| + D[v] -1}} \left\{i - |X^{\phi}| \right\} \ + \  \sum_{\mathclap{i = |X^{\phi}| + D[v]}}^{\mathclap{|X^{\phi}| + w - 1}} D[v]  \nonumber \\
&\displaystyle + \sum_{\mathclap{i = |X^{\phi}| + w}}^{\mathclap{|X^{\phi}| + w + D[v] - 1}}\left\{ D[v] - i + |X^{\phi}| + w \right\} \nonumber \\
% &=& \tfrac{1}{2}D(v) \left(2|X^{\phi}| + D(v) - 1\right) - D(v)|X^{\phi}| \nonumber \\
% && \ \ + \left(w - D(v)\right) D(v)  + D(v)\left(D(v) + |X^{\phi}| + w\right) \nonumber \\
% && \ \ - \tfrac{1}{2}D(v)\left(2|X^{\phi}| + 2w + D(v) -1\right) \nonumber \\ 
% && \ \ + |X^{\phi}| + D(v) - M(v) \nonumber \\
=& w D[v] + |X^{\phi}| + D[v] - M[v] \nonumber
\end{eqnarray}
Let $\Delta D := D[v] -  D[u]$ and $\Delta M := M[v] -  M[u]$.
\begin{eqnarray}
\displaystyle \sum_{\mathclap{i=|X^{\phi}|}}^{\mathclap{|E|-1}} \left\{\Delta V(v,i) -  \Delta V(u,i) \right\} = w \Delta D +  \Delta D - \Delta M\nonumber \\
\sim w \Delta D - \Delta M \ \ (\because w > \tfrac{|E|}{k_{max}} \gg 1) \nonumber
\end{eqnarray}
Therefore,
\begin{eqnarray*}
&& p(v) > p(u) \\
&\Leftrightarrow& \alpha \cdot D[v] - \beta \cdot M[v] > \alpha \cdot D[u] - \beta \cdot M[u] \\
&\Leftrightarrow& \sum_{k=k_{min}}^{k_{max}} \left(w \cdot D[v] - M[v]\right) > \sum_{k=k_{min}}^{k_{max}} \left(w \cdot D[u] - M[u]\right)\\
&\Leftrightarrow& \sum_{k=k_{min}}^{k_{max}}\left(w \Delta D - \Delta M\right) > 0 \\ 
&\Leftrightarrow& \displaystyle \sum_{k=k_{min}}^{k_{max}} \sum_{i=|X^{\phi}|}^{|E|-1} \left\{\Delta V(v,i) -  \Delta V(u,i) \right\} > 0 \\
&\Rightarrow& F_v > F_u. \ \ \ (\because \eqref{eq:ff})
\end{eqnarray*}
% In addition, if $D[v] = D[u]$, then 
% \begin{eqnarray}
% &&\displaystyle \sum_{i=0}^{|E|-1} \scriptstyle \left|V\left(\mathit{Xv}_{ch}\left(i-w+1, w\right)\right)\right| - \displaystyle\sum_{i=0}^{|E|-1} \scriptstyle\left|V\left(\mathit{Xu}_{ch}\left(i-w+1, w\right)\right)\right| \nonumber \\
% &=&\displaystyle \tfrac{1}{2}\Delta D^{2} + \left( |X^{\phi}| + \tfrac{|E|}{k} + \tfrac{1}{2} \right) \Delta D - \Delta M \nonumber \\
% &=& - \Delta M \nonumber
% \end{eqnarray}
% Therefore, if $D[v] = D[u]$, then
% \begin{eqnarray*}
% M[u] > M[v] \Rightarrow F(v) > F(u).
% \end{eqnarray*}
Thus, the lemma is proved. \qed
\end{proof}

\smallskip
Algorithm~\ref{alg:greedy-opt} significantly reduces the time complexity.
\begin{theorem}[Efficiency of Fast Algorithm] \label{theorem:algo2complexity}
Suppose $\mathit{PQ}$ be a standard priority-queue implementation, where dequeue, update or enqueue can be operated in $O(\log n)$. 
Then, time complexity of Algorithm~\ref{alg:greedy-opt} is \\
$O(d_{\mathit{max}}^2 |V| \log|V|)$, where $d_{\mathit{max}}$ is the maximum degree.
% The efficiency is improved from $O\left(\tfrac{k^2_{\mathit{max}}|E|^{2}\cdot|V|^{2}}{k_{\mathit{min}}}\right)$ of Algorithm\ref{alg:greedy} to $O(d_{\mathit{max}}^2 |V| \log|V|)$.
\end{theorem}

\begin{proof}
The outermost loop (Lines~\ref{alg:greedy-opt:line:loopfrom}~--~\ref{alg:greedy-opt:line:loopto}) requires $O(|V|)$ for each vertex.
Then, the inner loop (Lines~\ref{alg:greedy-opt:line:innerloopfrom}~--~\ref{alg:greedy-opt:line:innerloopto}) needs $O(d_{\mathit{max}})$ for each neighbors of $v_{\mathit{min}}$.
Finally, the innermost loop (Lines~\ref{alg:greedy-opt:line:innermostloopfrom}~--~\ref{alg:greedy-opt:line:innermostloopto}) requires $O(d_{\mathit{max}})$ for each neighbors and $O(\log |V|)$ for updating $PQ$ at Line~\ref{alg:line:updatepq}.
To sum up, the total time complexity of Algorithm~\ref{alg:greedy-opt} is $O(d_{\mathit{max}}^2 |V| \log|V|)$. \qed
\end{proof}

\section{Theoretical Analysis}\label{sec:upperbound}
In this section, we provide a theoretical analysis for the upper bound of the partitioning quality achieved by our method.
% Our method provides the upper bound in partitioning quality. 
% Once edge data are ordered by the graph edge ordering algorithm, edge partitions generated via the chunk-based edge partitioning method have a theoretical upper bound in partitioning quality.
Our analysis focuses on the observation that in typical practical situations, the minimum partition size is much larger than the maximum degree of the graph. In this case, the number of new ordered edges in each iteration of Algorithm~\ref{alg:greedy-opt} becomes smaller than the smallest partition size. Based on this fact, we formulate the following  theorem.

\begin{theorem}[Upper Bound of Partitioning Quality] \label{theorem:upperbound}
Consider that in each iteration of Algorithm~\ref{alg:greedy-opt} (Lines \ref{alg:greedy-opt:line:loopfrom}--\ref{alg:greedy-opt:line:loopto}) the number of new ordered edges is smaller than the smallest partition size and that $\delta = \left\lfloor\frac{|E|}{k_{max}}\right\rfloor - 1$.

% $\frac{|E|}{k_{max}} \gg  d_{max}$  and  holds, where $d_{max}$ is the max degree, $\frac{|E|}{k_{max}}$ is the the smallest partition size

% for each iteration in such a way that during any one iteration of Algorithm~\ref{alg:greedy-opt} (Lines \ref{alg:greedy-opt:line:loopfrom}--\ref{alg:greedy-opt:line:loopto}),  ($\approx \frac{|E|}{k_{max}}$) and that $\delta = \left\lfloor\frac{|E|}{k_{max}}\right\rfloor - 1$.}

Let $E$ be ordered by Algorithm~\ref{alg:greedy-opt} and then partitioned into $k$ parts, $\mathcal{E}_{k}$, via chunk-based partitioning. 
Then, the replication factor, $RF_k$, for the $k$ edge partitions has an upper bound as follows:
\begin{equation*}
RF_k := \sum_{p=0}^{k-1}\frac{|V(\mathcal{E}_{k}[p])|}{|V|} \leq \frac{|V| + |E| + k}{|V|} 
\end{equation*}
\end{theorem}

\begin{proof}
Assume that $k$ is given in advance.
We consider a new partitioning algorithm based on the ordering algorithm (Algorithm~\ref{alg:greedy-opt}) and conversion function $\texttt{ID2P}_{k}(i)$ (Algorithm~\ref{alg:id2partition_app}) as follows: 
\begin{itemize}
\item{Initially, the $k$ edge partitions, $\mathcal{E}_{k}$, are empty.}
\item{Run Algorithm~\ref{alg:greedy-opt} and insert $i$-th ordered edge \\ to $\mathcal{E}_{k}[\texttt{ID2P}_{k}(i)]$.}
\end{itemize}
In the new partitioning algorithm, the edge partitions are incrementally determined from $\mathcal{E}_{k}[0]$, $\mathcal{E}_{k}[1]$, ..., to $\mathcal{E}_{k}[k-1]$. 
Obviously, the partitioning results obtained by the new partitioning algorithm is the same as the ones by our proposed method (i.e., completing Algorithm~\ref{alg:greedy-opt} \emph{before} the chunk-based edge partitioning for $k$.)

Let $t$ be an iteration counter for Lines \ref{alg:greedy-opt:line:loopfrom}--\ref{alg:greedy-opt:line:loopto} of Algorithm~\ref{alg:greedy-opt}, and $\Phi(t)$ be a potential function over $t$ defined as follows:
\begin{equation*}
    \Phi(t) := |V_{rest}(t)| + |E_{rest}(t)| + p_{rest}(t) + \sum_{p=0}^{k-1}|V(\mathcal{E}_{k}(t)[p])|, 
\end{equation*}
where $V_{rest}(t)$ is a set of vertices adjacent to at least one non-ordered edge; $E_{rest}(t)$ is a set of non-ordered edges at $t$; $p_{rest}(t)$ is the number of partitions which still have spaces to insert edges; $\mathcal{E}_{k}(t)$ is a set of edge partitions at $t$.  
% just before the last insertion of $t$-th iteration (e.g., $p_{rest}$ will change at $t+1$ if the edge partition $\mathcal{E}_{k}[\cdot]$ is filled at the last insertion in $t$-th iteration.)}.

Suppose the ordering algorithm terminates at $T$. 
We will show that (a) $\Phi(0) = |V| + |E| + k$, (b) $\Phi(T) = \sum_{p=0}^{k-1} V(\mathcal{E}_{k}[p])$, and (c) $\Phi(0) \geq \Phi(T)$.
The first two equations (a) and (b) are obvious, i.e., 
\begin{align*}
    \Phi(0) &:= |V_{rest}(0)| + |E_{rest}(0)| \\
    &\ \ \ \ \ \ \ \ + p_{rest}(0) + \sum_{p=0}^{k-1}|V(\mathcal{E}_{k}(0)[p])| \\
    &= |V|+|E|+k \\
    \Phi(T) &:= |V_{rest}(T)| + |E_{rest}(T)| \\ 
    &\ \ \ \ \ \ \ \ + p_{rest}(T) + \sum_{p=0}^{k-1}|V(\mathcal{E}_{k}(T)[p])| \\
    &= \sum_{p=0}^{k-1}|V(\mathcal{E}_{k}[p])|
\end{align*}

For (c), we will show $\Phi(t) - \Phi(t-1) \geq 0$ for $0 < t \leq T$.
Let $\Delta\Phi := \Phi(t) - \Phi(t-1)$, $\Delta V_{rest} := |V_{rest}(t)| - |V_{rest}(t-1)|$, $\Delta E_{rest} := |E_{rest}(t)| - |E_{rest}(t-1)|$, $\Delta p_{rest} := p_{rest}(t) - p_{rest}(t-1)$, and $\Delta V(\mathcal{E}) := \sum_{p=0}^{k-1}|V(\mathcal{E}_{k}(t)[p])| - \sum_{p=0}^{k-1}|V(\mathcal{E}_{k}(t-1)[p])|$.

For $t$-th iteration where $v_{min}$ is selected for expansion, we define the number of $v_{min}$'s one-hop neighbor edges which are ordered at $t$ as $n_{one}$ and the number of $v_{min}$'s two-hop neighbor edges which are ordered at $t$ as $n_{two}$.
For example in Figure~\ref{fig:upperbound}, $v_{min}$ is selected and the edges are ordered from $x$ to $x+9$. Here, $n_{one} = 3$ (Edges $x$, $x+5$, $x+9$) and $n_{two} = 7$ (Edges $x+1, x+2, x+3, x+4, x+6, x+7, x+8$).

\begin{figure}[h]
  \centering
   \includegraphics[width=1.0\columnwidth]{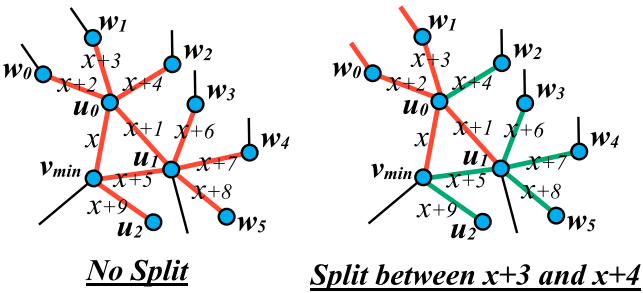}
  \caption{One Iteration for Ordering.}\label{fig:upperbound}
\end{figure}

Then, $\Delta V_{rest} \leq -1$ because all of $v_{min}$'s neighbor edges are assigned at the iteration.
$\Delta E_{rest} = -(n_{one} + n_{two})$ from the definition.

$\Delta p_{rest} = 0$ if all the ordered edges during the iteration are inserted to the same partition as the previous iteration and the partition has still have free space (Case A).
Otherwise (Case B), i.e., if the partitioning set becomes full during the current iteration, then $\Delta p_{rest} = -1$. 
Note that during an iteration of the algorithm, there cannot be more than one partitioning set that becomes full due to the assumption that the number of new assigned edges in each iteration is smaller than the smallest partition size.

% This is due to the assumption that the number of new assigned edges during any iteration is smaller than the smallest partitioning size.
% Under the assumption, $p_{rest}$ is changed at most one time during an iteration, and thus $\Delta p_{rest} = -1$ in Case B.}

For $\Delta V(\mathcal{E})$, we consider Case A and B as well. 
In Case A, $\Delta V(\mathcal{E}) \leq 1 + n_{one} + n_{two}$ because $v_{min}$ may be newly inserted; $n_{one}$ vertices and up to $n_{two}$ vertices may be to the current partition.
For example in Figure~\ref{fig:upperbound}, $u_i (i = 0,1,2)$ are newly inserted as $n_{one}$ and $w_i (i = 0,..,5)$ are as $n_{two}-1 (\leq n_{two})$.

In Case B, $\Delta V(\mathcal{E})$ satisfies the following equation:
\begin{equation}
    \Delta V(\mathcal{E}) \leq 2 + (n_{one} + 2) + (n_{two} - 2) \label{eq:caseb}.
\end{equation}
As explained above in Case B, there cannot be more than one partitioning set that becomes full.
This is translated into having up to two partitions at an iteration.
Thus, let the two partitions be $\mathcal{E}_k[p_0]$ and $\mathcal{E}_k[p_1]$, where each edge is inserted to $\mathcal{E}_k[p_0]$ at first and then $\mathcal{E}_k[p_1]$ after splitting.

The first ``$2$'' in Eq.\eqref{eq:caseb} means that $v_{min}$ may be inserted to up to two partitions ($\mathcal{E}_k[p_0]$ and $\mathcal{E}_k[p_1]$).
The second ``$(n_{one} + 2)$'' in Eq.\eqref{eq:caseb} means that $n_{one}$ vertices are inserted to either $\mathcal{E}_k[p_0]$ or $\mathcal{E}_k[p_1]$. In addition, up to two of $n_{one}$ vertices may be inserted to both $\mathcal{E}_k[p_0]$ and $\mathcal{E}_k[p_1]$. This is because there may be up to two partitions at an iteration.

Let the splitting point be between $i$ and $i+1$. 
The last ``$(n_{two} - 2)$'' in Eq.\eqref{eq:caseb} means that up to $n_{two}$ vertices are inserted to either $\mathcal{E}_k[p_0]$ or $\mathcal{E}_k[p_1]$. But at least the last two edges for $\mathcal{E}_k[p_0]$ (i.e., $i-1$-th and $i$-th edges) never increase the number of duplicated vertices. 
This is because the two-hop vertices adjacent to $i-1$-th and $i$-th edges must belong to $\mathcal{E}_k[p_0]$.
The above is proved as follows.

Let the two-hop vertices adjacent to $i-1$-th and $i$-th edges be $w$ and $w'$; $X^{\phi}$ be the ordered edges up to $i$-th edge.
Note that $w$ and $w'$ must be in $V(X^{\phi}_{\mathit{ch}}(i - \delta, \delta))$ and $V(X^{\phi}_{\mathit{ch}}(i + 1 - \delta, \delta))$ respectively, due to the condition in Line~11 of Algorithm~\ref{alg:greedy-opt}.
Then, according to the assumption that $\delta = \left\lfloor\frac{|E|}{k_{max}}\right\rfloor-1$, $w$ is in
\begin{align*}
&V\left(X^{\phi}_{\mathit{ch}}\left(i - \left(\left\lfloor\frac{|E|}{k_{max}}\right\rfloor- 1\right), \left\lfloor\frac{|E|}{k_{max}}\right\rfloor\right)\right) \\
&=V\left(X^{\phi}_{\mathit{ch}}\left(i + 1 - \left\lfloor\frac{|E|}{k_{max}}\right\rfloor, \left\lfloor\frac{|E|}{k_{max}}\right\rfloor\right)\right) \\
&\subseteq V\left(X^{\phi}_{\mathit{ch}}\left(i + 1 -\left\lfloor\frac{|E| + p_0}{k_{max}}\right\rfloor, \left\lfloor\frac{|E|+p_0}{k_{max}}\right\rfloor\right)\right) \\ 
&\subseteq V\left(X^{\phi}_{\mathit{ch}}\left(i + 1 - \left\lfloor\frac{|E|+p_0}{k}\right\rfloor, \left\lfloor\frac{|E|+p_0}{k}\right\rfloor\right)\right)\\
&=V(\mathcal{E}_k[p_0])
\end{align*}
Similarly, $w'$ is also in $V(\mathcal{E}_k[p_0])$.

% $v_{min}$ and at most one of its adjacent vertices may be split and inserted to the other of the current partitions.
% Here, ``at most one'' is \grevise{also due to the assumption that the number of new assigned edges during one iteration is always smaller than the smallest partitioning size.
% Suppose \emph{two} of $v_{min}$'s adjacent vertices, $u_a$ and $u_b$, are both split; the ordering id of $v_{min}$ to $u_a$ is $x+i$ and that of $v_{min}$ to $u_b$ is $x+j$ ($i < j$ and $a,b$ are integers). 
% Then, under the assumption, the algorithm makes exactly one splitting point between $x+i$ and $x+j$ because $u_a$ is split.
% Let the splitting point be between $x+s$ and $x+s+1$ ($i \leq s < j$).
% Then, since $u_b$ is also split, there is an edge from $u_a$ to $u_b$ whose ordering id is less than $x+j$.
% However, the existence of such an edge from $u_a$ to $u_b$ is contradict because the edge from $v_{min}$ to $u_b$ must have been assigned before the current $t$-th iteration.
% More specifically, in the $t$-th iteration, the edge from $u_a$ to $u_b$ is assigned iif there exists an edge from $u_b$ to a certain vertex $v'$ which has already been assigned to the other partition and the difference between the ordering ids of two edges is less than $\theta$.}

For example in Figure~\ref{fig:upperbound}, the ordered edges are split between $x+3$ and $x+4$ ($\mathcal{E}_k[p_0]$ is red. $\mathcal{E}_k[p_1]$ is green).
In this case, $v_{min}$ is inserted to both $\mathcal{E}_k[p_0]$ and $\mathcal{E}_k[p_1]$. 
$u_0$ and $u_1$ are also inserted to both, but $u_2$ is inserted only to $\mathcal{E}_k[p_1]$.
$w_0$--$w_5$ are inserted to either $\mathcal{E}_k[p_0]$ or $\mathcal{E}_k[p_1]$, but $w_0$ and $w_1$ are not duplicated because they include edges which have already been assigned to $\mathcal{E}_k[p_0]$.
Thus, in total, $\Delta V(\mathcal{E}) = 2+2+2+1+4 = 11$, which is smaller than $2 + (n_{one} + 2) + (n_{two} - 2) = 12$.

Then, to summarize the above discussion, $\Delta\Phi$ can be calculated as follows:
\begin{align*}
    &\Delta\Phi = \Delta V_{rest} + \Delta E_{rest} + \Delta p_{rest} + \Delta V(\mathcal{E}) \\
               &\leq \begin{cases}\scriptstyle -1 - (n_{one} + n_{two}) + 0 + 1 + n_{one} + n_{two} \ (\text{Case A})\\ 
                                  \scriptstyle - 1 - (n_{one} + n_{two}) - 1 + 2 + (n_{one}+2) + (n_{two}-2) \ (\text{Case B})
                     \end{cases}
               \\
               &= 0
\end{align*}
Therefore, $\Delta\Phi \leq 0$ and thus (c) hold.

Based on (a), (b), (c), we establish the following equation:
\begin{align*}
 RF_k &:= \sum_{p=0}^{k-1}\frac{|V(\mathcal{E}_{k}[p])|}{|V|} = \frac{\Phi(T)}{|V|} \leq \frac{\Phi(0)}{|V|} = \frac{|V|+|E|+k}{|V|}
\end{align*}
Finally, the theorem is proved. \qed
\end{proof}

\smallskip
\noindent\textbf{\textit{Comparison to Existing Upper Bounds.}}
We compare our upper bound to the existing edge partitioning methods.
Our method provides an upper bound for general graphs, but most of the existing methods provide an upper bound only for power-law graphs. 
Thus, we apply our upper bound to the power-law graph. 

By using Clauset's power-law model~\cite{clauset2009power}, we model a graph, $G_{\zeta}(V,E)$, satisfying the following condition:
\begin{equation}
    Pr[d] = d^{-\alpha} \cdot \zeta(d,d_{min})^{-1}, \label{eq:degree}
\end{equation}
where $Pr[d]$ is the probability that vertex's degree becomes $d$; $\alpha$ is the scaling
parameter (typically, $2 < \alpha < 3$ for real-world graphs); $\zeta(d,d_{min})$ is the generalized/Hurwitz zeta function; and $d_{min}$ is the minimum degree.
We assume that $d_{min} = 1$ (then, $\zeta(d,d_{min})$ becomes Riemann zeta function) and that for any $k$, $|V| \gg k$ s.t. $k / |V| \approx 0$.

The expected value of the upper bound for $G_{\zeta}$ is formalized as follows:
\begin{align*}
    &\mathbb{E}\left[\frac{|V| + |E| + k}{|V|}\right] \approx 1+\mathbb{E}\left[\frac{|E|}{|V|}\right] \\ 
    &= 1+\frac{(G_{\zeta}\text{'s mean degree})}{2} \\
    &= 1+\frac{(\text{Mean of Eq.\eqref{eq:degree}}, d_{min} = 1)}{2} \\ 
    &= 1+\frac{1}{2}\cdot\frac{\zeta(\alpha - 1, 1)}{\zeta(\alpha, 1)}, 
\end{align*}
where the last equation is given by the mean value of the zeta distribution.

Table~\ref{tbl:upperbound} shows the calculation results under various $\alpha$ when $k = 256$.
We calculate the existing upper bounds based on~\cite{NIPS2014_5396,Petroni:2015:HSP:2806416.2806424,Zhang:2017:GEP:3097983.3098033,dynamicscaling}.
\textit{NE}~\cite{Zhang:2017:GEP:3097983.3098033} provides the best quality.
Our method is the second best and its score is very similar to \textit{NE}. 
The quality gap between the top two methods and the other ones is significant especially when $\alpha$ is small (i.e., a graph is more skewed).
Also, the small quality difference between our method and \textit{NE} is due to the fact that
our method can be applied with arbitrary $k$ values while \textit{NE} is restricted with the fixed $k$.
Such trends also appear in the empirical result in Sec.~\ref{sec:evaluation}. 

\begin{table}[h]
\centering
\caption{Theoretical Upper Bound of Replication Factor in Power-law Graph ($256$ Partition, $|V| = 10^6$).}
\label{tbl:upperbound}
\scalebox{1.0}{
\begin{tabular}{|l|r|r|r|r|}
\hline
    %   & \multicolumn{5}{c|}{$\alpha$} \\ \cline{2-6}
Partitioner & $\alpha = $ 2.2 & 2.4 & 2.6 & 2.8 \\ \hlinewd{2\arrayrulewidth}
\textit{Random} (1D-hash)         & 5.88 & 3.46 & 2.64 & 2.23 \\
\textit{Grid} (2D-hash) & 4.82    & 3.13 & 2.47 & 2.13 \\
\textit{DBH}~\cite{NIPS2014_5396} & 5.59 & 3.21 & 2.43 & 2.05 \\
\textit{HDRF}~\cite{Petroni:2015:HSP:2806416.2806424} & 5.36 & 4.23 & 3.61 & 3.24 \\
\textit{NE}~\cite{Zhang:2017:GEP:3097983.3098033} & \textbf{2.81} & \textbf{1.68} & \textbf{1.31} & \textbf{1.13} \\
\textit{BVC}~\cite{dynamicscaling}  & 11.10 & 6.39 & 4.85 & 4.10 \\
\textit{Proposed Method}  & \textit{2.88} & \textit{2.12} & \textit{1.88} & \textit{1.75} \\ \hline
\end{tabular}
}
\end{table}

\section{Evaluation} \label{sec:evaluation}
In this section, we provide a comprehensive evaluation of our two proposed methods: the \emph{graph edge ordering} and the \emph{chunk-based edge partitioning}.  
% First, we describe the experimental frames of the evaluation.
% Then, we perform a comparative analysis with the existing graph partitioning methods and the existing graph ordering methods.
% After that, we evaluate the effect on the distributed graph applications.
% Finally, the migration and scalability are evaluated.

Our main results are summarized as follows:
% , where efficiency and quality are evaluated by the elapsed time and the replication factor respectively:

\smallskip
\noindent\textit{\textbf{Highest Scaling Efficiency.}}
The chunk-based edge partitioning is significantly faster than the existing methods.
Even compared to the very efficient simple hashing method, the chunk-based edge partitioning is over three orders of magnitude faster.

\smallskip
\noindent\textit{\textbf{Similar Partitioning Quality to the Best Graph Partitioning Method.}}
The quality of edge partitions that the graph edge ordering and the chunk-based edge partitioning generate is comparable to the high-quality but time-consuming $k$-way graph partitioning methods.
The quality-loss due to variable $k$ in our method is small.

\smallskip
\noindent\textit{\textbf{Highest Partitioning Quality Compared to Graph Ordering Methods.}} 
Among the existing ordering methods, the graph edge ordering delivers the best improvement to the partitioning quality of the chunk-based edge partitioning.
% We perform a comprehensive analysis with the other ordering methods, each of which has a different focus, such as, to optimize the L1-cache performance, to minimize bandwidth of adjacency matrix, and to maximize graph modularity.

\smallskip
\noindent\textit{\textbf{Acceptable Preprocessing Time.}} 
Due to the efficient greedy algorithm, the elapsed time for the graph edge ordering is similar to the other ordering methods.
It can order billion-edge graphs within an acceptable time.

\smallskip
\noindent\textit{\textbf{Performance Improvement in Distributed Graph Applications.}} 
Edge partitions generated by the chunk-based edge partitioning and the graph edge ordering highly improve the performance of the distributed graph applications (SSSP, WCC, and PageRank) due to the large reduction of communication volumes.

Moreover, the evaluation result for the migration cost and the scalability of our proposed algorithm are shown as additional experiments. 

\subsection{Experimental Frame}
% We illustrate the experimental frame: the graph data sets, the computational environments, and the scenario for the dynamic scaling.

% \smallskip
\noindent\textit{\textbf{Graph Data Sets.}}
We use various types of real-world large graphs (over 1 million vertices) provided by SNAP~\cite{snapnets} and KONNECT~\cite{KONECT} as summarized in Table~\ref{tbl:real_world}.
\texttt{Road-CA} is the road network of California.
\texttt{Skitter} is the internet topology graph of autonomous systems.
\texttt{Patents} is the citation network.
\texttt{Pokec}, \texttt{Flickr}, \texttt{LiveJ.}, \texttt{Orkut}, \texttt{Twitter}, and \texttt{FriendS.} are online social networks in each service.
\texttt{Road-CA} is a non-skewed graph, whereas the others are skewed graphs, namely, they have skewed degree distribution.

\begin{table}[t]
\centering
\caption{Datasets (M = Million, B = Billion)}
\label{tbl:real_world}
\scalebox{1.0}{
\begin{tabular}{|l|r|r|r|}
\hline
\textbf{Dataset}                & $|V|$ & $|E|$ & \textbf{Type} \\ \hlinewd{4\arrayrulewidth}
\texttt{Road-CA} \cite{leskovec2009community} & 1.96 M &  2.76 M & Traffic\\
\texttt{Skitter} \cite{leskovec2005graphs}  & 1.70 M & 11.09 M & Internet \\ 
\texttt{Patents} \cite{leskovec2005graphs}    & 3.77 M & 16.51 M & Citation\\ 
\texttt{Pokec} \cite{takac2012data}           & 1.63 M & 30.62 M & Social Net.\\ 
\texttt{Flickr} \cite{Mislove:2008:GFS:1397735.1397742}  & 2.30 M & 33.14 M & Social Net.\\
\texttt{LiveJ.} \cite{Backstrom:2006:GFL:1150402.1150412}& 4.8 M & 68 M & Social Net.\\ 
\texttt{Orkut} \cite{6413740}                 & 3.1 M & 117 M & Social Net.\\ 
\texttt{Twitter} \cite{Kwak:2010:TSN:1772690.1772751}  & 41.6 M     & 1.46 B & Social Net.\\
\texttt{FriendS.} \cite{yang2012defining} & 65.6 M & 1.80 B & Social Net.\\ \hline
\end{tabular}
}
\end{table}

\begin{table}[t]
\centering
\caption{Graph Partitioning Methods.}
\label{tbl:palgorithms}
\scalebox{1.0}{
 \begin{tabular}{|l|l|l|l|}
\hline
\textbf{Method}       & \textbf{Part by} & \textbf{Description}              \\ \hlinewd{4\arrayrulewidth}
\textit{BVC}~\cite{dynamicscaling} & Edge & State-of-the-art dynamic scaling \\ \hline
\textit{NE}~\cite{Zhang:2017:GEP:3097983.3098033} & Edge & Highest-quality offline method \\ \hline
\textit{DBH}~\cite{NIPS2014_5396} & Edge & Degree-based hashing method \\ \hline
\textit{HDRF}~\cite{Petroni:2015:HSP:2806416.2806424} & Edge & High-Degree Replicated First \\ \hline
\textit{1D}/\textit{2D}          & Edge & 1D / 2D random hash \\ \hline
\textit{MTS}~\cite{Karypis:1998:FHQ:305219.305248} & Vertex & \textit{METIS} \\ \hline
\textit{CVP}~\cite{zhu2016gemini} & Vertex & Chunk-based vertex partitioning \\ \hline
\textbf{\textit{CEP}} & \textbf{\textit{Edge}} & \textbf{\textit{Chunk-based edge partitioning}} \\\hline
\end{tabular}
}
\end{table}

\begin{table}[t]
\centering
\caption{Graph Ordering Methods}
\label{tbl:oalgorithms}
\scalebox{1.0}{
 \begin{tabular}{|l|l|l|l|}
\hline
\textbf{Method}       & \textbf{Order by} & \textbf{Description} \\ \hlinewd{4\arrayrulewidth}
\textit{GO}~\cite{wei2016speedup} & Vertex & Optimized to L1-cache \\ \hline
\textit{RO}~\cite{arai2016rabbit} & Vertex & \textit{RabbitOrder} \\ \hline
\textit{RGB}~\cite{dhulipala2016compressing} & Vertex & \textit{Recursive Graph Bisection} \\ \hline
\textit{LLP}~\cite{boldi2011layered} & Vertex & \textit{Layered Label Propagation} \\ \hline
\textit{RCM}~\cite{Cuthill:1969:RBS:800195.805928} & Vertex & \textit{Reverse Cuthill--McKee} \\ \hline
\textit{DEG}      & Vertex & Simple degree sorting \\ \hline
\textit{DEF}          & Vertex & Default ordering \\ \hline
\textbf{\textit{GEO}} & \textbf{\textit{Edge}} & \textbf{\textit{Proposed Greedy Algorithm}}\\ \hline
\end{tabular}%
}%
\end{table}

% \begin{table}[h]
% \vspace{-10pt}
% \centering
% \caption{Real-world Graphs (M = Million, B = Billion)}
% \vspace{-10pt}
% \label{tbl:real_world}
% \scalebox{0.8}{
% \begin{tabular}{|l|r|r|l|r|r|}
% \hline
% Dataset                & $|V|$ & $|E|$ & Dataset                & $|V|$ & $|E|$ \\ \hlinewd{2\arrayrulewidth}
% \texttt{Road-CA} \cite{leskovec2009community} & 1.96 M &  2.76 M & \texttt{LiveJ.} \cite{Backstrom:2006:GFL:1150402.1150412}& 4.8 M & 68 M \\
% \texttt{Skitter} \cite{leskovec2005graphs}  & 1.70 M & 11.09 M & \texttt{Orkut} \cite{6413740}                 & 3.1 M & 117 M \\
% \texttt{Patents} \cite{leskovec2005graphs}    & 3.77 M & 16.51 M & \texttt{Twitter} \cite{Kwak:2010:TSN:1772690.1772751}  & 41.6 M     & 1.46 B        \\ 
% \texttt{Pokec} \cite{takac2012data}           & 1.63 M & 30.62 M & \texttt{FriendS.} \cite{yang2012defining} & 65.6 M & 1.80 B \\
% \texttt{Flickr} \cite{Mislove:2008:GFS:1397735.1397742}  & 2.30 M & 33.14 M & & & \\ \hline 
% % \texttt{WebUK} \cite{BSVLTAG} & 105,153,952 & 3,717,169,969 \\ 
% % \texttt{RoadNet-PA} \cite{leskovec2009community} & 1,088,092 & 1,541,898 \\
% % \texttt{RoadNet-TX} \cite{leskovec2009community} & 1,379,917 & 1,921,660 \\
% \end{tabular}
% }
% \vspace{-5pt}
% \end{table}

% \smallskip
\noindent\textit{\textbf{Comparing Methods.}}
We compare our two methods with 15 existing methods, as shown in Table~\ref{tbl:palgorithms} and Table~\ref{tbl:oalgorithms}.
We refer to the chunk-based edge partitioning as \textit{CEP}, and to the efficient greedy algorithm (i.e., Algorithm~\ref{alg:greedy-opt}) for the graph edge ordering as \textit{GEO}.
We classify all the methods into five categories: dynamic scaling, edge/vertex partitioning, and edge/vertex ordering.

Table~\ref{tbl:palgorithms} shows algorithms of the dynamic scaling and graph partitioning. \textit{BVC} (precisely, \textit{BVC}$^{+/-}$) is the state-of-the-art dynamic scaling method for graph partitions based on consistent hashing.
\textit{NE} is the latest offline edge partitioning method, which basically provides the best-quality edge partitioning in practice.
\textit{DBH} is the degree-based-hashing edge partitioning.
\textit{HDRF} is High-Degree Replicated First streaming edge partitioning.
\textit{1D} and \textit{2D} are simple hash-based edge partitioning.
In \textit{1D}, each edge is randomly assigned to 1D integer partitioning id space (0,1,2,...) by hashing its edge id. In \textit{2D}, each edge is randomly assigned to 2D partitioning id space ($<0,0>$, $<0,1>$,.., $<1,0>$, $<1,1>$, ...) by separately hashing its source id and destination id. The hash value of the source id determines the first dimension while that of the destination id does the second dimension.
\
\textit{MTS} (\textit{METIS}) is the high-quality offline vertex partitioning method.
\textit{CVP} is the chunk-based vertex partitioning, where the ordered vertex is simply divided into the same size of vertex chunks.
Table~\ref{tbl:oalgorithms} shows algorithms of graph ordering.
\textit{GO} and \textit{RO} are vertex id ordering methods for maximizing CPU-cache utilization. 
% The key difference between the two is that \textit{GO} approximately solve the optimization problem whereas in \textit{RO}, the locality among vertices are determined by network modularity presented by M. E. J. Newman~\cite{newman2006modularity}.
\textit{RGB} and \textit{LLP} are for graph compression. 
% \textit{RGB} is based on Recursive Graph Bisection. \textit{LLP} is based on Layered Label Propagation.
\textit{DEG} is simple degree sorting. 
\textit{DEF} is default ordering.
% \textit{GO} and \textit{RO} are the vertex ordering method (vertex id reordering) for maximizing cache efficiency in graph analysis. The key difference between the two is that \textit{GO} approximately solve the optimization problem whereas in \textit{RO}, the locality among vertices are determined by network modularity presented by M. E. J. Newman~\cite{newman2006modularity}.
% \textit{RGB} and \textit{LLP} are the vertex ordering method for graph compression. \textit{RGB} is based on Recursive Graph Bisection. \textit{LLP} is based on Layered Label Propagation.
% \textit{DEG} is the simple vertex sorting by degree. 

% \smallskip
\noindent\textit{\textbf{Computational Infrastructure.}}
We use Ubuntu server (ver. 18.04) with dual sockets of Intel Xeon CPU E5-2697 v4 (18 cores per socket, 2.30GHz) and 500GB RAM.
All the programs except for \textit{LLP} and \textit{HDRF} are written in C/C++, which we compile via GCC 7.4.0 with --O3 optimization flag. 
For \textit{LLP} and \textit{HDRF}, we use OpenJDK version 11.0.4 on 400GB JVM memory.
Although some of the algorithms, such as \textit{RO} or \textit{BVC}, support the parallel execution, all the programs are run on a single core for a fair comparison.
The parallelization of our methods, especially \textit{GEO}, is an interesting problem but out of scope in this paper. 
We list it as our future work in Sec.~\ref{sec:conclusion}.

% \smallskip
\noindent\textit{\textbf{Parameters.}} According to the existing experimental studies on distributed graph processing systems and graph partitioning~\cite{Han:2014:ECP:2732977.2732980,6877273,Verma:2017:ECP:3055540.3055543,abbas2018streaming,Gill:2018:SPP:3297753.3316427,Pacaci:2019:EAS:3299869.3300076}, the number of distributed processes (i.e., partitions) for graph applications usually ranges from less than 10 to around one hundred.
Thus, in the evaluation, we change the number of partitions, $k$, from 4 to 128.
% (Specifically, $k$ is $4 \rightarrow 8 \rightarrow 16 \rightarrow 32 \rightarrow 64 \rightarrow 128$).
For \textit{GEO}, $k_{min}$ and $k_{max}$, as defined in Def.~\ref{def:ordering1} of Sec.~\ref{sec:def}, are $4$ and $128$, respectively.

\subsection{Comparison with Graph Partitioning}
We compare our methods to the existing graph partitioning methods and dynamic scaling methods, as shown in Table~\ref{tbl:palgorithms}.
Note that for 128 partitions, \textit{NE} and \textit{MST} cannot correctly execute \texttt{FriendSter}, nor can \textit{NE} do \texttt{Twitter}. 
For \textit{BVC}, we run the algorithm as $k$ is $4 \rightarrow 8 \rightarrow 16 \rightarrow 32 \rightarrow 64 \rightarrow 128$, and we set the balance factor $\epsilon = 0.001$ as defined in Def.~\ref{def:edgepartitioning}.

% \smallskip
\noindent\textit{\textbf{Scaling Efficiency.}}
Efficiency is measured by the elapsed time.
Figure~\ref{fig:performance} shows the elapsed time for each method.
In \textit{BVC}, we measure the repartitioning time from the previous partitions (e.g., for $k=8$, time from $k=4$ to $k=8$ is used.).
We ignore the initialization phase, such as, data loading and graph construction, and the graph data migration phase.
% , which is typically done after computing the partitioning.   

As expected, \textit{CEP} is significantly faster than the others due to its $\mathcal{O}(1)$ time complexity.
It is over 1,000 times faster than the other methods for all the data sets.
Also, the performance of \textit{CEP} is not changed with the increase of the graph size, which is a consistent result to Theorem~\ref{thr:cep}.
In the other existing methods, even for the very simple partitioning, such as \textit{1D/2D} or \textit{CVP}, each edge needs to be processed one-by-one, resulting in that the elapsed time increases proportionally to the graph size.

\begin{figure*}[t]
 \centering
  \subfigure[\texttt{Road-CA}]{\includegraphics[width=.25\textwidth]{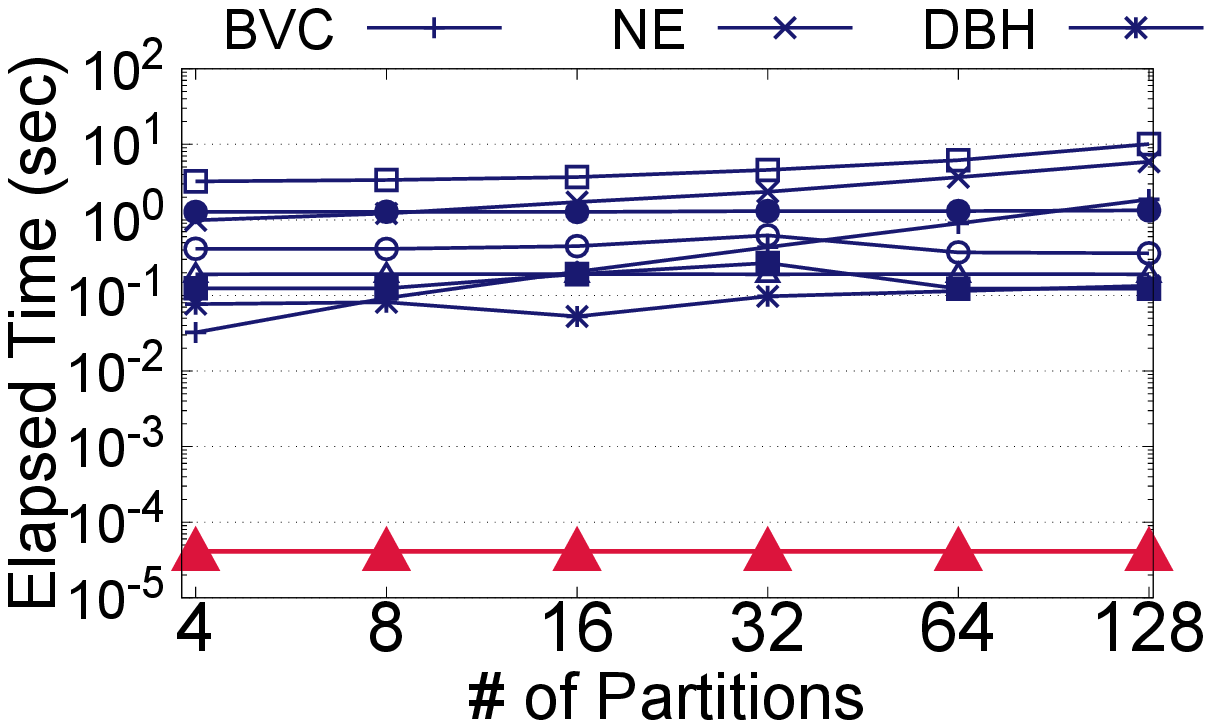}\label{fig:performance-pokec}}%
  \subfigure[\texttt{Skitter}]{\includegraphics[width=.25\textwidth]{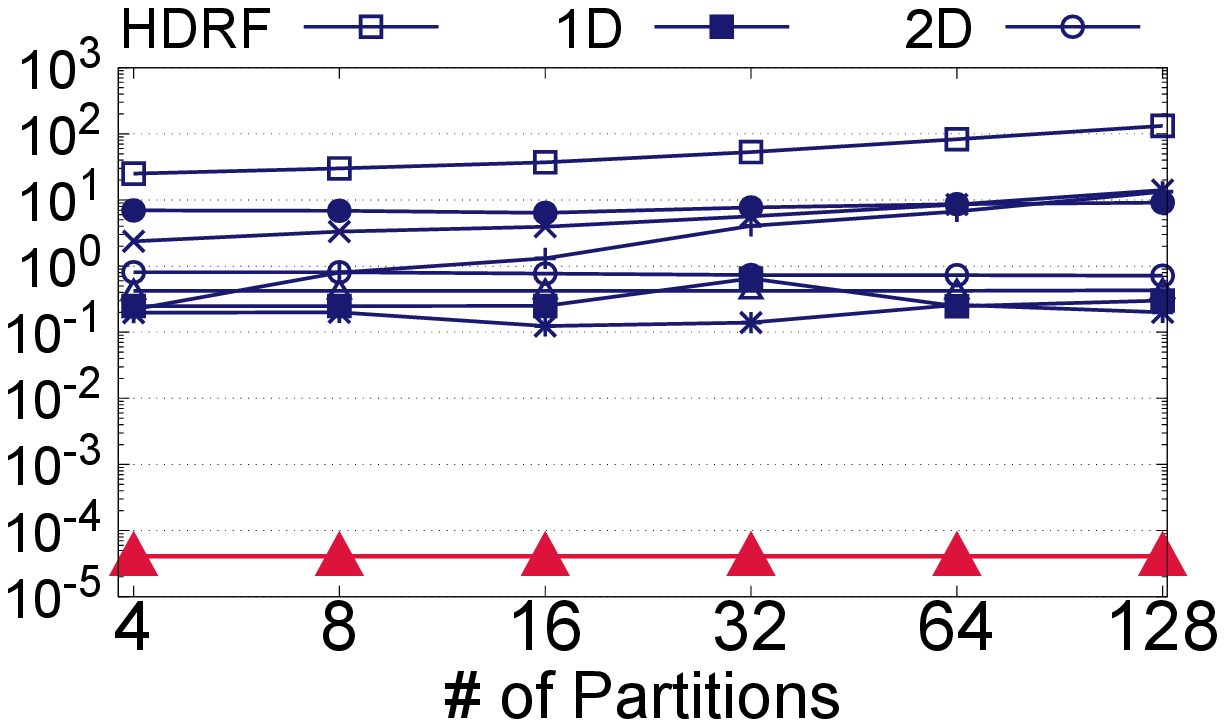}\label{fig:performance-pokec}}%
  \subfigure[\texttt{Patent}]{\includegraphics[width=.25\textwidth]{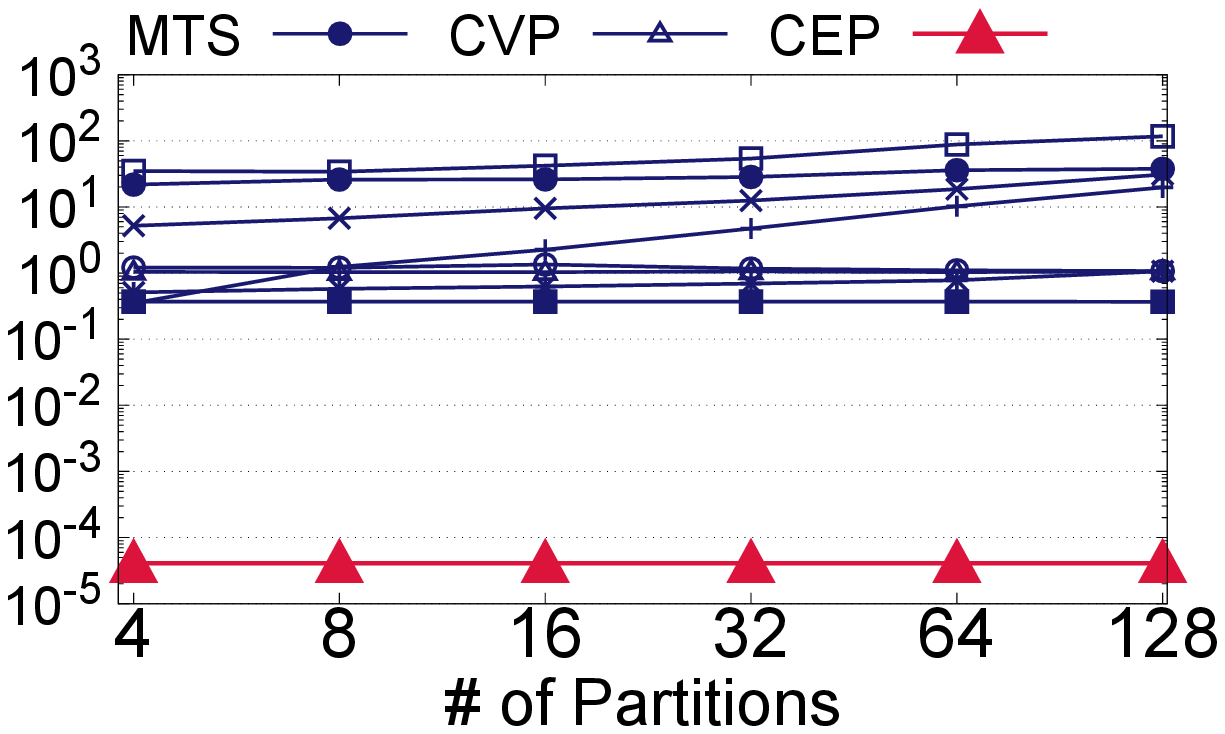}\label{fig:performance-pokec}}%
  \subfigure[\texttt{Pokec}]{\includegraphics[width=.25\textwidth]{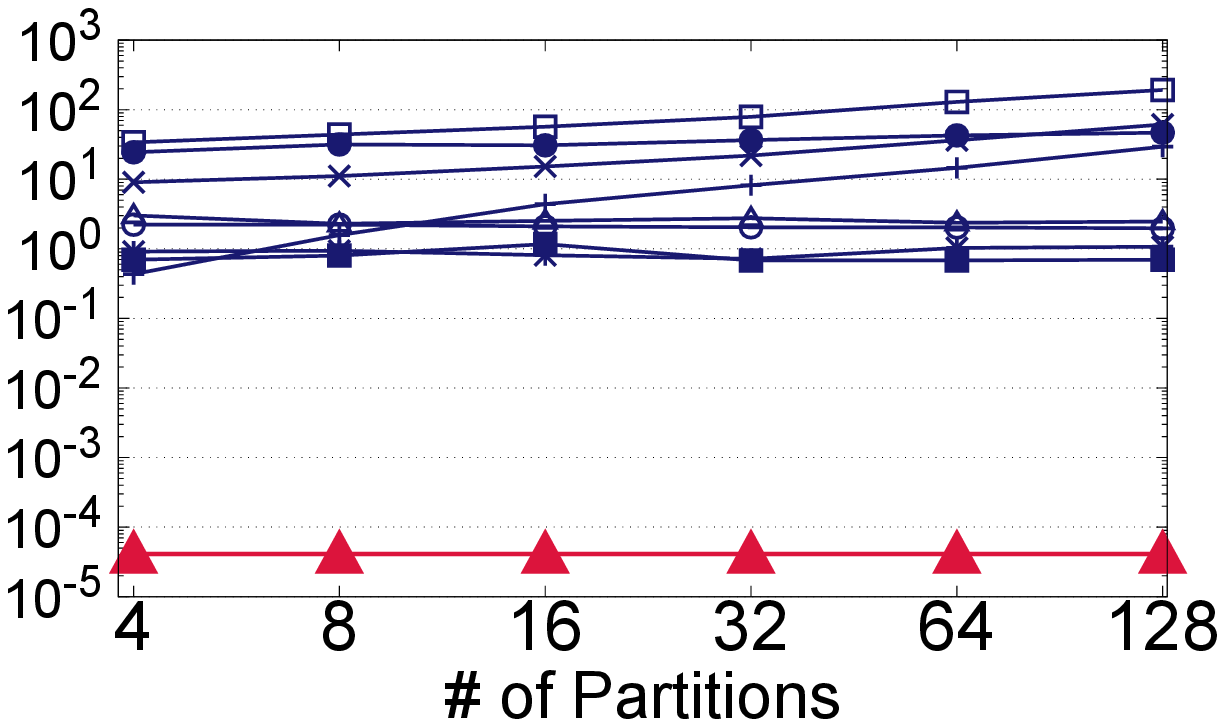}\label{fig:performance-pokec}} \\
 \subfigure[\texttt{Flickr}]{\includegraphics[width=.20\textwidth]{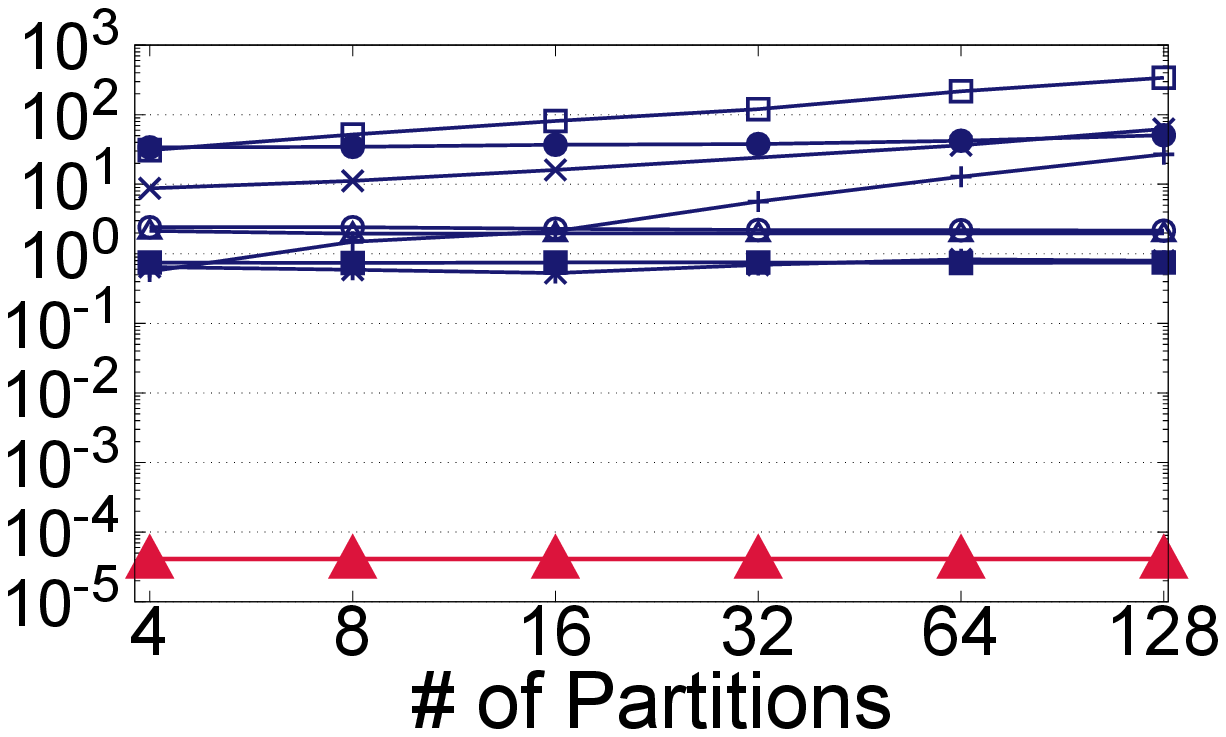}\label{fig:performance-pokec}}%
  \subfigure[\texttt{LiveJournal}]{\includegraphics[width=.20\textwidth]{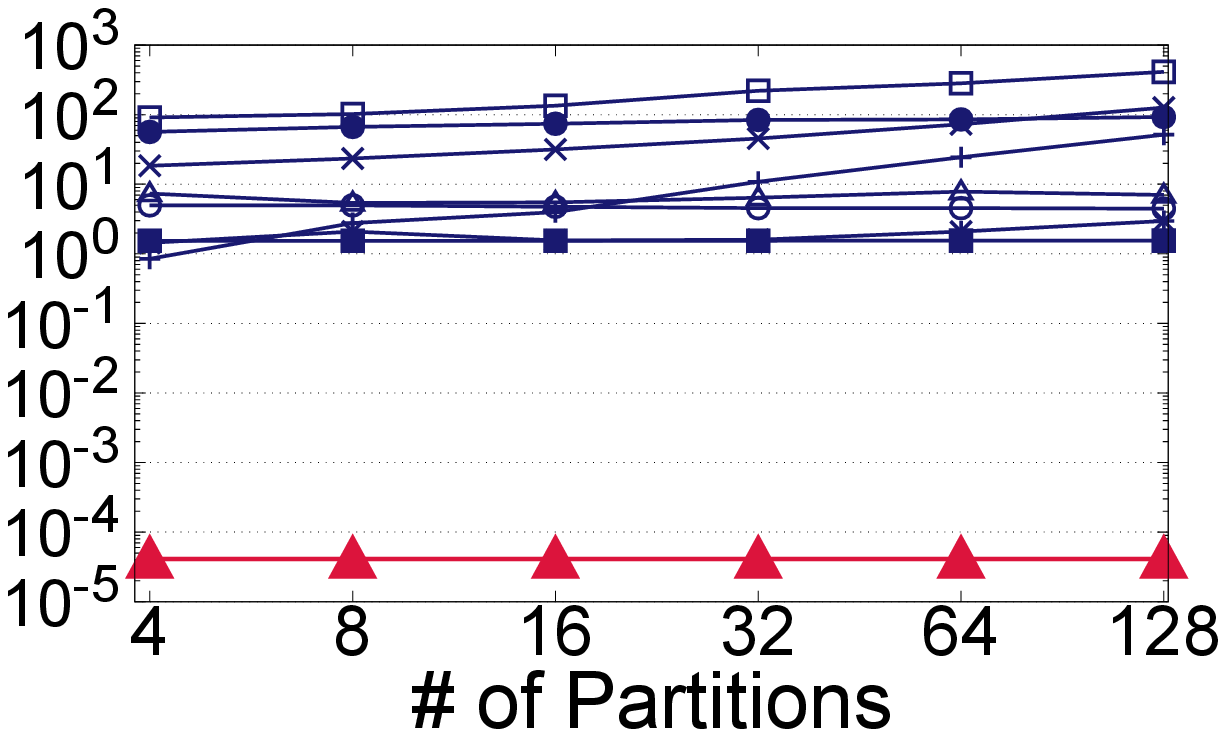}\label{fig:performance-pokec}}%
  \subfigure[\texttt{Orkut}]{\includegraphics[width=.20\textwidth]{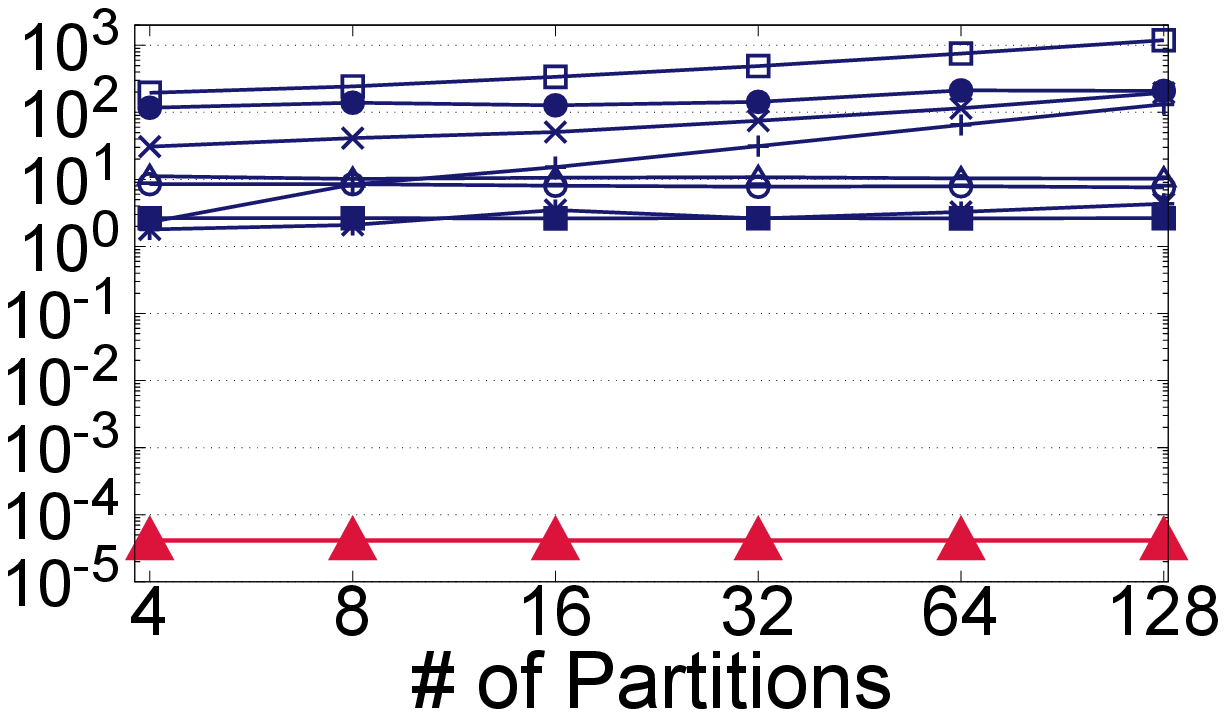}\label{fig:performance-pokec}}%
  \subfigure[\texttt{Twitter}]{\includegraphics[width=.20\textwidth]{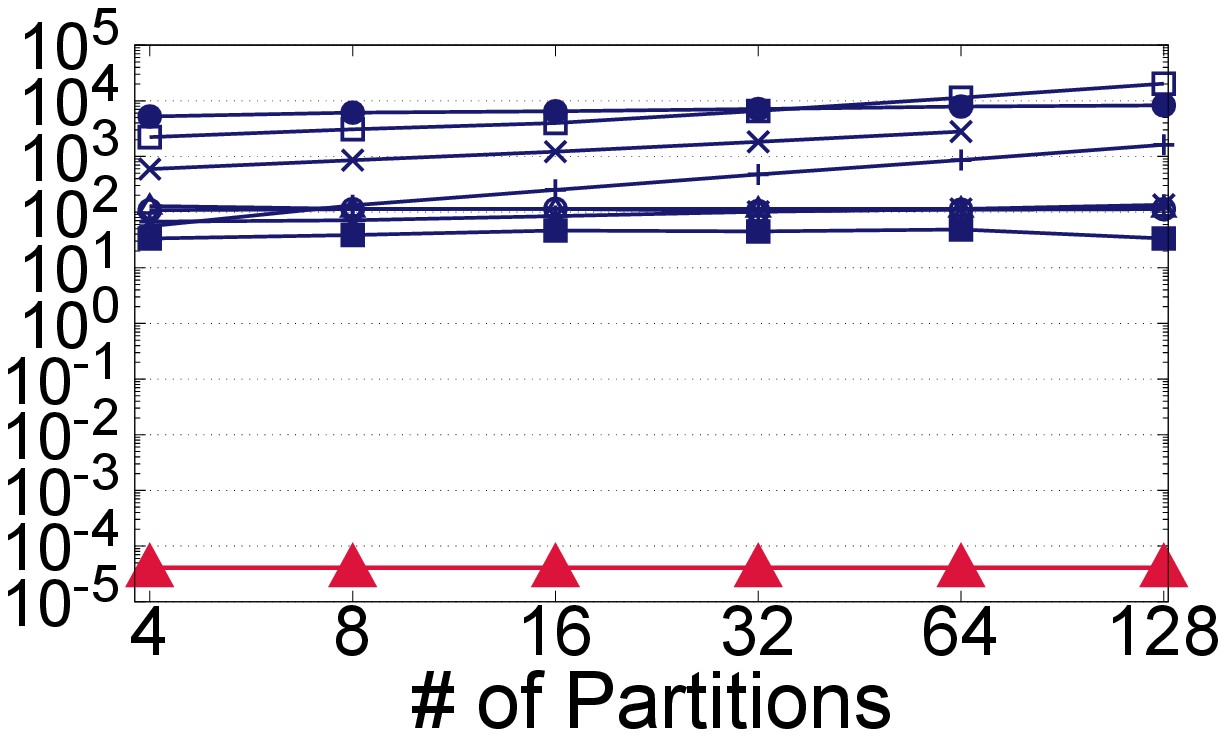}\label{fig:performance-pokec}}%
  \subfigure[\texttt{FriendSter}]{\includegraphics[width=.20\textwidth]{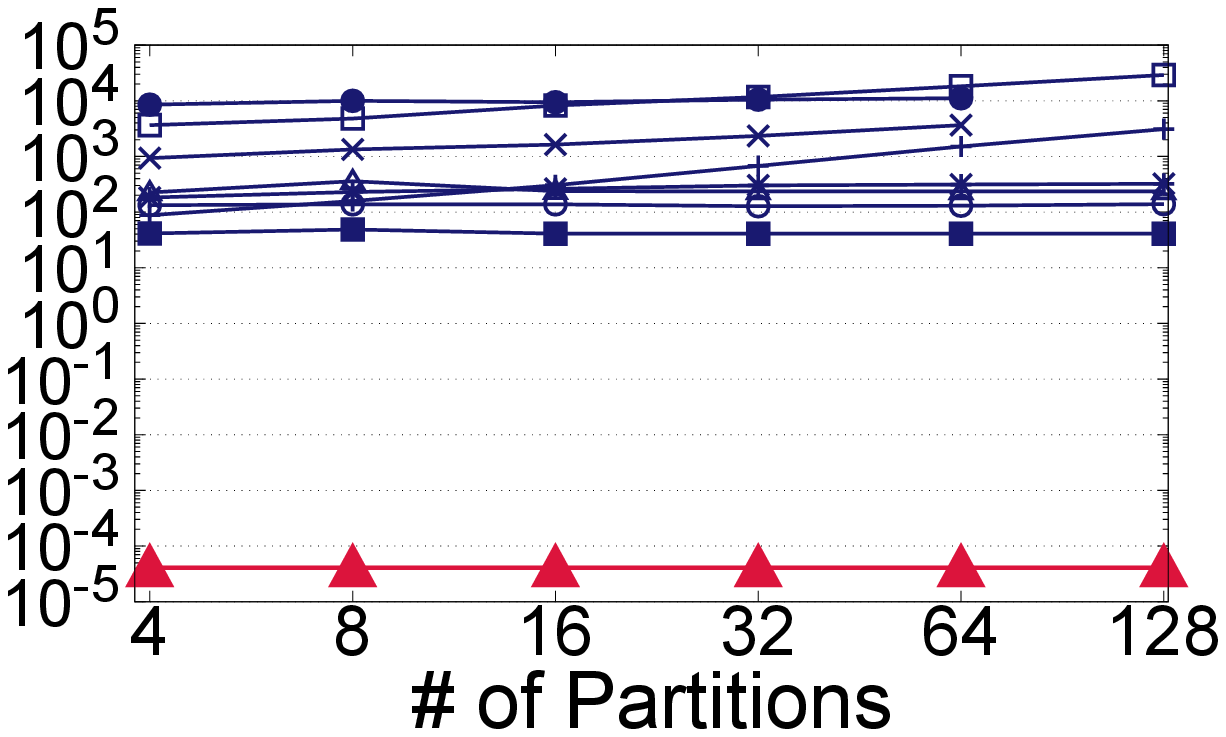}\label{fig:performance-pokec}}%
 \caption{Elapsed Time for Graph Partitioning.}\label{fig:performance}
\end{figure*}

\begin{figure*}[t]
 \centering
  \subfigure[\texttt{Road-CA}]{\includegraphics[width=.25\textwidth]{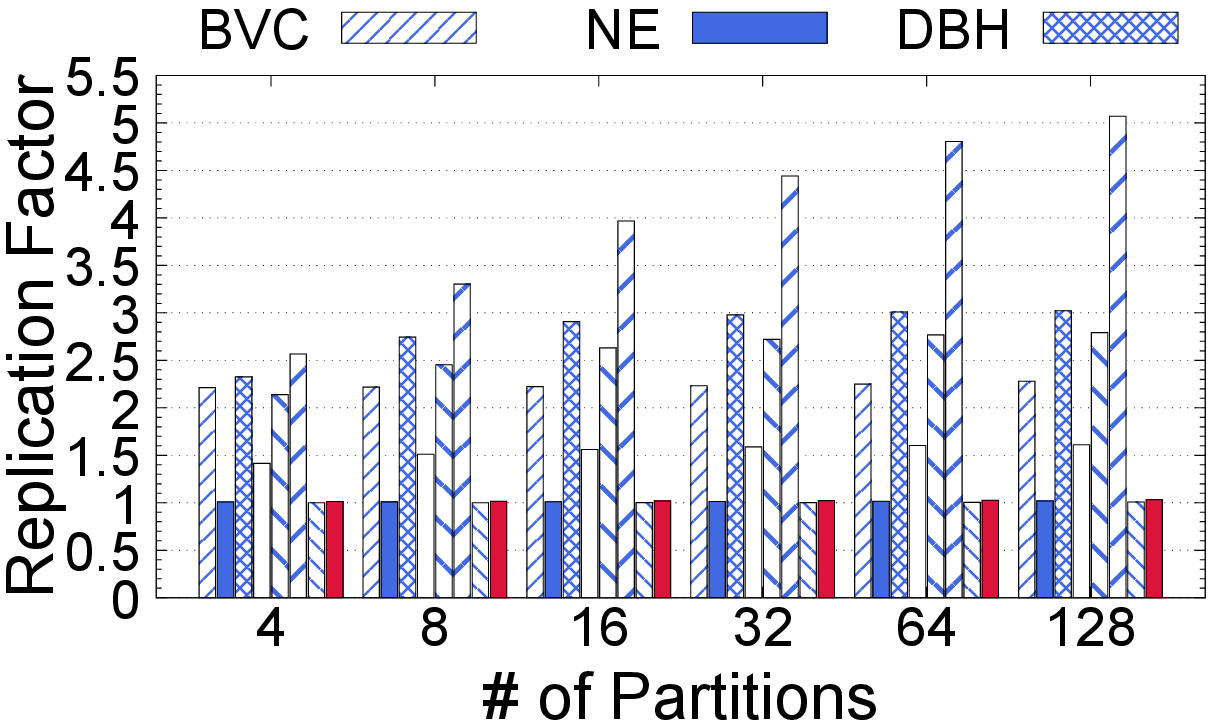}\label{fig:performance-pokec}}%
  \subfigure[\texttt{Skitter}]{\includegraphics[width=.25\textwidth]{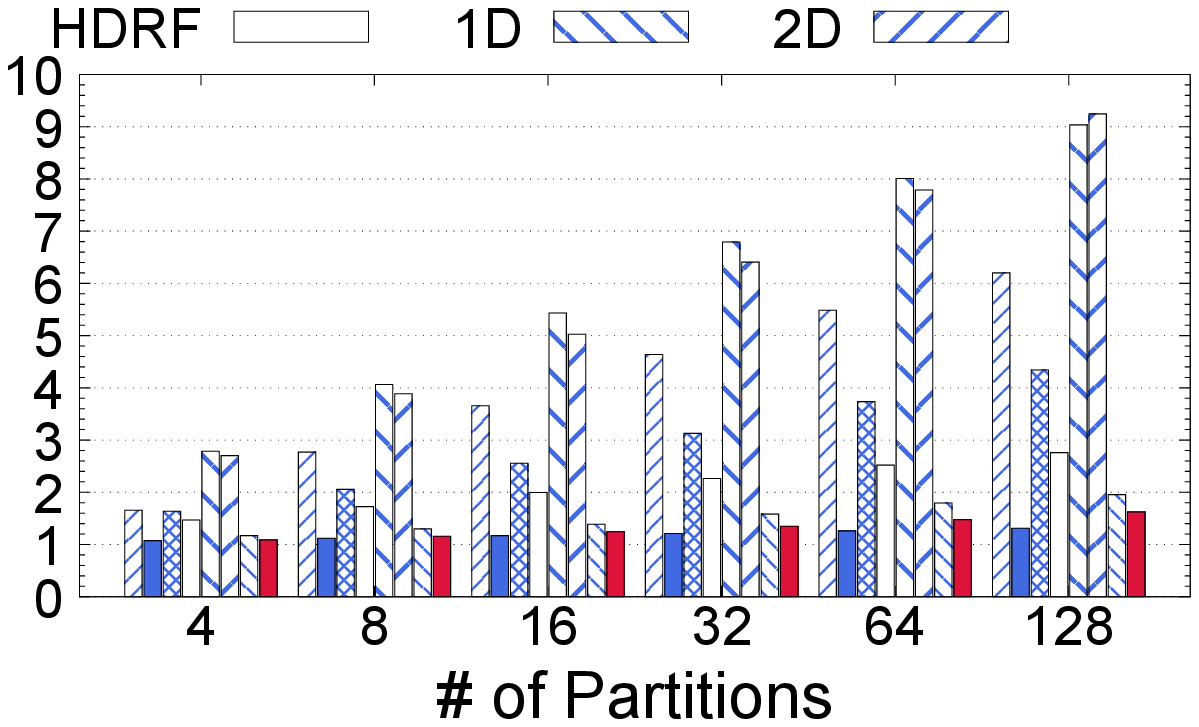}\label{fig:performance-pokec}}%
  \subfigure[\texttt{Patent}]{\includegraphics[width=.25\textwidth]{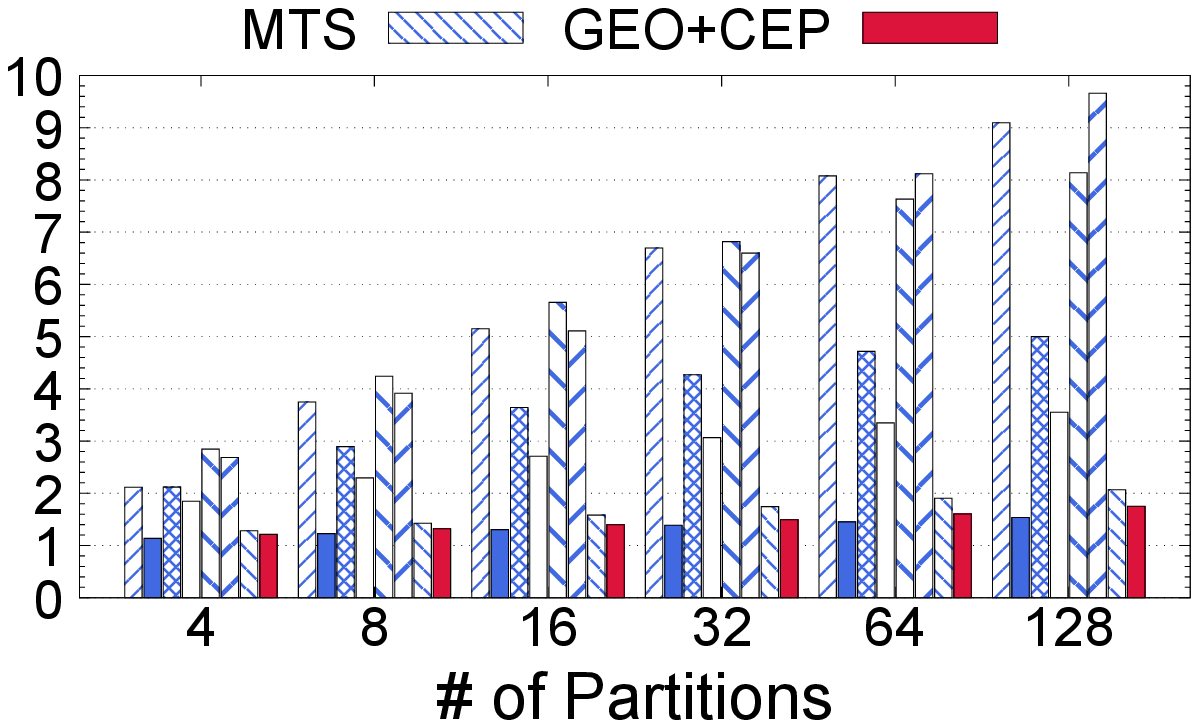}\label{fig:performance-pokec}}%
 \subfigure[\texttt{Pokec}]{\includegraphics[width=.25\textwidth]{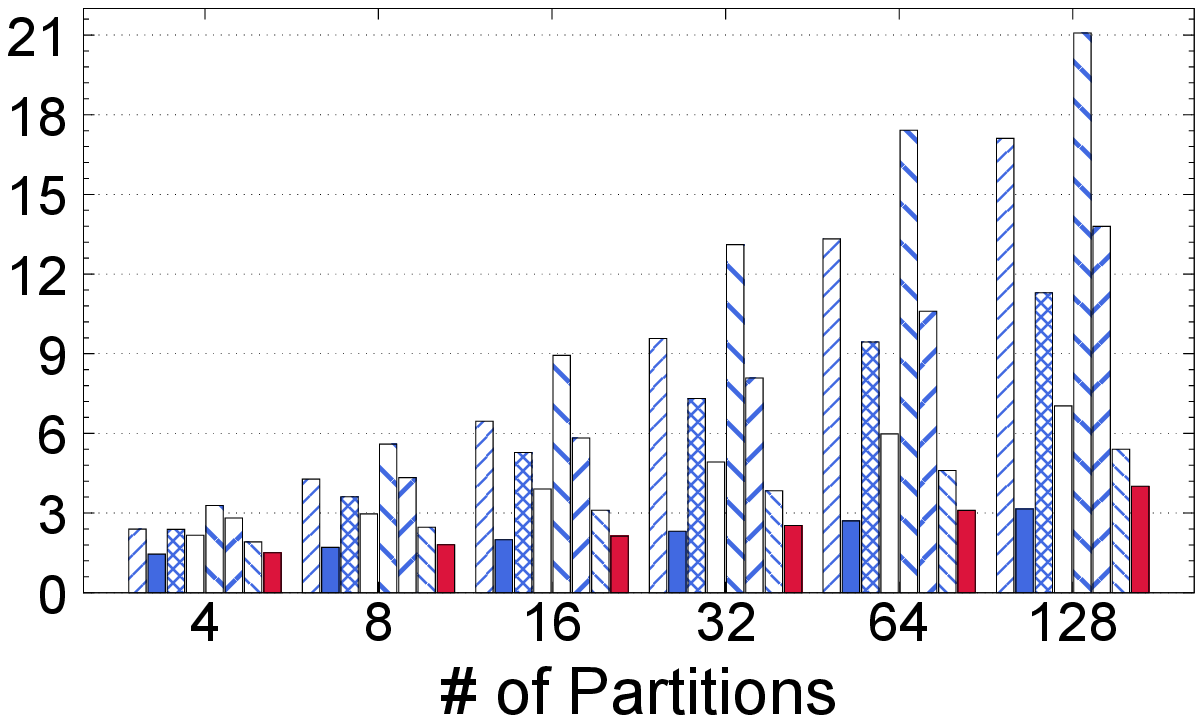}\label{fig:performance-pokec}} \\
  \subfigure[\texttt{Flickr}]{\includegraphics[width=.20\textwidth]{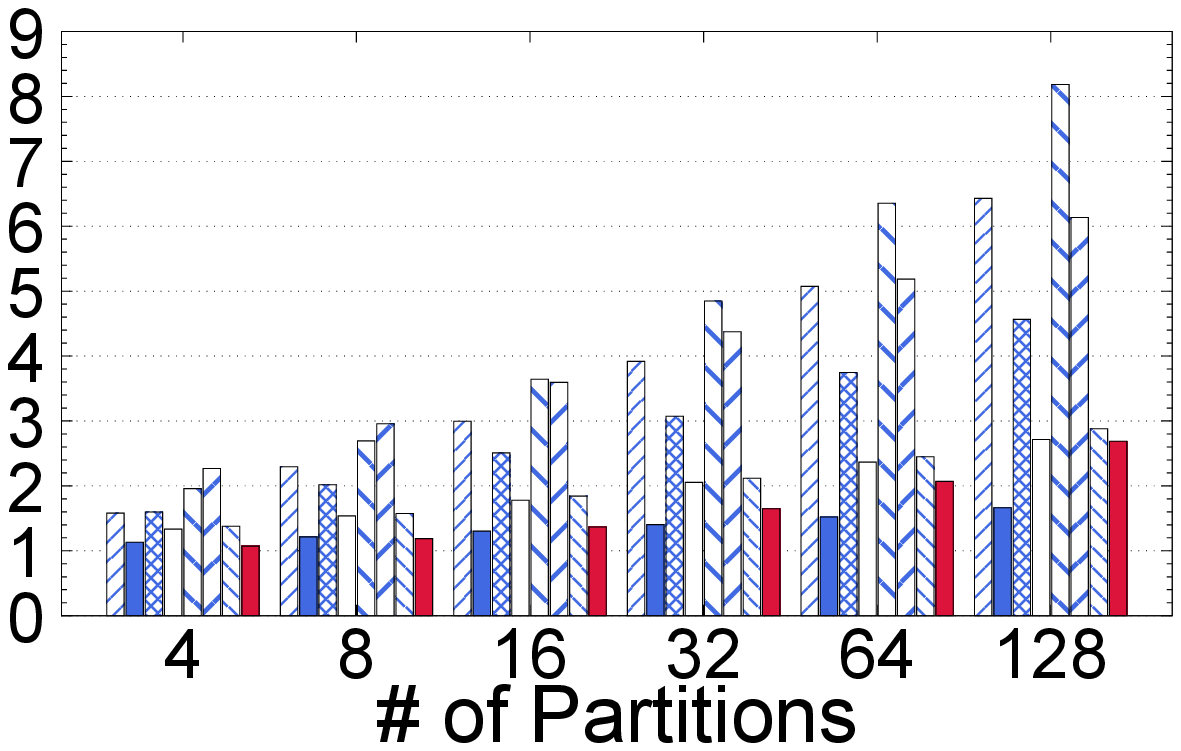}\label{fig:performance-pokec}}%
  \subfigure[\texttt{LiveJournal}]{\includegraphics[width=.20\textwidth]{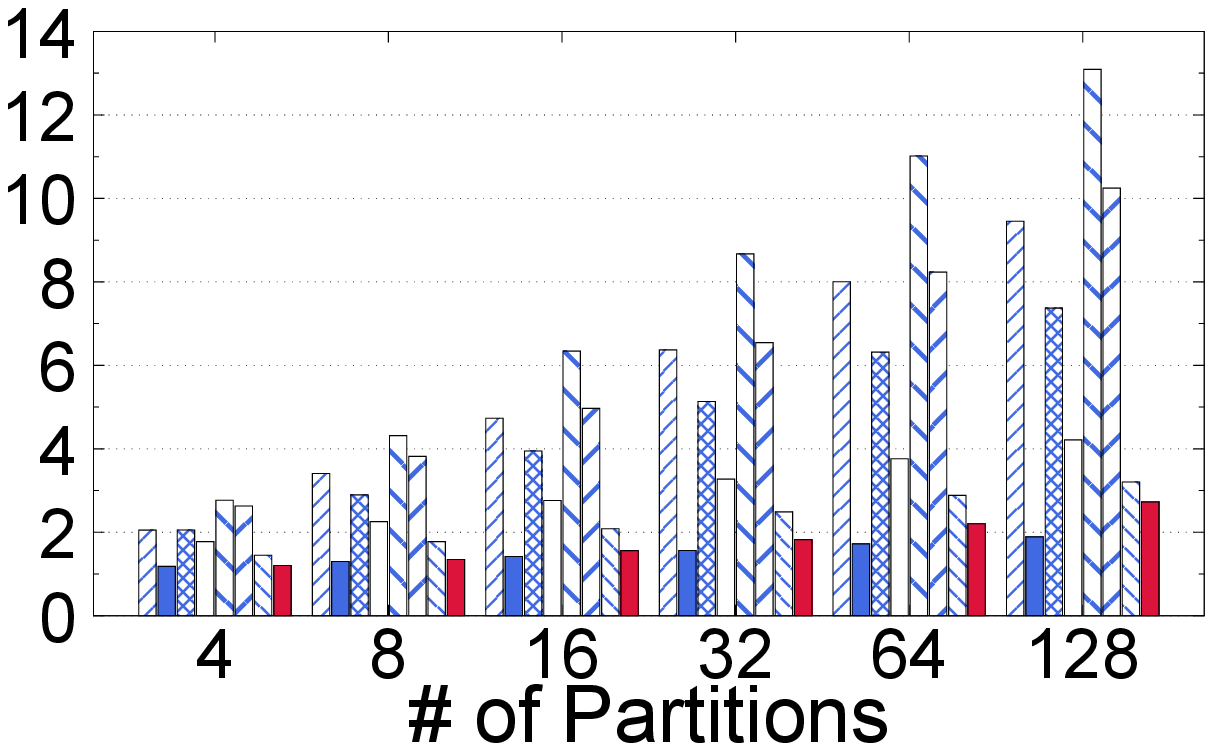}\label{fig:performance-pokec}}%
  \subfigure[\texttt{Orkut}]{\includegraphics[width=.20\textwidth]{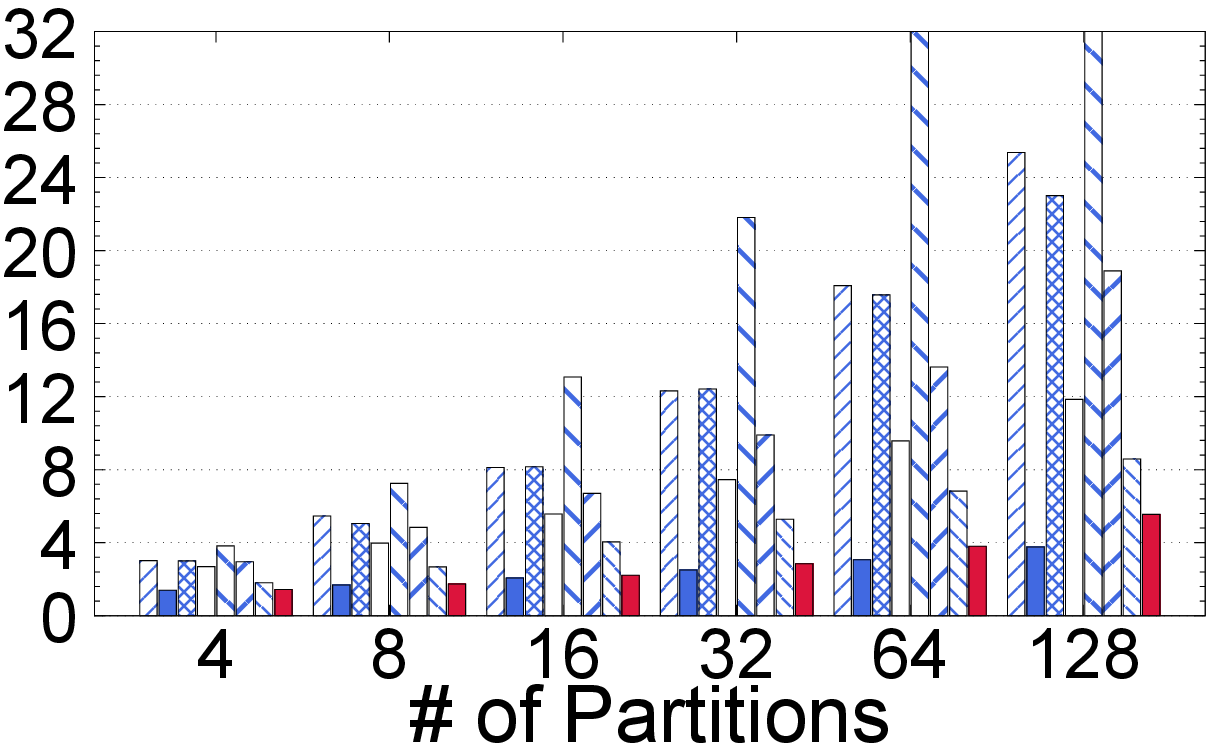}\label{fig:performance-pokec}}%
  \subfigure[\texttt{Twitter}]{\includegraphics[width=.20\textwidth]{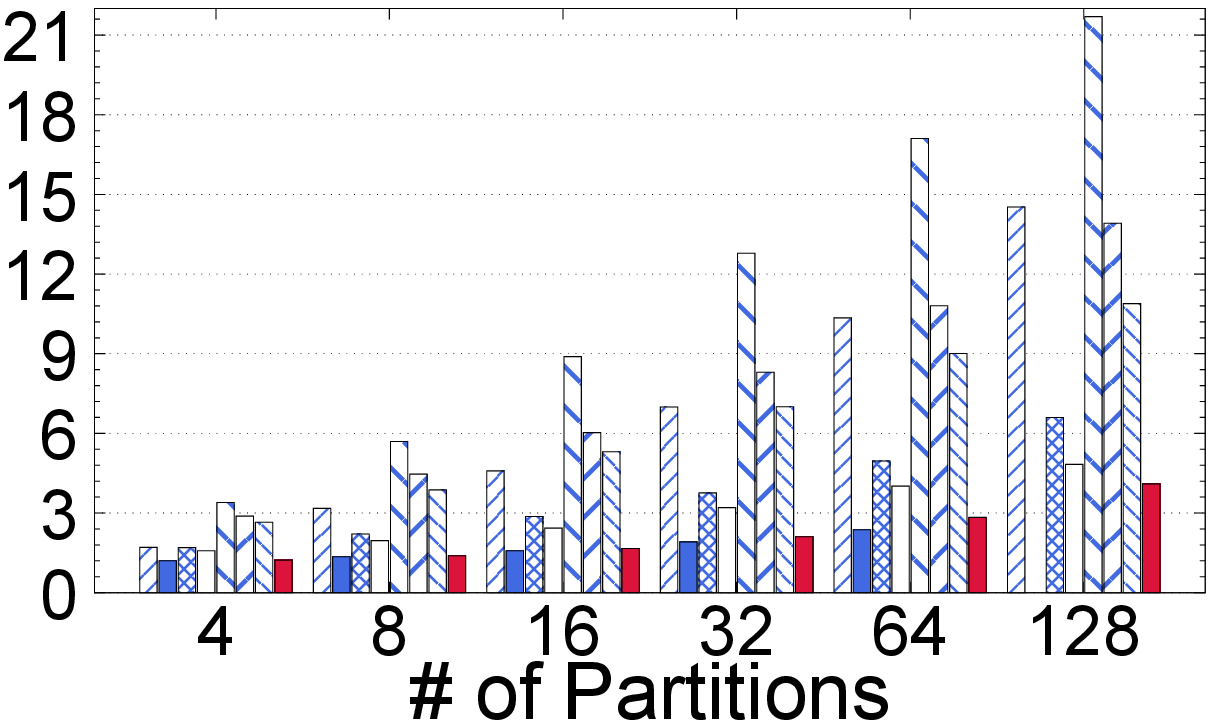}\label{fig:performance-pokec}}%
  \subfigure[\texttt{FrindSter}]{\includegraphics[width=.20\textwidth]{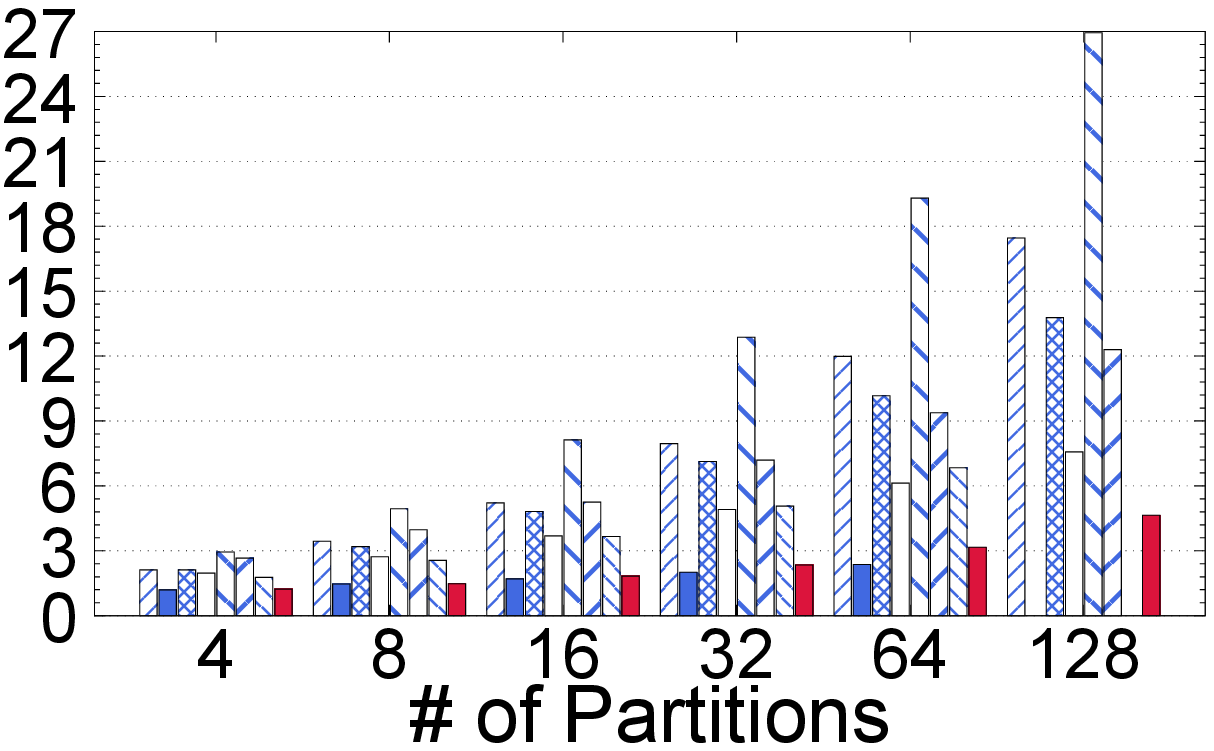}\label{fig:performance-pokec}}
\caption{Replication Factor Compared to Graph Partitioning Methods.}\label{fig:rf-parti}
\end{figure*}

% \smallskip
\noindent\textit{\textbf{Partitioning Quality Compared to Graph Partitioning Methods.}}
The partitioning quality is measured by the replication factor, as discussed in Def.~\ref{def:edgepartitioning} of Sec.~\ref{sec:graphedgepartitioning}.
The replication factor is the normalized number of the replicated vertices among partitions. 
The best score is $1.0$.
For the comparison of the partitioning quality with the vertex partitioning method (i.e., \textit{MTS}), we convert the vertex-partitioned graph into the edge-partitioned one as demonstrated in~\cite{Bourse:2014:BGE:2623330.2623660}, that is, each edge is randomly assigned to one of its adjacent vertices' partitions.
Our proposal is \textit{GEO+CEP}, where edges are ordered by \textit{GEO} in advance and partitioned by \textit{CEP}.

Figure~\ref{fig:rf-parti} shows the result.
Overall, \textit{GEO+CEP} delivers the second-best quality next to \textit{NE}, and these scores are similar.
\begin{revise-env}%
% The quality loss of \textit{GEO+CEP} from \textit{NE} is due to its flexibility for arbitrary $k$.
% The empirical result is consistent with the theoretical result as discussed in Section~\ref{sec:upperbound}.
\end{revise-env}%
The quality of \textit{GEO+CEP} is much better than hash-based methods, such as, \textit{BVC}, \textit{DBH}, \textit{1D}, and \textit{2D}.
Even compared to the high-quality vertex partitioning (i.e., \textit{MTS}), \textit{GEO+CEP} is always better except for \texttt{Road-CA}, whose graph structure is not so complicated that each result can be different. Its quality is almost $1.0$ in \textit{MTS}, \textit{NE}, and  \textit{GEO+CEP}.

% \medskip
% To summarize the above comparison with the existing dynamic scaling and graph partitioning methods, \textit{GEO+CEP} i

\subsection{Comparison with Graph Ordering} \label{sec:comp-orderin}
% We compare our methods with the existing ordering methods in Table~\ref{tbl:oalgorithms}.

\begin{figure*}[h]
 \centering
  \subfigure[\texttt{Road-CA}]{\includegraphics[width=.25\textwidth]{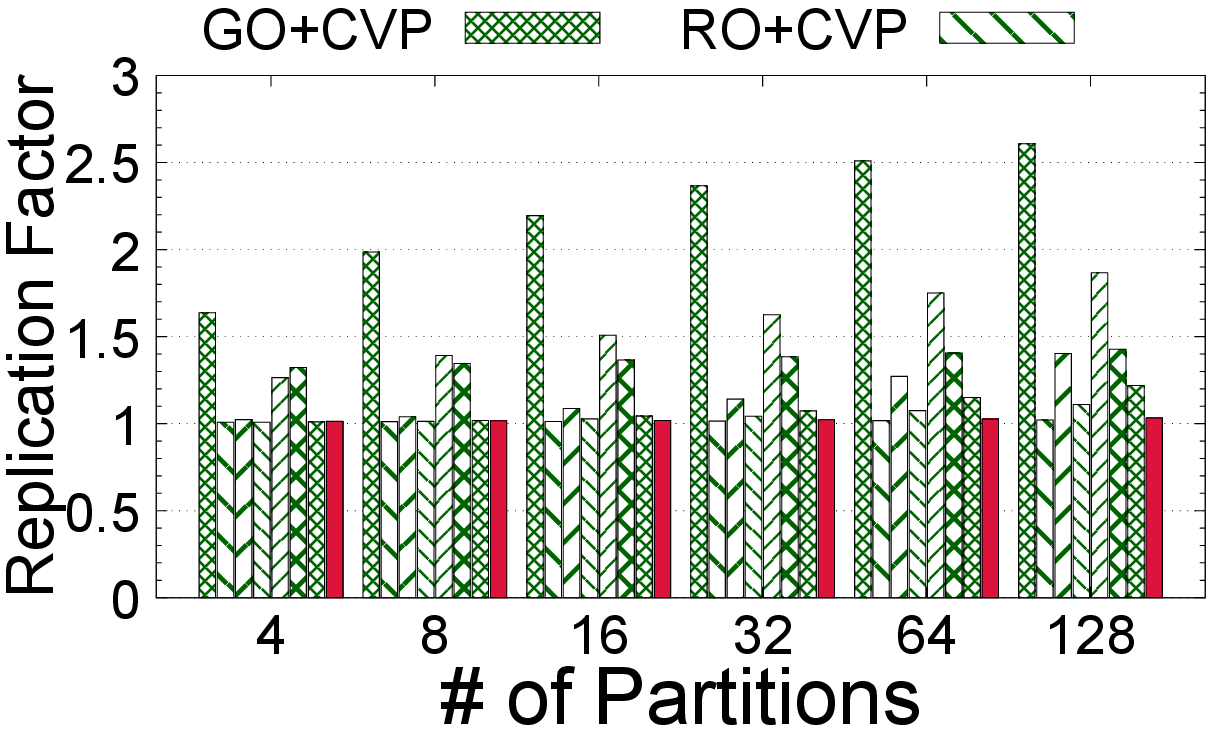}\label{fig:performance-pokec}}%
  \subfigure[\texttt{Skitter}]{\includegraphics[width=.25\textwidth]{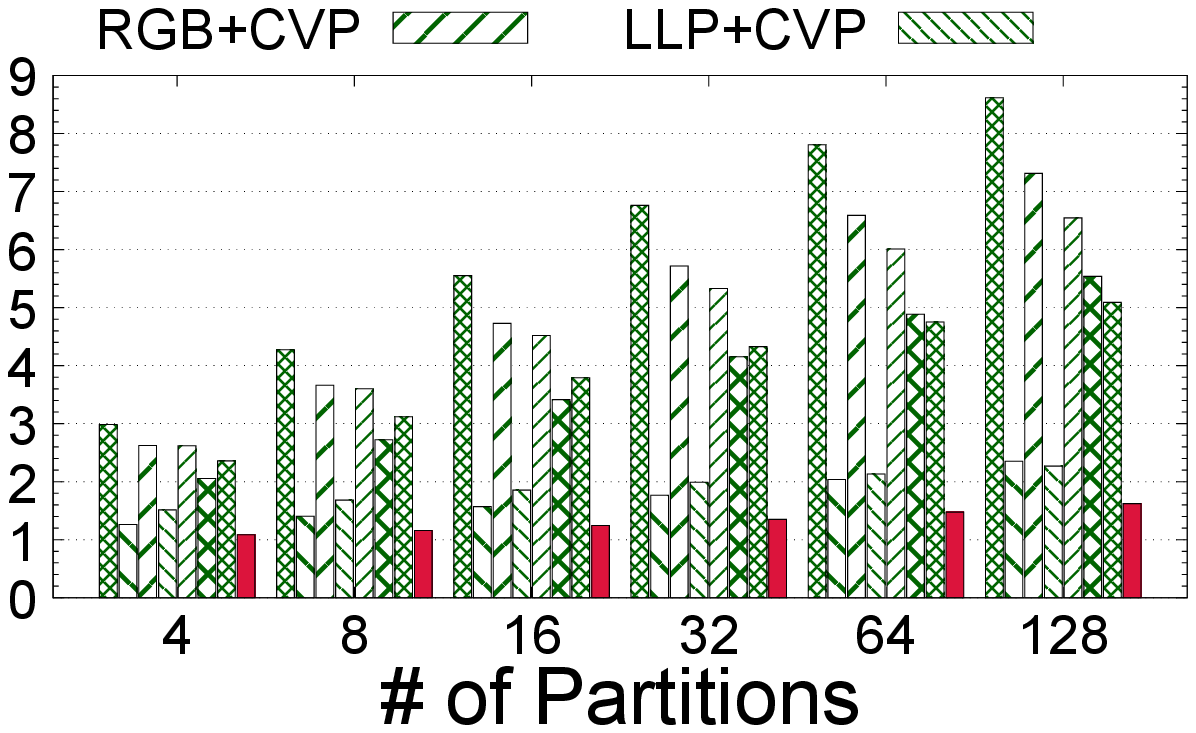}\label{fig:performance-pokec}}%
  \subfigure[\texttt{Patent}]{\includegraphics[width=.25\textwidth]{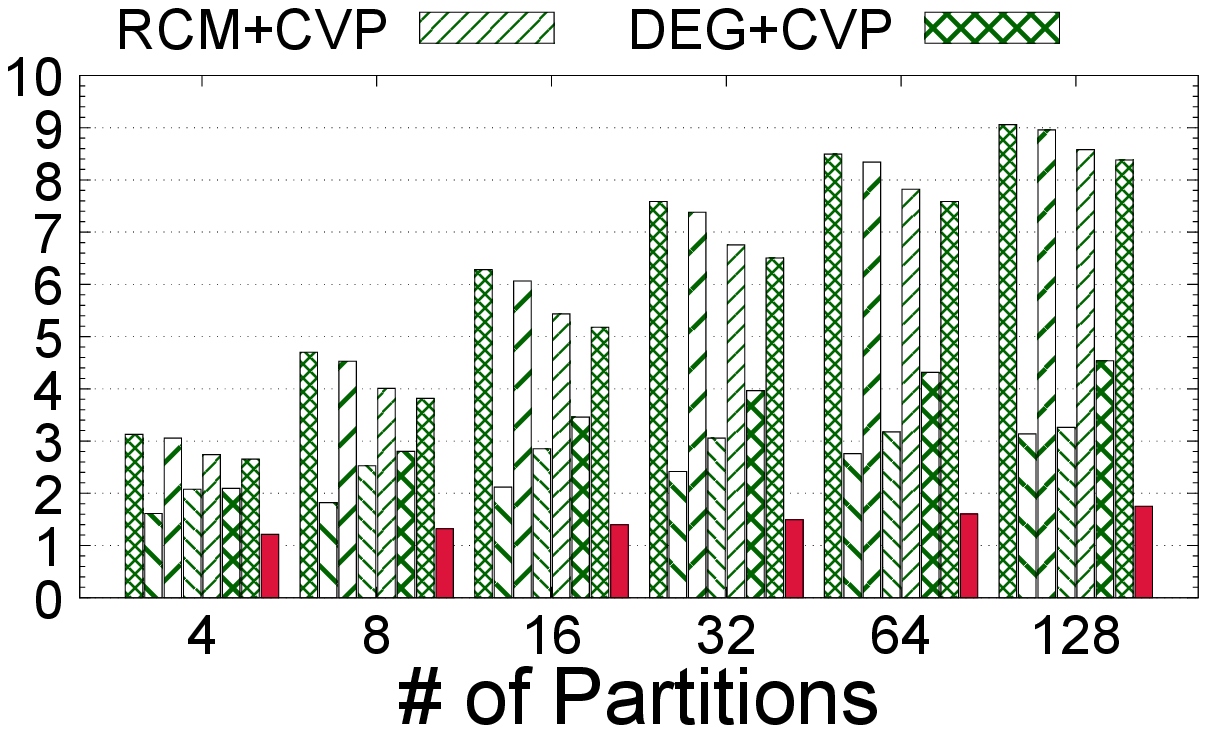}\label{fig:performance-pokec}}%
 \subfigure[\texttt{Pokec}]{\includegraphics[width=.25\textwidth]{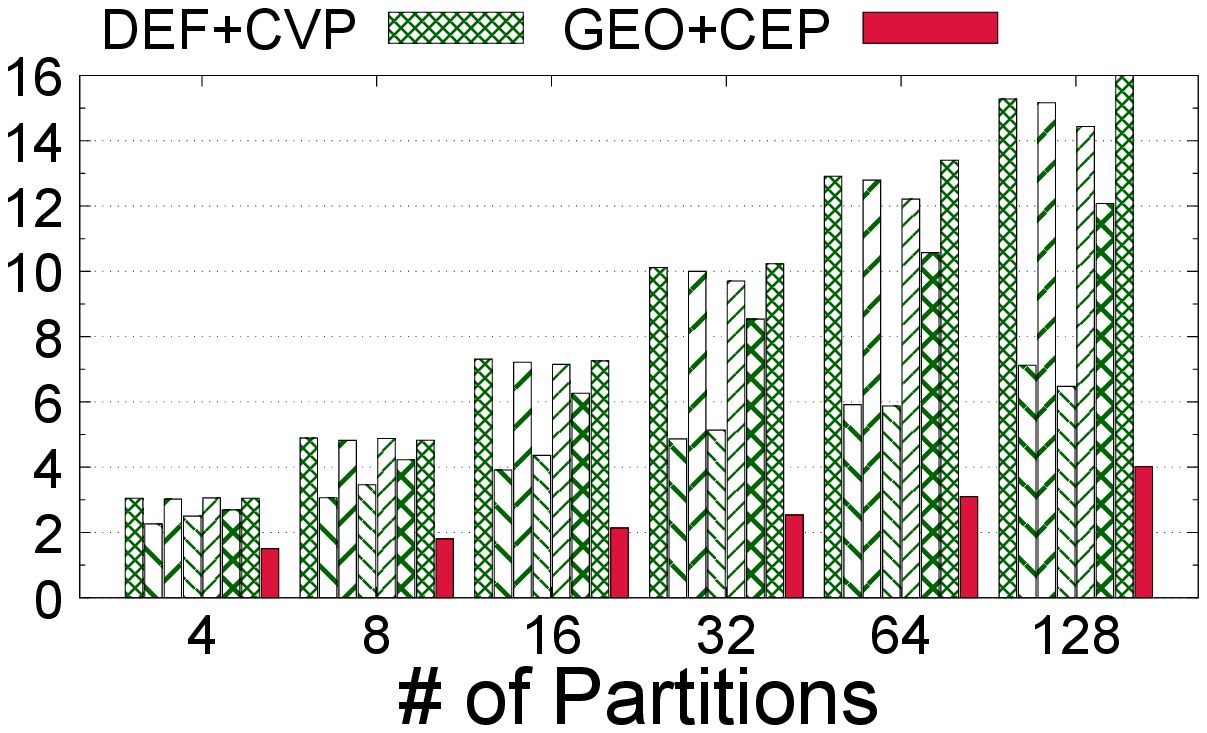}\label{fig:performance-pokec}}\\ \vspace{-5pt}
  \subfigure[\texttt{Flickr}]{\includegraphics[width=.20\textwidth]{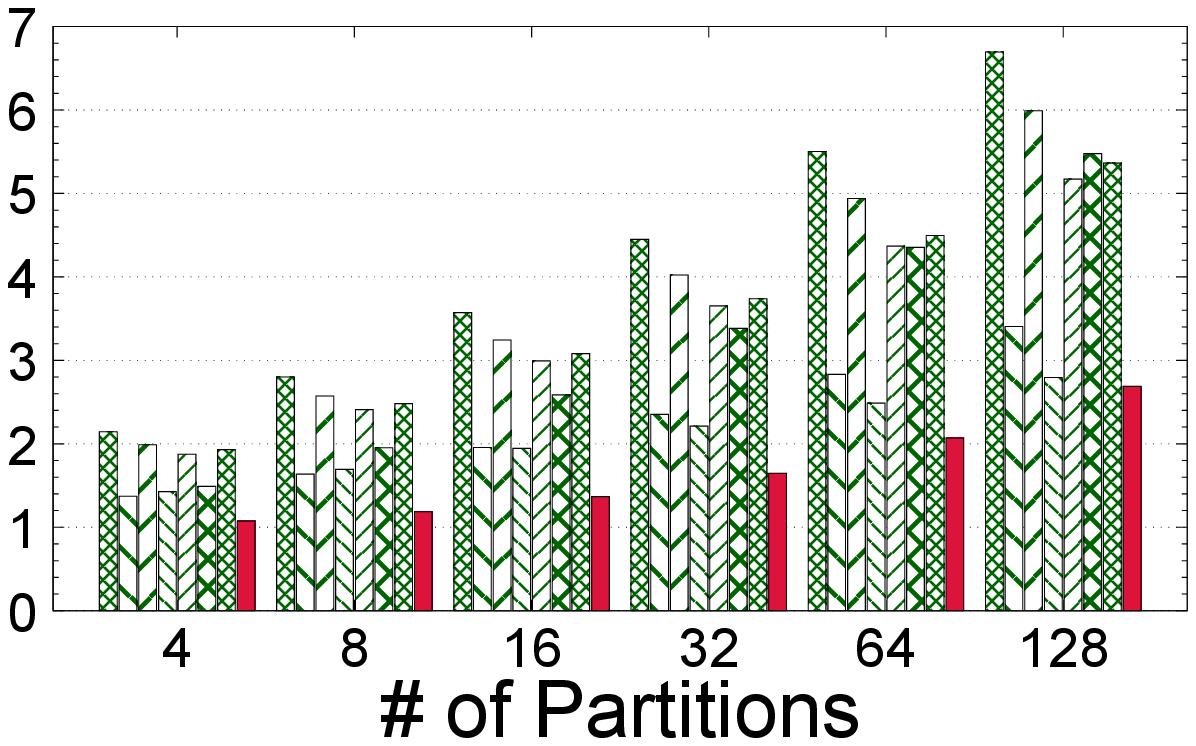}\label{fig:performance-pokec}}%
  \subfigure[\texttt{LiveJournal}]{\includegraphics[width=.20\textwidth]{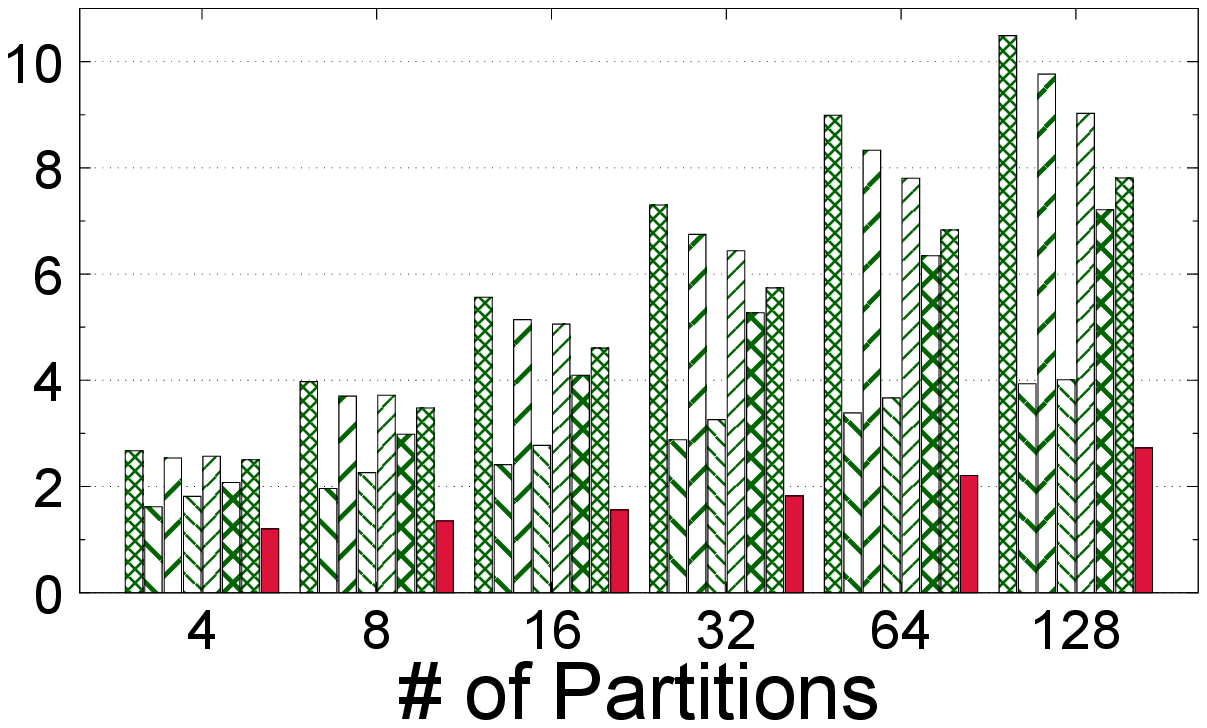}\label{fig:performance-pokec}}%
   \subfigure[\texttt{Orkut}]{\includegraphics[width=.20\textwidth]{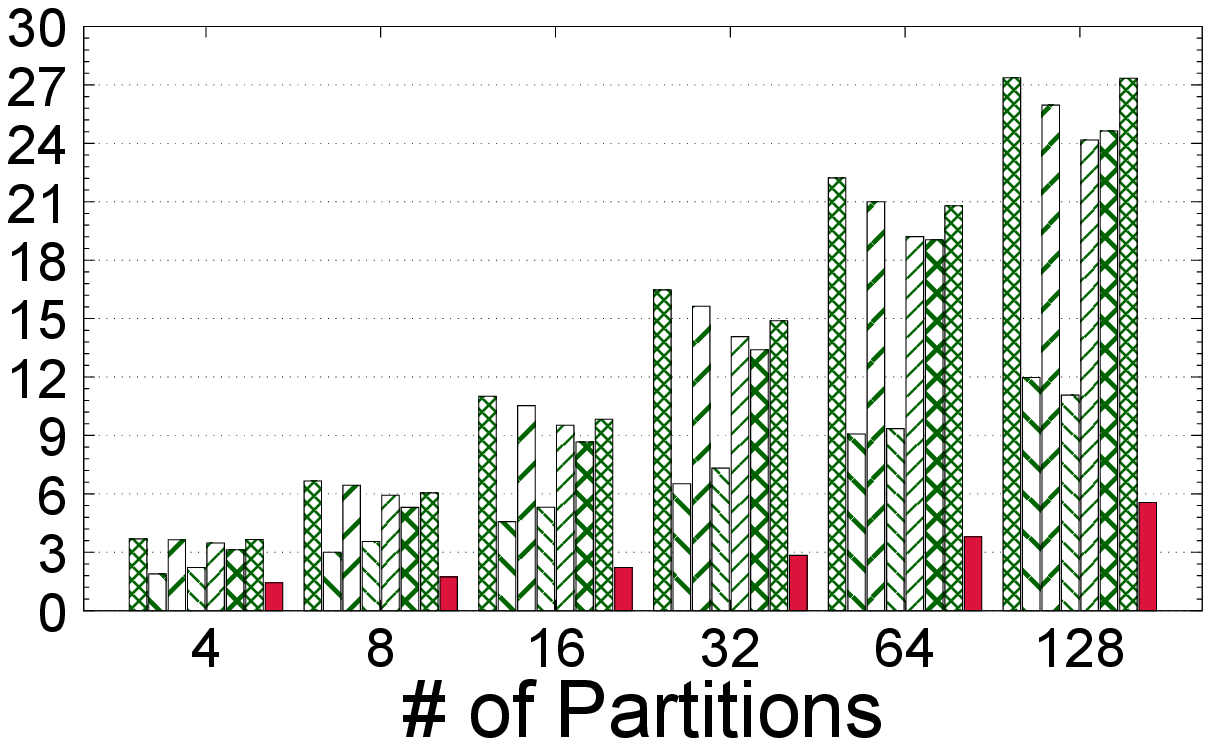}\label{fig:performance-pokec}}%
   \subfigure[\texttt{Twitter}]{\includegraphics[width=.20\textwidth]{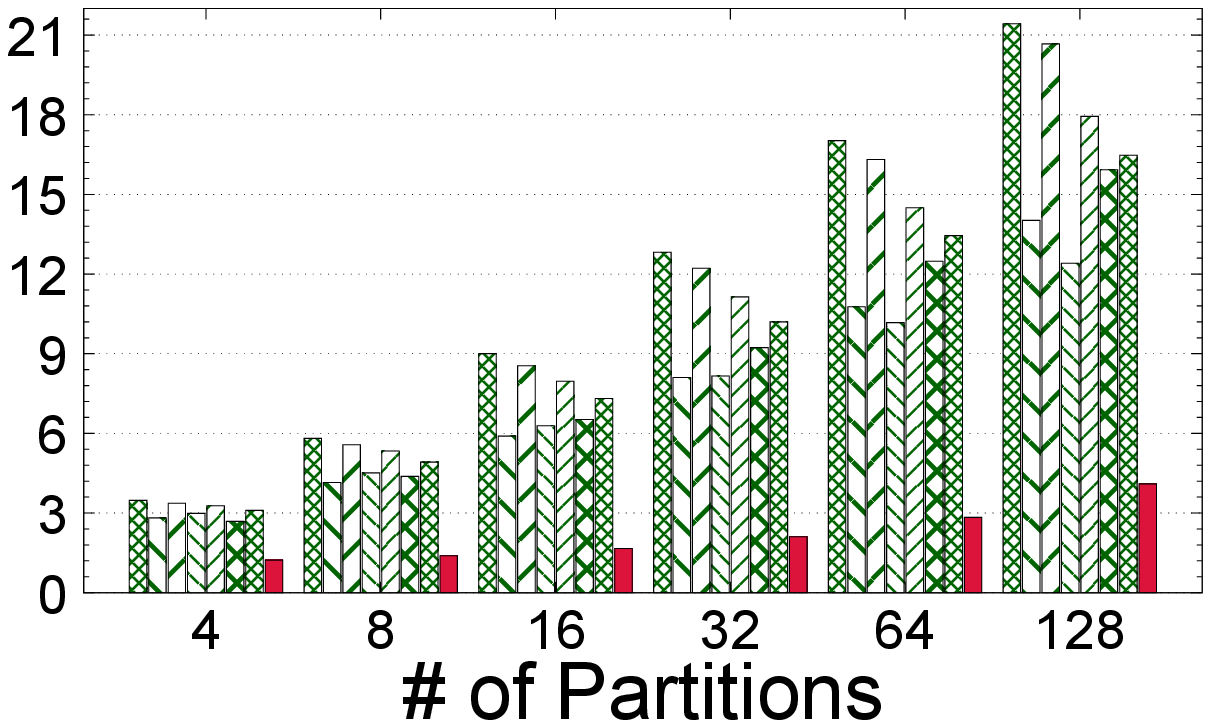}\label{fig:performance-pokec}}%
   \subfigure[\texttt{FriendSter}]{\includegraphics[width=.20\textwidth]{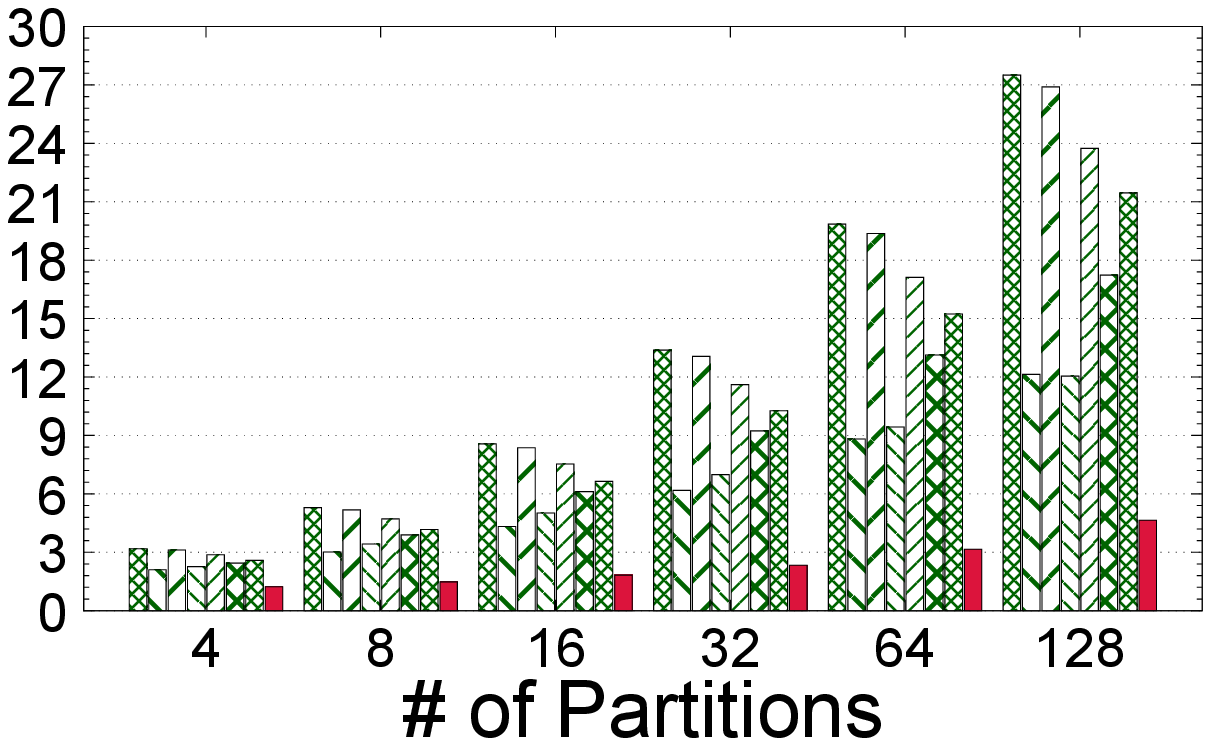}\label{fig:performance-pokec}}%
 \caption{Replication Factor Compared to Graph Ordering Methods.}\label{fig:rf-order}
\end{figure*}

% \smallskip
\noindent\textit{\textbf{Partitioning Quality Compared to Graph Ordering Methods.}}
Figure~\ref{fig:rf-order} shows the quality evaluation of the graph ordering methods.
All the existing methods are vertex ordering. Thus, we partition the ordered vertices via \textit{CVP} and generate vertex partitions.
For quality comparison, we convert vertex partitions into edge partitions in the same way as the previous subsection.

Overall, \textit{GEO+CEP} is always better than the other ordering methods.
Especially, the improvement is significant in \texttt{Orkut}, where the replication factor is totally high, meaning that, it is difficult to get good partitions.
\textit{RO} and \textit{LLP} become the similar quality to \textit{GEO+CEP} in \texttt{Road-CA} and \texttt{Flickr}. 
This is because these two methods capture `general' data locality (i.e., network modularity in \textit{RO} and community structure in \textit{LLP}) rather than to solve some problems highly specific to its purpose (i.e., \textit{GO} is for the L1-cache utilization; and \textit{RGB} is for graph compression).
In \texttt{Road-CA} and \texttt{Flickr}, these general localities become similar to one derived from the graph edge ordering.

% \begin{revise-env}
The high quality of \textit{GEO+CEP} essentially comes from the design of the priority (Eq.~\eqref{eq:priority}) derived from the objective of the graph edge ordering problem (Eq.~\eqref{eq:ordering1} and Eq.~\eqref{eq:ordering2}). 
This is due to the fact that some of the existing ordering methods, such as \textit{RCM} and \textit{GO}, are based on BFS and an algorithm very close to ours. 
Our priority differentiates the partitioning quality of \textit{GEO+CEP} from that of the existing methods.
% \end{revise-env}

% \smallskip
\noindent\textit{\textbf{Preprocessing Time.}}
We compare the elapsed time of each ordering method.
Figure~\ref{fig:Perf-order} shows the result.
Although \textit{GEO} is not the best performance compared to the simple methods, such as \textit{RCM} and \textit{DEG}, its performance is similar to the other ordering methods such as \textit{GO}, \textit{RGB}, and \textit{LLP}.
The graph edge ordering can preprocess the billion-scale graphs (\texttt{Twitter} and \texttt{FriendSter}) within an acceptable time.

\begin{figure}[h]
  \centering
   \includegraphics[width=.8\columnwidth]{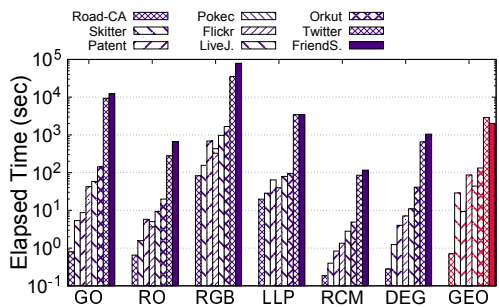} 
  \caption{Preprocessing Time for Graph Ordering}\label{fig:Perf-order}
\end{figure}

\subsection{Effect on Distributed Graph Analysis}

% We briefly evaluate the effect of our partitioning method on three common benchmarking graph applications with different workload characteristics: \textit{SSSP}, \textit{WCC}, and \textit{PageRank}.
We briefly evaluate the effect of our dynamic scaling method on three common benchmarking graph applications with different workload characteristics: \textit{SSSP}, \textit{WCC}, and \textit{PageRank}.
\textit{SSSP} is the lightest workload, starting from Vertex $0$ in this evaluation;
\textit{WCC} is the middle one; \textit{PageRank} is the heaviest one, where all vertices communicate with their neighbors at each iteration (the number of iterations is set to $100$).
We integrate our method to PowerLyra~\cite{Chen:2015:PDG:2741948.2741970} (forked from PowerGraph~\cite{joseph2012powergraph}) and compare it with four methods in the system: \textit{1D} (Random), \textit{2D} (Grid), \textit{Oblivious}, and \textit{Hybrid Ginger}.
% \begin{revise-env}%
% We omit a comparison with the time-consuming high-quality methods (such as \textit{NE}) since the evaluation of the partitioning time (Figure~\ref{fig:performance}) already shows a clear advantage of our method in dynamic scaling.
% Repeating such a method obviously imposes a serious performance overhead.
% Thus, we focus on the hash-based methods in this evaluation. 
% \end{revise-env}%
For a more comprehensive and detailed analysis of the effect of the partitioning quality on distributed graph applications, please refer to the previous experimental researches~\cite{Han:2014:ECP:2732977.2732980,6877273,Verma:2017:ECP:3055540.3055543,abbas2018streaming,Gill:2018:SPP:3297753.3316427,Pacaci:2019:EAS:3299869.3300076}.
The result of this evaluation is consistent with these researches.

% It should be noted that the comparative analysis illustrated in Table~\ref{} does not include the algorithms analyzed discussed in Section~\ref{} (summarized in Table~\ref{} and ~\ref{}). 
% This is due to the fact that there is no implementation of these algorithms available in PowerLyra.
We use two metrics: the elapsed time (\texttt{TIME}) and the communication volume (\texttt{COM}), as well as three metrics for the quality: the replication factor (\texttt{RF}), the edge balance (\texttt{EB}), and the vertex balance (\texttt{VB}).  
Specifically, let a balance factor among partitions ($p \in P$) be $B(\{x_p\}) := \tfrac{\max x_p}{\bar{x}}$, where $\bar{x} := \tfrac{\sum x_p}{|P|}$. 
Then, \texttt{EB} and \texttt{VB} are defined as $B(\{|E_p|\})$ and $B(\{|V(E_p)|\})$, respectively. 
Note that \texttt{EB} is the actual value of $1+\epsilon$ as difined in Def.~\ref{def:edgepartitioning}.

% For the elapsed time (\texttt{TIME}), we measure the time for applications and exclude setup time such as system preparation, data loading, data partitioning, and so forth. We execute five times and show the median value.

% \begin{revise-env}
We evaluate our proposed approach in two different ways:  (i) measuring the performance of applications and (ii) measuring the performance of the entire system including dynamic scaling.
% The first evaluation shows the impact of replication factor on the communication cost and its consequent application performance.
% The second shows the performance of the entire system including dynamic scaling.
% \end{revise-env}

\subsubsection{Application Performance}
Table~\ref{tbl:perf} shows the result on 36 partitions (one physical core per partition) without dynamic scaling by using the three large graphs (\texttt{Orkut}, \texttt{Twitter}, and \texttt{FriendSter}).
For the elapsed time (\texttt{TIME}), we measure the time only for applications and exclude setup time such as system preparation, data loading, data partitioning, and so forth. We execute five times and show the median value.

Overall, our method (\texttt{CEP+GEO}) outperforms the others in the elapsed time (\texttt{TIME}) due to the lowest replication factor (\texttt{RF}).
Its speed up from the others is the most significant in \textit{PageRank} due to the largest reduction of communication cost (\texttt{COM}).
Even though the vertex balance (\texttt{VB}) of our methods is slightly worse than that of the others, it does not play an important role for the elapsed time (\texttt{TIME}).
This is because the computational cost for the graph processing essentially depends on the number of edges rather than that of vertices, as already discussed in Sec.~\ref{sec:introduction}.
The edge balance is more dominant for the performance, and our method always achieves the perfect score (i.e., \texttt{EB} is 1). 

\begin{table*}[h]
\newcolumntype{B}{!{\vrule width 3\arrayrulewidth}}
\newcolumntype{G}{!{\vrule width 6\arrayrulewidth}}
\caption{Evaluation of Graph Applications on 36 Partitions. \texttt{TIME} unit is sec. \texttt{COM}  unit is GB.}
\centering
\label{tbl:perf}
\scalebox{1.0}{
  \begin{tabular}{|l|lGr|r|rGr|rBr|rBr|rB} \hline
  \multicolumn{2}{|cG}{} & \multicolumn{3}{cG}{Quality}  & \multicolumn{2}{cB}{\textit{SSSP}} & \multicolumn{2}{cB}{\textit{WCC}} & \multicolumn{2}{cB}{\textit{PageRank}} \\ \cline{3-11}
  \multicolumn{2}{|cG}{} & \texttt{RF} & \texttt{EB} & \texttt{VB} & \texttt{TIME} & \texttt{COM} & \texttt{TIME} & \texttt{COM} & \texttt{TIME} & \texttt{COM} \\ \hline
  \multirow{5}{*}{\rotatebox[origin=c]{90}{\texttt{Orkut}}} & \textit{1D} & 23.91 & \textbf{1.00} & \textbf{1.00} & 5.29 & 8.51 & 22.0 & 22.5 & 224 & 167 \\ 
                                                            & \textit{2D (Grid)}   & 9.76 & 1.01 & 1.01 & 3.93 & 4.30 & 13.67 & 9.5 & 130 & 69.7 \\ 
                                                            & \textit{Oblivious}   & 16.35 & 1.23 & 1.01 & 4.43 & 6.27 & 17.0 & 15.62 & 168 & 112 \\ 
                                                            & \textit{Hybrid Ginger} & 11.56 & 1.37 & 1.05 & 3.95 & 8.25 & 13.5 & 12.5 & 106 & 56.2 \\
                                                            & \textbf{\textit{GEO+CEP}} & \textbf{\textit{2.98}} & \textit{\textbf{1.00}} & \textit{1.32} & \textbf{\textit{2.89}} & \textbf{\textit{0.72}} & \textbf{\textit{8.20}} & \textbf{\textit{1.99}} & \textbf{\textit{66.6}} & \textbf{\textit{15.6}}\\ \hline
  \multirow{5}{*}{\rotatebox[origin=c]{90}{\texttt{Twitter}}} & \textit{1D}  & 14.11 & \textbf{1.00} & \textbf{1.00} & 47.5 & 74.1 & 136 & 126 & 2043 & 1262\\ 
                                                                 & \textit{2D (Grid)}    &  7.52 & 1.04 & \textbf{1.00} & 31.2 & 47.5 & 90.8 & 72.0 & 1239 & 647 \\ 
                                                                 & \textit{Oblivious}    & 11.04 & 1.05 & 1.01 & 38.4 & 61.4 & 108 & 100 & 1630 & 985 \\ 
                                                                 & \textit{Hybrid Ginger}   & 4.20 & 1.21 & 1.06 & 22.0 & 75.8 &  64.9 & 73.8 & 717 & 319 \\ 
                                                                 & \textbf{\textit{GEO+CEP}} & \textit{\textbf{2.20}} &  \textit{\textbf{1.00}} & \textit{2.92} & \textbf{\textit{17.6}} & \textbf{\textit{6.11}} & \textbf{\textit{47.6}} & \textit{\textbf{16.2}} & \textbf{\textit{518}} & \textit{\textbf{130}}\\ \hline
  \multirow{5}{*}{\rotatebox[origin=c]{90}{\texttt{FriendS.}}} & \textit{1D} & 14.46 & \textbf{1.00} & \textbf{1.00} & 81.7 & 112 & 389 & 297 & 3561 & 2160 \\ 
                                                          & \textit{2D (Grid)}    & 6.74  & \textbf{1.00} & \textbf{1.00} & 52.5 & 63.2 & 261 & 140 & 1985 & 983 \\ 
                                                          & \textit{Oblivious}    & 10.91  & \textbf{1.00} & \textbf{1.00} & 66.9 & 90.1 & 306 & 224 & 2609 & 1567 \\ 
                                                          & \textit{Hybrid Ginger}  & 7.28 & 1.14 & 1.10 & 51.2 & 117 & 241 & 181 & 1652 & 812 \\ 
                                                          & \textbf{\textit{GEO+CEP}} & \textit{\textbf{2.44}} & \textbf{\textit{1.00}} & \textit{3.04} & \textbf{\textit{39.7}} & \textbf{\textit{11.4}} & \textbf{\textit{169}} & \textbf{\textit{31.1}} & \textbf{\textit{963}} & \textbf{\textit{241}}\\ \hline
\end{tabular}
}
% \vspace{-10pt}
\end{table*}

% \begin{revise-env}
\subsubsection{End-to-end Performance}
We evaluate the entire performance of \textit{PageRank} (100 iterations) including the setup such as system initialization, graph (re)partitioning, data migration, and graph (re)construction.

% \smallskip
\noindent\textit{\textbf{Dynamic Scaling Scenario.}} 
We use two scenarios: \textit{ScaleOut} and \textit{ScaleIn}.
In \textit{ScaleOut}, a process is added each 10 iterations from 26 processes. Thus, the number of partitions is changed as follows: $26 \rightarrow 27 \rightarrow ... \rightarrow 36$.
In \textit{ScaleIn}, a process is removed each 10 iterations from 36 processes. Thus, the number of partitions is changed as follows: $36 \rightarrow 35 \rightarrow ... \rightarrow 26$.

% \smallskip
\noindent\textit{\textbf{Result.}} 
We show the total elapsed time (\texttt{ALL}) and the breakdown of its three constituent components (\texttt{INIT}, \texttt{APP}, and \texttt{SCALE}).
\texttt{INIT} is the initialization time including system setup, data loading, initial partitioning and graph construction.
\texttt{APP} is the application time for PageRank computation.
\texttt{SCALE} includes the repartitioning, data migration (structural data and intermediate values), and graph reconstruction. 
% In PowerLyra, these three components in \texttt{SCALE} are highly coupled so that we cannot precisely measure each of them separately.

As shown in Table~\ref{tbl:perf-end-to-end}, our method significantly outperforms the others in \texttt{ALL} due to the large performance improvement not only in \texttt{APP} but also in \texttt{INIT} and \texttt{SCALE}.
In \texttt{INIT}, the improvement mainly comes from the efficient partitioning and data loading from the file system.
In our method, the partitioning can be computed by directly loading from the file system without any data shuffling among the distributed processes.
Whereas, in the other methods, the partition of each edge of a graph needs to be processed one-by-one after data loading. 
% Then, each edge is shuffled among the distributed processes. 
In \texttt{SCALE}, the improvement is mainly due to the efficient repartitioning as discussed in Theorem~\ref{thr:cep}.
% Since the repartitioning of our method is highly efficient, the computational cost for \texttt{SCALE} is mainly data migration and graph reconstruction.
% Our method is advantageous also in data migration as discussed in Theorem~\ref{thr:migration} and Section~\ref{sec:ev-migration} later.

\begin{table*}[h]
\newcolumntype{B}{!{\vrule width 3\arrayrulewidth}}
\newcolumntype{G}{!{\vrule width 6\arrayrulewidth}}
\caption{Overall Time (\texttt{ALL}) and its Breakdown (\texttt{INIT}, \texttt{APP}, \texttt{SCALE}) for \textit{PageRank} with Dynamic Scaling (sec.).}
\label{tbl:perf-end-to-end}
\centering
\scalebox{1.0}{
  \begin{tabular}{|l|lGrBrrrGrBrrr|} \hline
  \multicolumn{2}{|cG}{} & \multicolumn{4}{cG}{\textit{ScaleOut}}  & \multicolumn{4}{c|}{\textit{ScaleIn}} \\ \cline{3-10}
  \multicolumn{2}{|cG}{} & \texttt{ALL} & \texttt{INIT} & \texttt{APP} & \texttt{SCALE} & \texttt{ALL} & \texttt{INIT} & \texttt{APP} & \texttt{SCALE} \\ \hline
  \multirow{4}{*}{\rotatebox[origin=c]{90}{\texttt{Orkut}}} & \textit{1D} & 301 & 6.8 & 220.2 & 72.9 & 298 & 8.0 & 216.8 & 72.8 \\ 
                                                            & \textit{Oblivious} & 282 & 7.8 & 184.0 & 89.4 & 279 & 7.7 & 181.9 & 88.7\\ 
                                                            & \textit{Hybrid Ginger} & 205 & 9.0 & 105.2 & 90.4 & 210 & 9.6 & 106.3 & 93.5 \\
                                                            & \textbf{\textit{GEO+CEP}} & \textbf{\textit{96}} & \textit{\textbf{2.4}} & \textbf{\textit{71.5}} & \textbf{\textit{21.7}} & \textbf{\textit{98}} & \textbf{\textit{4.8}} & \textbf{\textit{70.7}} & \textbf{\textit{22.1}}\\ \hline
  \multirow{4}{*}{\rotatebox[origin=c]{90}{\texttt{Twitter}}} & \textit{1D}  & 2893 & 75 & 2042 & 769 & 2843 & 86 & 1979 & 771 \\ 
                                                                 & \textit{Oblivious} & 2803 & 95 & 1673 & 1030 & 2767 & 91 & 1643 & 1029 \\ 
                                                                 & \textit{Hybrid Ginger} & 1673 & 114 &602 &955 &  1853 &  290 & 603 & 958  \\ 
                                                                 & \textbf{\textit{GEO+CEP}} & \textit{\textbf{837}} &  \textit{\textbf{37}} & \textit{\textbf{541}} & \textbf{\textit{257}} & \textbf{\textit{851}} & \textbf{\textit{54}} & \textit{\textbf{532}} & \textbf{\textit{264}} \\ \hline
  \multirow{4}{*}{\rotatebox[origin=c]{90}{\texttt{FriendS.}}} & \textit{1D} & 4937 & 117 & 3581 & 1228 & 4974 & 123 & 3569 & 1274 \\ 
                                                          & \textit{Oblivious}    & 4607 & 126 & 2990 & 1482 & 4576 & 146 & 2934 & 1488 \\ 
                                                          & \textit{Hybrid Ginger}  & 3700 & 198 & 1583 & 1915 & 3684 & 199 & 1562 & 1917\\ 
                                                          & \textbf{\textit{GE0+CEP}} & \textbf{\textit{1512}} & \textbf{\textit{56}} & \textbf{\textit{1035}} & \textbf{\textit{418}} & \textit{\textbf{1487}} & \textbf{\textit{49}} & \textbf{\textit{1007}} & \textbf{\textit{429}}  \\ \hline
\end{tabular}
}
\end{table*}

\subsubsection{Additional Experiment}
% \subsubsection{Migration Cost} 
\noindent\textit{\textbf{Migration Cost.}}
We evaluate the migration cost in dynamic scaling (\textit{ScaleOut} and \textit{ScaleIn} in the previous section).
We use three methods for the comparison: \textit{BVC}, \textit{1D}, and \textit{CEP}.
\textit{BVC} is designed for the efficient migration as its objective is defined as the minimization of the migration cost.
\textit{1D} is a representative of the other partitioning methods that do not take the migration cost into account. Each partitioned edge may basically move to any of the other partitions.

Figure~\ref{fig:migration} shows the number of migrated edges in the two scenarios.
\textit{BVC} and \textit{CEP} are almost the same number, outperforming \textit{1D}. 
This is due to the fact that the migration methods in \textit{BVC} and \textit{CEP} are very similar, where their difference is to align the edges to the ordering id space (\textit{CEP}) or to the hash ring in consistent hashing (\textit{BVC}).
In both methods, edges are split into the continuous chunks, and thus, the number of migrated edges is almost the same.
% The empirical results are consistent to Theorem~\ref{thr:migration}.

\begin{figure}
    \centering
    \includegraphics[width=.8\columnwidth]{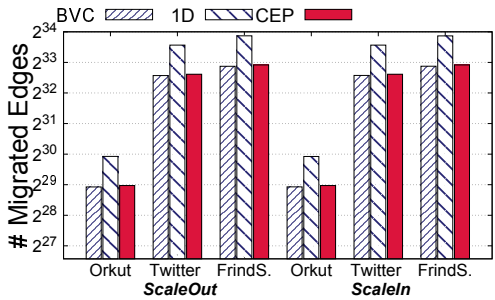}
    \caption{Total \# of Migrated Edges in \textit{ScaleOut} and \textit{ScaleIn}.}\label{fig:migration}
\end{figure}

\begin{figure}
    \centering
    \includegraphics[width=.8\columnwidth]{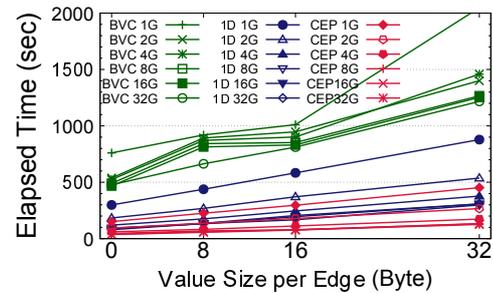}
    \caption{Migration Time for \textit{ScaleOut} with \texttt{FriendS.}.}\label{fig:perf-migration}
\end{figure}

% \begin{figure*}[h]
%   \vspace{-10pt}
%   \centering
%   \begin{minipage}{.3\textwidth}
%     \centering
%     \includegraphics[width=\columnwidth]{Migration.eps}
%     \vspace{-18pt}
%     \caption{Total \# of Migrated Edges in \textit{ScaleOut} and \textit{ScaleIn}.}\label{fig:migration}
%   \end{minipage}%
%   \hfill
%   \begin{minipage}{.3\textwidth}
%     \centering
%     \includegraphics[width=\columnwidth]{Perf-migration.eps}
%     \vspace{-18pt}
%     \caption{Migration Time for \textit{ScaleOut} with \texttt{FriendS.}.}\label{fig:perf-migration}
%   \end{minipage}%
%   \hfill
%   \begin{minipage}{.3\textwidth}
%     \centering
%     \includegraphics[width=\columnwidth]{Scale.eps}
%     \vspace{-18pt}
%     \caption{Scalability of \textit{GEO} with RMAT Graphs.}\label{fig:scalability}
%   \end{minipage}
%   \vspace{-12pt}
% \end{figure*}

Figure~\ref{fig:perf-migration} shows the actual elapsed time to migrate the edges and their values under the different network performances and sizes of each edge value.
We emulate the different network bandwidth from 1Gbps to 32Gbps according to the instance specifications in Amazon EC2~\cite{instancetype}.
The size of value per edge is changed from 0 to 32 bytes.  

In contrast to the number of migrated edges, \textit{CEP} and \textit{1D} outperform \textit{BVC}.
This is because, in \textit{BVC}, the edges are communicated in two phases: the initial migration and refinement for balancing edges. 
The refinement includes a lot of barrier synchronizations to share the edge balanceness among the distributed processes, especially in small $\epsilon$ and $k$.
\textit{BVC} is considered to be more appropriate for larger $\epsilon$ and $k$ as evaluated in \cite{dynamicscaling} (where $\epsilon$ is around 100 times bigger than our case and $k$ is over 100). 
On the other hand, in \textit{CEP} and \textit{1D}, the graph data are communicated in the single data shuffling and do not include the multiple global synchronizations.
% Moreover, \textit{CEP} is slightly better than \textit{1D} due to the reduction of the migrated edges.

An interesting insight from the evaluation is that the performance difference/improvement in data migration time is relatively small even though the number of migrated edges is largely different and the data migration itself is time-consuming (in some cases, it is slower than the partitioning time).
% Even for the efficient method, its speed up is very small and has a minimal impact onto the total workload as evaluated in Table~\ref{tbl:perf-end-to-end}. 
In contrast, the partitioning time as shown in Figure~\ref{fig:performance} exhibits a lot of variation in each of the methods examined, and thus its performance improvement may substantially influence the overall workload.
% This is the reason why we focus more on the efficiency of the partitioning time.

\smallskip
\noindent\textit{\textbf{Scalability.}}
Figure~\ref{fig:scalability} shows the scalability of \textit{GEO}.
We use RMAT, a common synthetic model for social networks~\cite{chakrabarti2004r}.
According to the real-world social networks in Table~\ref{tbl:real_world}, we change Edge Factor of RMAT (i.e., average degree) from 16 to 40 and the graph size up to 10 billion-edge scale.
Overall, the performance changes linearly as the increase of the graph size.
However, \textit{GEO} as well as its other counterparts (i.e., high-quality graph partitioning and graph ordering methods) have a scalability limitation.
% \textit{GEO} inherently has a general limit of preprocessing methods, which also appears in the other counterparts such as high-quality graph partitioning and graph ordering.
That is, if the preprocessing time is very large (e.g., due to the large graph size), whereas the actual analysis time is relatively small (e.g., due to the high parallelization), then the benefit by the preprocessing cannot be amortized.
Such a limitation gives us the motivation to devise parallel and distributed algorithms to speed up \textit{GEO}.
This is listed as our future work in Sec.~\ref{sec:conclusion}.

\begin{figure}
    \centering
    \includegraphics[width=.8\columnwidth]{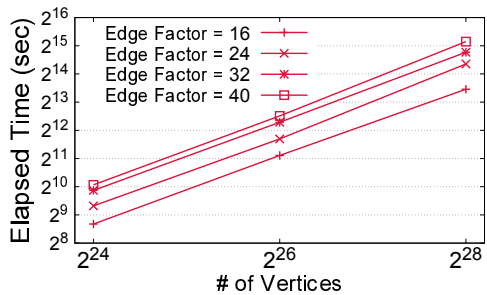}
    \caption{Scalability of \textit{GEO} with RMAT Graphs.}\label{fig:scalability}
\end{figure}

\section{Conclusion and Future Work}\label{sec:conclusion}
In this paper, we presented a novel approach to the dynamic scaling of graph partitions.
Our idea is based on the graph edge ordering and the chunk-based edge partitioning.
The former is the preprocessing method to provide high-quality partitions.
The latter is the very fast $O(1)$ partitioning algorithm.
We show that the maximization of the partitioning quality via graph edge ordering is NP-hard.
We proposed an efficient greedy algorithm to solve the problem within an acceptable time for large real-world graphs.
As a result, once the preprocessing is done, our dynamic scaling method is between three to eight orders of magnitude faster than the other existing methods while achieving high partitioning quality, which is similar to the best existing method.

There are mainly four future directions for our work.
First, the graph edge ordering needs to support the dynamic change of graph structures.
The requirement to reconfigure the number of partitions and recompute the graph analysis is higher for such dynamic graphs.
Second, a parallel and distributed algorithm of the graph edge ordering will be investigated.
The current sequential algorithm cannot handle extremely large graphs, such as trillion-edge graphs.
Third, the application to more complicated and time-consuming distributed graph processing, such as graph-based machine learning, is a very interesting and attractive problem.
Finally, the extension to more complicated graphs, such as, weighted-vertex/edge graphs, hyper graphs, property graphs, temporal graphs, will be investigated. 

% \textbf{The graph edge ordering and the chunk-based edge partitioning are publicly available in \url{http://anonymous}} 
% \todo{Limit 12 pages}
\balance
\bibliographystyle{unsrt}
\bibliography{ref}
\balance
\end{document}